\documentclass[11pt]{article}
\usepackage[english]{babel}

\usepackage[letterpaper,margin=1in]{geometry}
\usepackage[utf8]{inputenc}

\usepackage{mathtools}

\usepackage[backend=biber, style=alphabetic,  backref, maxcitenames=3,
  maxbibnames=99  ]{biblatex}
\usepackage{amsmath,amsthm,amssymb}
\usepackage{bbm}
\usepackage{graphicx}
\usepackage{abstract}
\usepackage{titlesec}
\usepackage{titletoc}
\usepackage{tikz-cd}
\usepackage{physics}

\usepackage[colorlinks=true, allcolors=blue]{hyperref}

\newcommand{\comments}[1]{}

\DeclareMathOperator{\im}{im}

\newcommand{\CZ}{\mathrm{C}Z}

\newcommand{\CCZ}{\mathrm{C}\mathrm{C}Z}
\DeclareMathOperator{\supp}{supp}
\DeclareMathOperator{\C}{\mathcal{C}}
\DeclareMathOperator{\type}{\text{type}}
\DeclareMathOperator{\1}{\textbf{1}}

\def\F{\mathcal{F}}

\def\t{\tilde}
\def\id{\mathrm{id}}
\newcommand{\te}[1]{\emph{#1}}

\numberwithin{equation}{section} 
\theoremstyle{definition}
\newtheorem{definition}{Definition}[section]
\newtheorem*{example}{Example}

\newtheorem{remark}{Remark}
\newtheorem{claim}[definition]{Claim}

\theoremstyle{plain}
\newtheorem{theorem}[definition]{Theorem}
\newtheorem{proposition}[definition]{Proposition}
\newtheorem{lemma}[definition]{Lemma}
\newtheorem{corollary}[definition]{Corollary}

\newtheorem{conjecture}[definition]{Conjecture}

\title{Poincaré Duality and Multiplicative Structures on Quantum Codes}
\author{
    Yiming Li \\  \normalsize Tsinghua \& Harvard
	\and Zimu Li \\ \normalsize Tsinghua
	\and Zi-Wen Liu  \\  \normalsize Tsinghua
	\and Quynh T. Nguyen \\ \normalsize  Harvard
}

\date{}

\begin{document}

\maketitle

\begin{abstract}

Quantum LDPC codes have attracted intense interest due to their advantageous properties for realizing efficient fault-tolerant quantum computing. In particular, sheaf codes represent a novel framework that encompasses all known good qLDPC codes with profound underlying mathematics. In this work, we generalize Poincaré duality from manifolds to both classical and quantum codes defined via sheaf theory on $t$-dimensional cell complexes. Viewing important code properties including the encoding rate, code distance, local testability soundness, and efficient decoders as parameters of the underlying (co)chain complexes, we rigorously prove a duality relationship between the $i$-th chain and the $(t-i)$-th cochain of sheaf codes.

We further build multiplicative structures such as cup and cap products on sheaved chain complexes, inspired by the standard notions of multiplicative structures and Poincaré duality on manifolds. This immediately leads to an explicit isomorphism between (co)homology groups of sheaf codes via a cap product. As an application, we obtain transversal disjoint logical $\CZ$ gates with $k_{\CZ}=\Theta(n)$ on families of good qLDPC and almost-good quantum locally testable codes. Moreover, we provide multiple new methods to construct transversal circuits composed of $\CCZ$ gates as well as for higher order controlled-$Z$ that are provably logical operations on the code space. We conjecture that they generate nontrivial logical actions, pointing towards fault-tolerant non-Clifford gates on nearly optimal qLDPC sheaf codes.
Mathematically, our results are built on establishing the equivalence between sheaf cohomology in the derived-functor sense, Čech cohomology, and the cohomology of sheaf codes, thereby introducing new mathematical tools into quantum coding theory.
\end{abstract}

\newpage
\tableofcontents

\newpage
\section{Introduction}

\subsection{Background and main contributions}\label{sec:background}

As we strive for practically useful large-scale quantum computers, quantum error-correcting codes are a cornerstone, allowing quantum information to be stored and processed reliably in the presence of noise and ultimately enabling fault tolerance with sufficiently low overhead. Recently, there has been a surge of interest in quantum low-density parity-check codes (qLDPCs or qLDPC codes) and locally testable codes (qLTCs) (especially those with high encoding rates) for both practical and theoretical reasons, as they offer pathways toward efficient fault tolerance in the real world~\cite{gottesman2013fault,FawziGrospellierLeverrier2018_ConstOverheadFT,CohenEtAl2022_LowOverheadFTqLDPC,LeverrierZemor2023_DecodingQTanner,yamasaki2024time,GuEtAl2024_SingleShotGoodqLDPC,nguyen2025quantum,TamiyaKoashiYamasaki2024_PolylogConstSpaceqLDPC} as well as bearing profound connections with fundamental questions in computation theory such as the prominent quantum probabilistically checkable proofs (qPCP) conjecture~\cite{AharonovAradVidick2013_qPCPSurvey,AharonovEldar2015_QLTC,eldar2017local,anshu2023nlts,CrossHeNatarajanSzegedyZhu2024_QLTCConstantSoundness,Golowich2024_NLTSPlantedqLDPC,Vidick2023_QCodesLocalTestabilityIP, anshu2024circuit}.

Besides improving the information protection or memory capability of quantum codes (which is what the extensive work on improving the standard code parameters is about), the implementation of logical actions, ideally in a simple and fault-tolerant manner, represents another vital yet morally competing objective.
Recent years have seen intensive interest and effort devoted to understanding fault-tolerant gates on qLDPC codes, leading to results on both achievable constructions \cite{KrishnaPoulin2021_FTGatesHGP,QuintavalleWebsterVasmer2023_PartitionHGPClifford,breuckmann2022fold,PatraBarg2025_TargetedCliffordHGP,Breuckmann:2024cupandgate,zhu2023non,ScrubyPesahWebster2024_QuantumRainbowCodes,lin2024transversalnoncliffordgatesquantum,Golowich_Lin2024,zhu2025transversalnoncliffordgatesqldpc} and fundamental no-go theorems~\cite{BurtonBrowne2022_HGPTransversalLimits,FuZhengLiLiu2025_ProductCodesNoGo}.  However, the extent to which desirable code parameters and logical gates can be combined remains far from being well understood.

The construction and study of codes have drawn deeply on formalisms, insights and methods from algebraic topology. 
More concretely, CSS codes admit a natural algebraic correspondence to (co)chain complexes, enabling the usage of various powerful mathematical techniques to advance quantum coding theory. 
Notably, it has been recognized in recent years that sheaf theory provides a powerful machinery for systematically tracking locally defined data and gluing it into the global structures, giving rise to a rich framework of CSS codes known as sheaf codes. In particular, recent breakthroughs in the construction of good qLDPC, good classical LTC and almost good qLTC \cite{PK2022Good,DHLV2022,leverrier2022quantum,Dinur2024sheaf} are achieved within the sheaf code framework.

In the constructions of  \cite{PK2022Good,leverrier2022quantum}, the chain complexes associated with the codes are symmetric, so one only needs to establish one-sided code parameters. In general, however, the chain complex needs not be symmetric, as is the case in \cite{DHLV2022,Dinur2024sheaf}. In these constructions, methods from high-dimensional expanders are first used to establish one-sided (coboundary direction) parameters such as distance or soundness, and then the parameters on the other side (boundary direction) are related to the coboundary direction parameters of the same complex yet equipped with a different ``dual" sheaf. Roughly speaking, this relation involves the parameters on the $i$-th chain complex equipped with a sheaf and the parameters on the $(t-i)$-th cochain complex equipped with the dual sheaf, where $t$ is the dimension of the cubical complex. Additionally, in \cite{DHLV2022,nguyen2025quantum}, it is also shown that the decoders have a similar relationship. This prompts us to ask the question: is this kind of ``duality" a general phenomenon?

In algebraic topology, a fundamental result known as Poincar\'e duality states that for a $t$-dimensional oriented closed manifold, the $i$-th cohomology group is isomorphic to the $(t-i)$-th homology group, and the isomorphism is given by the cap product with the fundamental class, which is a homology class lying in $t$-th homology group. The ``duality" phenomenon discovered in sheaf codes on cubical complexes closely resembles the Poincaré duality for manifolds, or a more general version called Verdier duality. In a recent work \cite{lin2024transversalnoncliffordgatesquantum}, it is argued that a duality of homology groups for sheaf codes holds provided that the sheaf satisfies a property called local acyclicity, but an explicit isomorphism is missing. Meanwhile, if the associated local codes satisfy the multi-orthogonality condition \cite{bravyi2012magic,Paetznick_2013}, a kind of cup product can be defined to support transversal non‑Clifford gates. Inspired by these observations, we go further and establish a systematic theory of duality for sheaf codes as well as a wide variety of multiplicative structures that enable logical gates. We also define a topological basis on the combinatorial cell complex as in \cite{First2022,Panteleev2024}, which enables a comprehensive sheafification process to formalize the sheaf structures for quantum codes. More importantly, this fully rigorous framework illuminates our discovery of the equivalence between several important cohomologies and motivates the definition of our cup products, along with other results.

The key message of our code duality theory can be distilled as follows:

\begin{theorem}[Informal, see Theorem \ref{thm:Poincaré_duality}]
    Let $X$ be a $t$-dimensional cell complex equipped with locally acyclic sheaf $\F$, then there is a dual sheaf $\F^\perp$ such that
   for any location $0 \leq i \leq t$, the code rate, code distance, soundness and decoder properties of the quantum or classical codes associated with $C_i(X,\F)$ and $C^{t-i}(X,\F^\perp)$ are essentially equivalent. There is also an explicit isomorphism $H_i(X,\F) \cong H^{t-i}(X,\F^\perp)$ induced by cap product. 
\end{theorem}

{The phenomenon that sheaf codes constructed from $\mathcal{F}$ and $\mathcal{F}^\perp$ exhibit closely related parameters has been observed in prior works such as \cite{DHLV2022, Dinur2024sheaf, First2024}. However, these results were obtained only in explicit constructions, and the relationships they present are not always formulated in a clean or general manner. Our work generalizes these earlier observations into a systematic duality theory, which also yields new general insights for code and logical gate constructions.
}

The multiplicative structures (such as cup and cap products) and Poincar\'e duality are closely related. We also establish a systematic theory of cup and cap products for sheaf codes. It has already been recognized that the cup product can be used to construct logical multi-controlled-$Z$ gates  \cite{10.21468/SciPostPhys.14.4.065,Chen_2023,Wang_2024,Breuckmann:2024cupandgate,zhu2025transversalnoncliffordgatesqldpc,Golowich_Lin2024}. Here we can also use our cup and cap products to produce logical gates, leading to new results. Using our code Poincar\'e duality theorem, we show in Theorem \ref{thm: CZ on good qLDPC} and \ref{thm: CZ on almost good qLTC} that  
    there exist $[\![n,\Theta(n),\Theta(n)]\!]$ qLDPC codes with transversal disjoint logical $\CZ$ gates such that $k_{\CZ}=\Theta(n)$ (linearly many independent logical $\CZ$'s guaranteed by subrank; see formal definitions in  Section~\ref{sec:controlled-Z}), 
    and $[\![n,\Theta(n),\Theta(n/(\log\,n)^3)]\!]$ quantum LTCs with soundness $1/(\log\, n)^3$ and transversal disjoint logical $\CZ$ gates such that $k_{\CZ}=\Theta(n)$.
To our knowledge, these results are the first demonstration of (almost) good qLDPC and qLTC supporting transversal $\CZ$ gates on $\Theta(n)$ many logical qubits.

We are also able to construct multi-linear cohomological invariants on the almost good qLTC of \cite{Dinur2024sheaf}, which induce transversal logical multi-controlled-$Z$ gates as long as they are not trivial, i.e., $k_{\CCZ}>0$. While we cannot rigorously prove this at the moment, we provide evidence (see Section~\ref{sec:explicit gate construction}) that one of our invariants will indeed yield $k_{\CCZ}>0$ and present the problem as a conjecture.
\begin{conjecture}\label{main conjecture}
    Let $X$ be a $t$-dimensional cubical complex and $\F$ be a sheaf satisfying the requirement in~\cite{Dinur2024sheaf}. We conjecture that for $2\leq i,j,k\leq t$, $i+j+k\leq t$, there exist cohomology classes $\alpha\in H^i(X,\F)$, $\beta\in H^j(X,\F)$, $\gamma\in H^k(X,\F)$, such that $\alpha\smile\beta\smile\gamma\neq 0\in H^{i+j+k}(X,\F^{\otimes 3})$. Then there exist $[\![n,\Theta(n),\Theta(n/\emph{poly}\log\,n)]\!]$  quantum codes with soundness $1/\Theta(n/\emph{poly}\log\,n)$ (nearly good qLTC) that support nontrivial transversal logical $\CCZ$. 
\end{conjecture}

Note that one can always perform numerical experiments to compute the cup and cap products, thereby determining whether the resulting logical operation is trivial (i.e., identity) or not. While the computational efficiency of such experiments is presently unclear, they may yield concrete insights and merit further investigation.

We point out that our framework does not impose multi-orthogonality or any other extra properties on the local codes, different from existing methods for building (almost) good qLDPC codes \cite{kalachev2025maximallyextendableproductcodes}. This is because our cohomological invariants follow from a rigorous synthesis of cup and cap products on combinatorial cell complexes which is independent of the choice of local coefficients.

From a more mathematical perspective, another contribution of our work is that it establishes connections among three cohomology theories: the sheaf cohomology for sheaf codes, the sheaf cohomology as the right derived functor of global sections, and \v{C}ech cohomology. There are studies on the relation between the cohomology of sheaf codes and right derived functors in \cite{curry2014sheavescosheavesapplications,First2022,First2024}. However, the advanced techniques from the theory of right derived functors were not fully exploited and utilized in studying quantum codes before. In our work, we obtain the bounds for $k_{\CZ}$ in Theorems \ref{thm: CZ on good qLDPC} and \ref{thm: CZ on almost good qLTC} by employing these abstract yet powerful tools from sheaf theory, and we are not aware of any possible alternative proof. Moreover, our framework is established by  systematically switching perspectives among these three equivalent cohomology theories: for example, we establish the cup and cap products by translating sheaf codes into the theory of \v{C}ech cohomology, and provide a proof of duality of logical qubits using flabby resolutions of sheaves. These bridges enable us to import a broad range of techniques from the well-developed theory of sheaves in mathematics into the study of quantum codes.
Furthermore, we believe that our constructions are of independent mathematical interest  in algebraic topology and combinatorics.

\subsection {Proof overview}
\paragraph{Poincar\'e duality.} Given a sheaf $\F$ on a $t$-dimensional cell complex $X$, we are able to construct a series of sheaves $\F_i, 0\leq i\leq t$ which serves as an resolution of the dual sheaf $\F^\perp$ (Definition \ref{def:sheaf_dual_sheaf}). In Proposition \ref{prop:flabby_resolution}, we show that there is an exact sequence of sheaves
\begin{equation}
	\begin{tikzcd}
		0 \arrow[r] & \mathcal{F}^{\perp}\arrow[r]
		& \mathcal{F}_t \arrow[r] & \mathcal{F}_{t-1}\arrow[r] & \cdots \arrow[r] & \mathcal{F}_1 \arrow[r] & \mathcal{F}_0\arrow[r] & 0
	\end{tikzcd}.
\end{equation}
Formally, this is a flabby resolution (Definition \ref{def:flabby}), based on which we are able to prove that the following two (co)homology groups are isomorphic:
\begin{equation}
	H^i(X,\F^\perp) \cong H_{t-i}(X,\F)
\end{equation}
for any $0 \leq i \leq t$.

Given this resolution, it is natural to consider the double complex $K^{p,q} \coloneqq  C^p(X,\mathcal F_{t-q})$, as is standard in Čech cohomology (for example, see \cite[Section 20.25]{stacks-project} for more details). This double complex has already been used in \cite{Dinur2024sheaf}.
As a double complex, it admits two differentials or (co)boundary operators: the horizontal differential $d'$ (Eq.~\eqref{eq:d'_1}) and the vertical differential $d''$ (Eq.~\eqref{eq:d''}). If the based complex $X$ is well-behaved, e.g., any simplicial complex or cubical complex, then $d'$ should always be exact. If $\F$ is a locally acyclic sheaf (Definition \ref{def:acyclic_sheaf}), then $d''$ is also exact. That is, the local acyclicity of $\F$ endows the double complex $K^{p,q}$ with a certain symmetry. By explicit calculation (Proposition \ref{prop:d'_d''_exact}), it can be shown that the chain complex $C_\bullet(X,\F)$ is encoded in the vertical direction of the double complex, and $C^\bullet(X,\F^\perp)$ is encoded in the horizontal direction. The exactness of $d'$ and $d''$ acts as a bridge connecting the two directions. More explicitly, this allows us to apply the technique of diagram chasing from homological algebra, which produces the explicit isomorphism between $H^i(X,\F^\perp)$ and $H_{t-i}(X,\F)$. By studying the relation of block Hamming weights of elements in $K^{p,q}$ during diagram chasing, we are able to demonstrate that the code distances, decoders, and (co)cycle expansions of the codes on $C^i(X,\F^\perp)$ and $C_{t-i}(X,\F)$, respectively, are tightly related (Theorem \ref{thm:Poincaré_duality}). Actually, a similar method can be used to prove the original Poincar\'e duality with scalar field coefficients.

Moreover, after establishing multiplicative structures on sheaf codes in Section \ref{sec:multiplicative}, we find that the isomorphism $H^i(X,\F^\perp) \cong H_{t-i}(X,\F)$ can be expressed by cap product with a particular chain element, denoted by $[X]$ (Theorem \ref{thm:Cap induced poincare dual map}). This completes the description of Poincar\'e duality on sheaf codes.

\paragraph{Logical gates.} CSS codes can be represented by cochain complexes $C^\bullet$ over a field $\mathbb{F}$ (typically of characteristic $2$). Usually, these cochain complexes are obtained from certain topological objects such as manifolds, graphs or high-dimensional expanders. It is well known that the key to constructing a multi-controlled-$Z$ gate is to find a multi-linear function. For instance, to build a logical $\CCZ$, we need 
\begin{equation}
	f: C^i \times C^j \times C^k \longrightarrow \mathbb{F},
\end{equation}
such that $f$ is well-defined on cohomology classes; that is, if $\forall \alpha'\in C^{i-1},\beta\in \ker\delta^j, \gamma\in\ker\delta^k$, we require $f(\delta\alpha',\beta,\gamma)=0$, and similarly for coboundaries in the other arguments. Such maps are called cohomological invariants.

A well-studied multiplicative structure on cohomologies is the cup product. For example, we may have
\begin{equation}
	\smile:C^i\times C^j\longrightarrow C^{i+j},
\end{equation}
such that the Leibniz rule holds, i.e., for $\forall \alpha\in C^i,\beta\in C^j$,
\begin{equation}
	\delta(\alpha\smile\beta)=\delta\alpha\smile\beta+\alpha\smile\delta\beta.
\end{equation}
Now, given any cohomological invariant linear function $g: C^{i+j}\rightarrow\mathbb{F}$, let $h(\alpha,\beta)=g(\alpha\smile\beta)$ and direct calculations show that $h$ is a cohomological invariant. This argument can be generalized to higher orders by the associativity of cup products (Proposition \ref{prop:cup_assoc}). Therefore, the cup product helps to reduce the task of constructing logical gates to finding single-linear cohomological invariant functions.

Fortunately, the multiplicative structures on cohomology also offer a natural function for this purpose, namely the pairing between (co)homology groups (see Section \ref{subsec:cup_cap_pairing}):
\begin{equation}
	\langle-,-\rangle : C^i\times C_i \longrightarrow\mathbb{F}.
\end{equation}
Any cocycle $\alpha\in C^i$ or cycle $x\in C_i$ will induce such linear functions $\langle\alpha,-\rangle$, $\langle-,x\rangle$. More importantly, the cup product is often given by geometrical relations. For example, the cup product on a simplicial chain complex $C^\bullet(X,\mathbb{F})$ is given by
\begin{equation}
	(\alpha\smile\beta)([v_0,\cdots, v_{i+j}])=\alpha([v_0,\cdots, v_i])\beta([v_i,\cdots, v_{i+j}]),
\end{equation}
where $[v_0,\cdots, v_{i+j}]$ is an $(i+j)$-dimensional simplex in a simplicial complex $X$. Therefore, the cup product of two standard bases (representing $i$- and $j$-dimensional cells respectively) is nonzero if and only if they are adjacent. Note that qLDPC constructions the underlying complex is sparse, meaning that each cell is joined to only $O(1)$ other cells.
Therefore, the multi-linear cohomological invariants built by cup products always lead to constant-depth physical circuits, which are fault-tolerant and thus crucial for our quest for practical quantum computation. Another multiplicative structure, the cap product, can also be used to multi-linear cohomology invariants.

 We can construct logical multi-controlled-$Z$ gates on sheaf codes by extending these multiplicative structures to sheaf codes. We make the following observation: given a simplicial complex $X$ equipped with Alexandrov topology, the sheaf cochain complex and the cochain complex of \v{C}ech are equivalent. Since the cup product is well-studied in \v{C}ech theory, we can directly translate it to sheaf codes. For any two sheaves $\F$ and $\mathcal{G}$ on $X$, we define the cup product
\begin{equation}
	\smile: C^i(X,\F)\times C^j(X,\mathcal{G})\longrightarrow C^{i+j}(X,\F\otimes\mathcal{G}),
\end{equation}
which satisfies the Leibniz rule. 
In Section~\ref{subsec:cup_cap_pairing} we also define other kinds of products.

The qLDPC codes and qLTCs with best known parameters are constructed based on more general cell complexes, such as the cubical complexes in \cite{PK2022Good,leverrier2022quantum,Dinur2024sheaf}. To deal with this general case, we follow the approach of \cite{FreedmanHastings2021,Portnoy2023,lin2024transversalnoncliffordgatesquantum} to subdivide $X$ into a simplicial complex $\tilde{X}$. We then apply powerful tools from sheaf theory, such as pullback and pushforward of sheaves, to define a sheaf $\tilde{\F}$ on $\tilde{X}$ and prove that there is an isomorphism
\begin{equation}
	H^\bullet(X,\F)\cong H^\bullet(\t{X},\t{\F}).
\end{equation}
through some (co)chain maps (Eq.~\eqref{eq:S_sharp} and \eqref{eq:A_sharp}). With an explicit inverse of the isomorphism, we are able to define both cup and cap products on $X$ via those on $\tilde{X}$. When $X$, $\tilde{X}$ are sparse, the (co)chain maps are represented by sparse matrices, which {ensure that the induced gates are constant-depth}. 

Again, using techniques from sheaf theory, we can prove that $\t{\F}$ is locally acyclic when $\F$ is. Consequently, for a $t$-dimensional cell complex $X$ with locally acyclic sheaf $\F$, we are able to construct the following bilinear cohomological invariant:
\begin{equation}
	P:C^i(X,\mathcal{F}^\perp)\times C^{t-i}(X,\mathcal{F})\longrightarrow\mathbb{F},
\end{equation}
given by 
\begin{equation}
	P(\alpha,\beta)=\langle\alpha\smile\beta,[X]\rangle,
\end{equation}
where we call $[X]=\sum_{\tau\in X(t) }\tau\in C_t(X,\F^\perp\!\otimes \F)$ a generalized ``fundamental" class. It is defined in Theorem \ref{thm:Cap induced poincare dual map} and is a $t$-cycle. We prove that $P$ is a dual pairing, i.e., behaves like inner products on cohomology groups, and we may define a dual basis according to this pairing on cohomology groups. As an immediate application of the isomorphism $P$, we note that the sheaves in \cite{DHLV2022} and \cite{Dinur2024sheaf} are locally acyclic. The former is built on $2$-dimensional cubical complexes, and by setting $i=1$ we get Theorem \ref{thm: CZ on good qLDPC}. The latter is built on $4$-dimensional cubical complexes, and by setting $i=2$ we get Theorem \ref{thm: CZ on almost good qLTC}.

Furthermore, we define more kinds of cap products, which give rise to a variety of multilinear cohomological invariants in general sheaf codes. As mentioned at the end of Section \ref{sec:background}, it is not yet clear how to prove nontrivial bounds on $k_{\CCZ}$ (Definition \ref{def:invariant}) for these invariants, but we conjecture that under the same assumption as in \cite{Dinur2024sheaf}, one can confirm that some of these invariants do not induce logical identity. By planting all-ones vectors into the local codes, one may show $k_{\CCZ}>0$ at the price of preserving the original code parameters for only one side, see Theorem~\ref{thm:ccz1}. Therefore, additional techniques are expected for the analysis of cup and cap products in order to provide a more precise estimation on $k_{\CCZ}$.

\subsection{Future directions}
\noindent\textbf{Good quantum locally testable code.} We anticipate that our result will be useful for the major open problem of constructing triply good qLTCs. In the recent work \cite{Dinur2024sheaf}, almost-good qLTCs were constructed. However, the polylog factor may be fundamental within their framework. In this construction, one needs a set $G$ with at least four pairwise commuting permutation sets. Intuitively, abelian structures lead to bad expansion properties, which is roughly the reason why the polylog factor occurs in their construction. Our results provide a framework for constructing sheaf codes which potentially enable us to abandon this pairwise commutative requirement, and by our Poincar\'e duality, since the parameters are related, it may suffice to bound code parameters and construct decoders for only one side. \\

\noindent\textbf{Transversal logical $\CCZ$ in good qLDPC and qLTC.} In our framework, the problem of transversal logical $\CCZ$ reduces to the calculation of cup and cap products. In algebraic topology, such multiplicative structures are notoriously difficult to compute, and a substantial body of theory has been developed for this purpose. We expect to develop analogous computational techniques for sheaf codes, which enable the certification of nontrivial logical $\CCZ$ in our framework. 

Proving a nontrivial $k_{\CCZ}$ lower bound on the almost good qLTC~\cite{Dinur2024sheaf} (as we conjecture) will have significant implications for the overhead of quantum fault tolerance. In~\cite{nguyen2025quantum}, the authors rely on this code to achieve a constant-space and $\log^{1+o(1)}$-time overhead. The factor $o(1)$ in the exponent originates from the use of $\CCZ$ state distillation subroutine in their fault tolerance scheme. If the almost-good qLTC has transversal gate with $k_{\CCZ} > 0$, then this subroutine can be omitted, thereby reducing the $\log^{o(1)}$ factor to $\mathrm{polyloglog}$.

\subsection{Organization}

In Section~\ref{sec:preliminaries}, we review the necessary background in topology (cell complexes, Alexandrov topology, sheaves) and the basic notions of quantum codes and logical gates. In Section~\ref{sec:duality}, we establish our Poincaré duality for sheaf codes. In Section~\ref{sec:multiplicative}, we develop cup and cap products for sheaf codes and apply them to construct fault-tolerant logical multi-controlled-$Z$ gates. Finally, in Appendix~\ref{appendix plant all-ones vector}, we generalize the technique of planting all-ones vectors into local codes in \cite{Golowich2024_NLTSPlantedqLDPC}, which we expect to be useful for bounding \(k_{\CCZ}\) for our construction.

\subsection{Acknowledgements}
We thank Anurag Anshu for valuable discussion and support. YL is supported by Tsinghua Xuetang Talents Program. ZL and ZWL are supported in part by a startup funding from YMSC, Dushi Program, and NSFC under Grant No.~12475023.

\section{Preliminaries}\label{sec:preliminaries}

Algebraic topology has been deeply involved in the study of quantum codes. Notably, sheaf theory underpins the recent construction of good classical LTC, qLDPC and almost good qLTC~\cite{PK2022Good,DHLV2022,leverrier2022quantum,Dinur2024sheaf}. In these works, sheaves on posets are used as a generalized coefficient in (co)chain complexes. However, in mathematics,  sheaves and sheaf cohomology are typically not formulated in this way. One may naturally ask whether the ``sheaf cohomology" in the context of quantum code is the same as that in mathematics defined by the right derived functor of global sections. Fortunately, this is a well-studied problem, and one can find a systematic exposition in \cite{curry2014sheavescosheavesapplications,First2022,First2024}. 

By equipping a poset $X$ with Alexandrov topology, we will show the equivalence between the cohomology theory of sheaf codes, the right derived functor of global sections, and the \v{C}ech cohomology. The interesting point is that by changing perspectives from different definitions, we are able to obtain an abundance of new mathematical objects with more properties that are not discovered before. These abstract tools have important applications: we offer a new theory of ``Poincar\'e duality" on sheaf codes first proposed in \cite{lin2024transversalnoncliffordgatesquantum} by using flabby resolutions of sheaves; we define cup and cap products by translating into the language of \v{C}ech cohomology, which leads to the construction of logical multi-controlled-$Z$ gate; we prove that the pullback sheaf of a locally acyclic sheaf is still locally acyclic, which plays an important role in working out the isomorphism of Poincar\'e duality using tools from the right derived functor definition.
These methods and results are newly introduced to the context of quantum codes and are essential to our study. They build on deep insights on the connections among various areas in sheaf cohomology and are of independent interest in algebraic topology and combinatorics.

Unless otherwise stated, we assume throughout the paper that the linear spaces are over a fixed finite field $\mathbb{F}_q$ with characteristic $2$ and are finite dimensional. Sometimes we may write $\mathbb{F}_q = \mathbb{F}$ for simplicity.

\subsection{Cell complexes and sheaves}
We first give some standard definitions on cell complexes and sheaves. Readers may refer to e.g., \cite{Bott:1982xhp,Hatcheralgtop} for a more comprehensive introduction and discussion of these concepts. 

\begin{definition}[Cell complex]
	A \te{finite cell complex} (or finite CW complex) is a topological space $X$ constructed inductively as follows:
	
	\begin{enumerate}
		\item Start with a discrete set $X^0$, whose points are regarded as $0$–cells.
		
		\item Inductively form the \emph{$n$–skeleton} $X^n$ from $X^{n-1}$ by attaching $n$–cells $e_\alpha^n$
		via maps
		\begin{equation}
		\varphi_\alpha : S^{n-1} \longrightarrow X^{n-1},
		\end{equation}
        where $S^{n-1} = \partial D^n$ is the boundary of an $n$-dimensional disk (a sphere).
		This means that $X^n$ is the quotient space of the disjoint union with $n$-dimensional disks $D_\alpha^n$
		\begin{equation}
		X^{n-1} \;\sqcup\; \bigsqcup_\alpha D_\alpha^n
		\end{equation}
		under the identifications $x \sim \varphi_\alpha(x)$ for $x \in \partial D_\alpha^n$.
		The cell $e_\alpha^n$ is the homeomorphic image of
		$D_\alpha^n - \partial D_\alpha^n$ under the quotient map.
		Therefore, as a set,
		\begin{equation}
		X^n = X^{n-1} \cup_\alpha e_\alpha^n,
		\end{equation}
		where each $e_\alpha^n$ is an open $n$–disk.
		
		\item Terminate this inductive process at a finite stage,
		setting $X = X^n$ for some $n < \infty$.
	\end{enumerate}
\end{definition}

Throughout the paper, we only care about finite cell complexes and will omit  ``finite" henceforth. We further require that the closure of each cell is compact to avoid unnecessary technicalities. For our coding-theoretic applications, we will not emphasize the point-set topological aspects of cell complexes. Instead, we view a cell complex primarily as a combinatorial object and focus on the incidence relations between cells. We therefore introduce the following definition capturing this combinatorial structure.

\begin{definition}[Cell poset]
	A \te{cell poset} $P_X$ is a poset constructed from a cell complex $X$ by
    \begin{equation}
    P_X\coloneqq \{e^n_\alpha:e^n_\alpha \ \text{is a $n$-cell in} \ X\},
    \end{equation}
    and for $e^n_\alpha$ and $e^m_\beta$, we define $e^n_\alpha\leq e^m_\beta$ if and only if $e^n_\alpha$ is contained in the closure of $e^m_\beta$ in $X$. 
\end{definition}

From now on, we will abuse notation and simply denote  $P_X$ as $X$. To avoid confusion, we will explicitly emphasize ``cell complex $X$'' or ``topological space $X$'' for the original $X$, and call $P_X$ ``cell poset $X$''.
We denote by $X(k)$ the set of $k$-dimensional cells in $X$, and call an element $\sigma\in X(k)$ a $k$-cell. For $\sigma\leq \tau\in X$ and $\dim(\sigma)<\dim(\tau)$, we will write $\sigma<\tau$. In particular, when $\dim(\sigma)=\dim(\tau)-1$, we write $\sigma\lessdot \tau$, and say $\sigma$ and $\tau$ are joined.

In this paper, we restrict attention to cell complexes whose cellular incidence numbers are either $\pm 1$ or $0$, or equivalently, a $k$-cell is incident to any given $(k-1)$-cell at most once.  This further implies that for fixed $k$-cell $\sigma$ and $(k+2)$-cell $\pi$ and $\sigma < \pi$, there exist even numbers of $(k+1)$-cell $\tau$ such that $\sigma \lessdot \tau \lessdot \pi$. For simplicity, one may also regard this property as the definition.

In fact, all results can be established without this assumption. However, this assumption substantially simplifies the notation and most constructions of interest satisfy it.
\begin{definition}[Sparse cell complex]
	We say that a family of cell complexes $\{X_n\}_{n=1}^\infty$ is \te{sparse} if for each $k$-cell in $X_n$, it is only joined to a uniform constant number of $(k+1)$-cells and $(k-1)$-cells.
\end{definition}

Throughout this paper, each $X_n$ is a finite cell complex, while the number of cells $|X_n|$ goes to infinity. We often neglect the subscript $n$ for $X_n$ and say $X$ is a sparse cell complex. To build LDPC codes on cell complex, this assumption is common. 

\begin{proposition}
	Suppose $X$ is a sparse cell complex, than for each $i<j<k$ and $j$-cell $f$, $|X_{\leq f}(i)|$, $|X_{\geq f}(k)|$ are all constants.  
\end{proposition}


\begin{definition}[Direct system and direct limit]\label{def:direct_limit}
Given a poset  $I$, a \te{direct system} of $\mathbb{F}$-vector spaces is a set of vector spaces and linear maps $\{V_i, f_{ij}:V_i\rightarrow V_j\}_{i\le j,\; i,j\in I}$ such that
\begin{equation}
f_{ii} = \mathrm{id}_{V_i}, \qquad f_{ik} = f_{jk}\circ f_{ij} \quad \text{for all } i\le j\le k.
\end{equation}
The \te{direct limit} of this system, denoted $\varinjlim_i V_i$, is defined as
\begin{equation}
\varinjlim_i V_i \;=\;
\left( \bigoplus_{i\in I} V_i \right)
\Big/ R,
\end{equation}
where $R$ is the subspace generated by all elements of the form
\begin{equation}
f_{ij}(v_i) - v_j, \qquad i\le j,\; v_i\in V_,\:v_j\in V_j.
\end{equation}
\end{definition}

\begin{definition}[Inverse system and inverse limit]
Given a poset $I$, an \te{inverse system} of $\mathbb{F}$-vector spaces is a set of vector spaces and linear maps $\{V_i, f_{ij}:V_j\rightarrow V_i\}_{i\le j,\; i,j\in I}$ such that
\begin{equation}
f_{ii} = \mathrm{id}_{V_i}, \qquad f_{ik} = f_{ij}\circ f_{jk} \quad \text{for all } i\le j\le k.
\end{equation}
The \te{inverse limit} of this system, denoted $\varprojlim_i V_i$, is the subspace
\begin{equation}
\varprojlim_i V_i
= \left\{ (v_i)_{i\in I} \in \prod_{i\in I} V_i
  \;\middle|\;
  f_{ij}(v_j)=v_i \text{ for all } i\le j
\right\}.
\end{equation}
\end{definition}

\begin{definition}[Presheaf]
Given a topological space $X$,  a \te{presheaf} $\F$ on  $X$ is a rule which assigns each open set $U\subseteq X$ a $\mathbb{F}$-linear space $\F(U)$ and to each inclusion $V\subseteq U$ a linear map called \te{restriction map} $\F_{U,V}:\F(U)\rightarrow\F(V)$ such that $\F_{U,U}=\id_U$, and whenever $W\subseteq V\subseteq U$ we have $\F_{U,W}=\F_{V,W}\circ \F_{U,V}$. 
\end{definition}

\begin{definition}[Morphism of presheaves]
Let $\F,\mathcal{G}$ be two presheaves. A \te{morphism} $\varphi : \mathcal{F} \to \mathcal{G}$ of presheaves on $X$ is a rule which assigns to each open $U \subseteq X$ a $\mathbb{F}$-linear map
\begin{equation}
\varphi(U) : \mathcal{F}(U) \to \mathcal{G}(U)
\end{equation}
compatible with restriction maps, i.e., whenever $V \subseteq U \subseteq X$ are open,
the diagram
\begin{equation}
\begin{tikzcd}
\mathcal{F}(U) \arrow[r, "\varphi"] \arrow[d, "\mathcal F_{U,V}"'] &
\mathcal{G}(U) \arrow[d, "\mathcal F_{U,V}"] \\
\mathcal{F}(V) \arrow[r, "\varphi"'] &
\mathcal{G}(V)
\end{tikzcd}
\end{equation}
commutes.
\end{definition}

\begin{definition}[Stalk]
Let $\mathcal F$ be a presheaf on $X$, and let $x \in X$.  
The \emph{stalk} of $\mathcal F$ at $x$ is defined as the direct limit of $\mathcal F(U)$ over all open sets $U$ containing $x$, namely
\begin{equation}
\mathcal F_x \coloneqq  \varinjlim_{x \in U} \mathcal F(U).
\end{equation}
\end{definition}

\begin{definition}[Sheaf]\label{def:sheaf}
A \te{sheaf} $\F$ on a topological space $X$ is a presheaf that satisfies the \te{equalizer condition}, i.e. for each open set $U$ and open cover $U=\bigcup_{i\in I}U_i$, the following sequence is exact 
	\begin{equation}
	\begin{tikzcd}
	0 \arrow[r] & \mathcal{F}(U) \arrow[r] & \displaystyle\prod_{i\in I} \mathcal{F}(U_i) \arrow[r]
	& \displaystyle\prod_{i,j\in I}\mathcal{F}(U_i\cap U_j) \\ [-4ex]
	& f \arrow[r, mapsto] &\displaystyle\prod_{i\in I}\mathcal{F}_{U,U_i}(f) \\ [-4ex]
	& & \displaystyle\prod_{i\in I}g_i \arrow[r,mapsto] & \displaystyle\prod_{i,j\in I}(\mathcal{F}_{U_i,U_i\cap U_j}(g)-\mathcal{F}_{U_j,U_i\cap U_j}(g)).
	\end{tikzcd}
	\end{equation}

	A morphism of sheaves is defined as a morphism of the underlying presheaves. Similarly, the stalk of a sheaf is defined to be the stalk of the underlying presheaf.

\end{definition}
\begin{definition}[Sheafification]
Let $\F$ be a presheaf on $X$. The \te{sheafification} of $\F$ is the sheaf $\widehat{\F}$ defined as follows: for each open set $U \subseteq X$, 
\begin{equation}
  \widehat{\F}(U)\coloneqq  \Bigl\{s \in \prod_{x \in U} \F_x \ :\ \forall x \in U,\;\exists V \ni x,\; V \subseteq U,\;\exists t \in \F(V)\text{ such that }t_y = s(y)\ \forall y \in V\Bigr\},
\end{equation}
where $t_y$ denotes the image of $t$ under the map $\F(V)\rightarrow\F_y$.
\end{definition}

It can be verified that sheafification preserve stalks, i.e. $\widehat{\F}_x\cong \F_x$ for each $x\in X$.

\subsection{Alexandrov topology}
\begin{definition}[Cellular sheaf]
	Let $X$ be a cell poset.  
A \emph{cellular sheaf} $\mathcal F$ on $X$ is a functor from $X$ to $\mathbb{F}$-linear spaces, i.e., assigns to  each cell $\sigma$ a linear space $\mathcal{F}_{\sigma}$, and to each $\sigma\leq \sigma'$  a linear map $\mathcal{F}_{\sigma, \sigma'}:\mathcal{F}_{\sigma}\rightarrow \mathcal{F}_{\sigma'}$ such that for all $\sigma\leq\sigma'\leq\sigma''$, we have $\mathcal{F}_{\sigma,\sigma''}=\mathcal{F}_{\sigma',\sigma''}\circ \mathcal{F}_{\sigma,\sigma'}$. 
A morphism $f$ between two sheaves $\mathcal{F}$ and $\mathcal{G}$ is a natural transformation between the corresponding functors, i.e., a collection of linear maps $
f_\sigma : \mathcal F_\sigma \to \mathcal G_\sigma$
such that for every $\sigma \le \tau$ the following diagram commutes:
	\begin{equation}
	\begin{tikzcd}
		\mathcal{F}_{\sigma} \arrow[r, "f_{\sigma}"] \arrow[d, "\mathcal{F}_{\sigma,\tau}"'] & 
		\mathcal{G}_{\sigma} \arrow[d, "\mathcal{G}_{\sigma,\tau}"] \\
		\mathcal{F}_{\tau} \arrow[r, "f_{\tau}"'] & 
		\mathcal{G}_{\tau}
	\end{tikzcd}
	\end{equation}
\end{definition}

One may naturally ask about the relationship between sheaves on cell posets and sheaves on topological spaces. This question is well studied, with Alexandrov topology playing a critical role.

\begin{definition}[Basis of topology]
	Given a set  $X$, a \te{basis} for a topology on $X$ is a collection $\mathcal{B}$ of subsets of $X$ satisfying the following conditions:
\begin{enumerate}
    \item For every $x \in X$, there exists $B \in \mathcal B$ such that $x \in B$.
    \item If $x \in B_1 \cap B_2$ for some $B_1, B_2 \in \mathcal B$, then there exists $B_3 \in \mathcal B$ such that
    \[
    x \in B_3 \subseteq B_1 \cap B_2.
    \]
\end{enumerate}
	{We define the induced topology by declaring a subset $U\subseteq X$ to be open if and only if it is a union of sets in $\mathcal B$.}
\end{definition}

Given a cell poset $X$, the set $\mathcal{B}\coloneqq \{X_{\geq \sigma}:\sigma\in X \}$ forms a basis of topology, since $X = \bigcup_{v\in X(0)} X_{\geq v}$, and whenever there exists $\rho\in X_{\geq\sigma}\cap X_{\geq\tau}$, we have $X_{\geq\rho}\subseteq X_{\geq\sigma}\cap X_{\geq\tau}$. For notational convenience, we let $U_{\sigma} \coloneqq  X_{\geq\sigma}$ denote the open set. Throughout this paper, we let $\mathcal{B}$ stand for the basis define above.

\begin{definition}[Alexandrov topology]
    Given a poset $X$, the \te{Alexandrov topology} of  $X$ is the topology generated by the basis 
    \begin{equation}
       \mathcal{B} \coloneqq  \{X_{\ge \sigma} : \sigma \in X \}, 
    \end{equation}
    {i.e., a subset $U\subseteq X$ is open if and only if it is a union of sets in $\mathcal B$.}
\end{definition}

Given the Alexandrov topology on a poset $X$ and a sheaf $\F$ on the poset, we are able to use the method from Kan extension to define $\F(U)$ on an arbitrary open set $U$,
\begin{equation}\label{eq: inverse limit as extended def of cell sheaf}
\F(U)\coloneqq \varprojlim_{U_\sigma\subseteq U}\F_\sigma.
\end{equation}
Actually this yields an one-one correspondence between sheaves on poset $X$ and sheaves on topological space $X$. For a proof, one may refer to  \cite[Theorem 4.2.10]{curry2014sheavescosheavesapplications}. Henceforth, we no longer need to distinguish between them. We note that this is a self-consistent notation because 
\begin{equation}
\F(U_\sigma)=\varprojlim_{U_{\rho}\subseteq U_\sigma}\F_\rho=\F_\sigma,
\end{equation}
and
\begin{equation}
\varinjlim_{\sigma\in U}\ \F(U)=\F(U_\sigma).
\end{equation}
Consequently, the stalk of the sheaf $\F$ on topological space $X$ at $\sigma$ should be exactly $\F_\sigma$ as defined in the sheaf on poset $X$: they are essentially the same.

We give a basic but important example of a sheaf:
\begin{proposition}\label{definition of F_k}
    Given a sheaf $\F$ over a t-dimensional cell complex $X$, for each $0\leq k \leq t $, we can construct sheaf $\F_k$ as follows: for each open set $U\subseteq X$, 
    \begin{equation}
    \F_{k}(U)\coloneqq \prod_{\sigma\in U(k)}\F_\sigma,
    \end{equation}
    where $U(k)$ denotes the set of $k$-cells contained in $U$.
    For $V\subseteq U$, we define $\F_{U,V}$ to be the restriction of domain, i.e., for a section $f\in \F_k(U)$, $\F_{U,V}(f)\coloneqq f|_{V}$.
\end{proposition}
\begin{proof}
	Consider an open set $U\subseteq X$ with an open cover $U = \bigcup_{i\in I}U_i$. Then
	\begin{equation}
	\begin{tikzcd}
	0 \arrow[r] & \mathcal{F}(U) \arrow[r] & \displaystyle \prod_{i\in I} \mathcal{F}(U_i) \arrow[r]
	& \displaystyle\prod_{i,j\in I}\mathcal{F}(U_i\cap U_j)\\ [-4ex]
	& f \arrow[r, mapsto] &\displaystyle\prod_{i\in I}f|_{U_i} \\ [-4ex]
	& & \displaystyle\prod_{i\in I} f_i \arrow[r,mapsto] & \displaystyle\prod_{i,j\in I}(f_i|_{U_i\cap U_j}-f_j|_{U_i\cap U_j}).
	\end{tikzcd}
	\end{equation}
	Suppose $\prod_{i \in I} f|_{U_i} = 0$, since $\{U_i\}_{i\in I}$ is an open cover of $U$, $f = 0$ and hence the first map is injective. Suppose $\prod_{i\in I} f_i$ is mapped to zero, then $f_i$ and $f_j$ always agree on the intersection of their domain. We can define $f \in \prod_{\sigma \in U(k)} \mathcal{F}_{\sigma}$ by $f|_{U_i} \coloneqq  f_i$. Therefore, $\mathcal{F}_k$ is indeed a sheaf. 
\end{proof}

\subsection{Chain complexes and quantum codes}
Here we formally introduce classical and quantum codes, and the chain complex descriptions of them.
\begin{definition}[Classical code]
    A \te{classical linear code} with parameters $[n,k,d]_q$ over alphabet $\mathbb{F}_q$ is a $k$-dimensional linear subspace $\mathcal C \subseteq \mathbb{F}_q^n$. The distance of the code is $d = \min_{x \in \mathcal C \backslash \{0\}} |x|_H$, where $|\cdot|_H$ denotes the Hamming weight. When the alphabet is $2$ or clear from context, we simply write $[n,k,d]$. The \te{dual code} of $\mathcal C$ is $\mathcal C^\perp = \{y \in \mathbb{F}_q^n: x \cdot y =0 \ \forall x \in \mathcal{C}\}$. A \emph{generator matrix} of $\mathcal{C}$ is a matrix $G \in \mathbb{F}_q^{k\times n}$  such that $\mathcal{C} = \mathrm{Im}(G^\top)$. A \te{parity check matrix} of $\mathcal{C}$ is a matrix $H$ for which $\mathcal{C} = \ker(H)$.

    The code $\mathcal{C}$ has \te{sparsity} $\Delta$ if there exists a parity-check matrix $H \in \mathbb{F}_q^{r \times n}$ such that every row and column of $H$ has at most $\Delta$ nonzero entries. The code \(\mathcal{C}\) defined by the check matrix \(H \in \mathbb{F}_q^{r \times n}\) is \te{locally testable} with soundness \(\rho\) if for all \(x \in \mathbb{F}_q^{n}\),
  \begin{align}
    \frac{|H x|}{r} \ge \rho \frac{d(x, \mathcal{C})}{n}.
  \end{align}
\end{definition}

\begin{definition}[Chain complex and cochain complex]
    A \te{chain complex} $C_\bullet$ is a sequence of $\mathbb{F}$-vector spaces 
    \begin{equation}
    C_\bullet=\cdots{\longrightarrow} C_{i+1}\overset{\partial_{i+1}}{\longrightarrow} C_i\overset{\partial_i}{\longrightarrow} C_{i-1}\longrightarrow\cdots
    \end{equation}
    such that $\partial_i\partial_{i+1}=0$. The corresponding \te{cochain complex} $C^\bullet$ is defined by applying the functor $\text{Hom}(-,\mathbb{F})$
    \begin{equation}
    C^\bullet=\cdots{\longrightarrow} C^{i-1}\overset{\delta^{i-1}}{\longrightarrow} C^i\overset{\delta^i}{\longrightarrow} C^{i+1}\longrightarrow\cdots
    \end{equation}
    where $C^i=\text{Hom}(C_i,\mathbb{F})$ is the dual vector space of $C_i$, $\delta^i=\text{Hom}(\partial_{i+1},\mathbb{F})$.
\end{definition}

Note that we only consider finite-dimensional linear spaces, hence there is a canonical isomorphism $C^i\cong C_i$, under which the coboundary operator $\delta^i=\partial_{i+1}^T$ is simply the matrix transpose of the boundary operator. 

\begin{definition}
    The \te{$i$-boundaries} are elements in $B_i(C_\bullet)\coloneqq \im \partial_{i+1}$ and the \te{$i$-cycles} are elements in $Z_i(C_\bullet)\coloneqq \ker \partial_i$. Similarly, the \te{$i$-coboundaries} $B^i(C^\bullet)\coloneqq \im \delta^{i-1} $ and the \te{$i$-cocycles} are $Z^i(C^\bullet)\coloneqq \ker \delta^i$.
    
    The $i$-th \te{homology group} of the chain complex $C_\bullet$ and \te{cohomology group} of the cochain complex $C_\bullet$ are defined respectively by
    \begin{equation}
    H_i(C_\bullet)\coloneqq \ker \partial_i/\im \partial_{i+1},\quad H^i(C^\bullet)\coloneqq \ker\delta^i/\im\delta^{i-1}.
    \end{equation}
\end{definition}
\begin{definition}[Chain map]\label{def:chain_map}
    Suppose $C_\bullet$ and $D_\bullet$ are chain complexes with boundary map $\partial_C$ and $\partial_D$, respectively. We say  $f_\bullet: C_\bullet \rightarrow D_\bullet$ is a \te{chain map} if $f_\bullet \circ \partial_C = \partial_D \circ f_\bullet$.  \te{Cochain maps} between cochain complexes are defined analogously.
\end{definition}

\begin{definition}[CSS quantum code]
    A \te{Calderbank–Shor–Steane (CSS) code} with parameters $[\![n,k,d]\!]_q$ over alphabet $\mathbb{F}_q$ is a tuple of two classical codes $\mathcal Q = (\mathcal C_X, \mathcal C_Z)$ such that $\mathcal C_X^\perp \subseteq \mathcal C_Z$.

Given a cochain complex $C^\bullet$, we can define an associated CSS code by placing physical qubits on $C^i$, $Z$-checks on $C^{i+1}$ and $X$-checks on $C^{i-1}$ with the corresponding parity-check matrices given by $H_Z\coloneqq \delta^i$, $H_X\coloneqq \partial_i$ where $\partial_i = (\delta^{i-1})^{T}$ under the canonical identification of chains and cochains. 
The number of logical qudits is
\begin{equation}
k_i \coloneqq  \dim H^{i}(C^\bullet)
     = \dim \ker H_Z - \dim \operatorname{im} H_X^{T},
\end{equation}
and the distance is $d_i \coloneqq  \min\{d_X,d_Z\}$, where
\begin{equation}
d_X = \min_{c^i \in \ker H_Z \setminus \operatorname{im} H_X^T} |c^i|_H,
\qquad
d_Z = \min_{c_i \in \ker H_X \setminus \operatorname{im} H_Z^T} |c_i|_H .
\end{equation}
\end{definition}

\begin{definition}[Quantum low-density parity-check code (qLDPC)]
A family of CSS codes
\(\{\mathcal{Q}_n\}_{n \in \mathbb{N}}\) where each code \(\mathcal{Q}_n\) has blocklength \(n\)
and parity-check matrices \(H_X^{(n)}\) and \(H_Z^{(n)}\)  is called \emph{quantum low-density parity-check}, if there exists a constant
\(w = O(1)\) such that
every row of \(H_X^{(n)}\) and \(H_Z^{(n)}\) has Hamming weight at most \(w\), uniformly in \(n\).
\end{definition}

\begin{definition}[Quantum locally testable code (qLTC)]
A family of CSS codes
\(\{\mathcal{Q}_n\}_{n \in \mathbb{N}}\) where each code \(\mathcal{Q}_n\) has blocklength \(n\)
and parity-check matrices \(H_X^{(n)}\) and \(H_Z^{(n)}\) is called \emph{quantum locally testable}, if there exists a constant \(\rho>0\)
independent of \(n\) such that both \(H_X^{(n)}\) and \(H_Z^{(n)}\) define classical
locally testable codes with soundness at least \(\rho\).
Namely, for all \(n\) and for all \(x \in \mathbb{F}_q^{n}\),
\begin{equation}
\frac{|H^{(n)} x|}{r_n} \;\ge\; \rho \,\frac{d(x,\mathcal{C}_n)}{n},
\end{equation}
where \(H^{(n)}\) denotes either \(H_X^{(n)}\) or \(H_Z^{(n)}\),
\(r_n\) is the number of checks, and \(\mathcal{C}_n\) is the corresponding classical code.
\end{definition}

\subsection{Sheaf cohomology}\label{sec:sheaf_cohomology}
\subsubsection{Cohomology of sheaf codes}
{
Let $X$ be a cell complex equipped with a cellular sheaf $\F$. We define the corresponding sheaved cellular chain complex by setting
\begin{equation}
    C^i(X,\F)\coloneqq \prod_{\sigma\in X(i)}\F_\sigma,
\end{equation}
with the coboundary map $\delta:  C^i(X, \mathcal{F}) \rightarrow C^{i+1}(X, \mathcal{F})$ given by
\begin{equation}\label{eq: sheaved coboundary map}
(\delta \alpha)(\sigma') = \sum_{\sigma \lessdot \sigma'} \mathcal{F}_{\sigma,\sigma'}(\alpha(\sigma)).
\end{equation}
The cohomology of this sheaved chain complex, which we refer to as the cohomology of sheaf codes, is defined as
\begin{equation}
     H^i(X,\F)\coloneqq \ker\delta^i/\im\delta^{i-1}.
\end{equation}

Throughout this paper we consider finite-dimensional vector spaces. We may therefore identify each vector space $C^i(X,\F)$ with its dual $C_i(X,\F)$ by choosing a fixed basis. Under this identification, the boundary map $\partial_{i+1}:C_{i+1}(X,\F)\rightarrow C_{i}(X,\F)$ is taken to be the transpose of the matrix representing $\delta^{i}$; explicitly,
\begin{equation}\label{eq: sheaved boundary map}
(\partial x)(\sigma'') = \sum_{\sigma\gtrdot \sigma ''} \mathcal{F}_{\sigma'',\sigma}^T(x(\sigma)).
\end{equation}
The corresponding homology group is then defined as
\begin{equation}
H_i(X,\F)\coloneqq \ker \partial_i/\im \partial_{i+1}.
\end{equation}

Note that for each cell $\sigma\in X$, we can restrict the sheaf $\F$ to $U_\sigma$. This restriction yields a new chain complex, called the \emph{local chain complex}, defined by
\begin{equation}
    C^i(U_\sigma,\F)\coloneqq \prod_{\tau\in U_\sigma(i)}\F_\tau,
\end{equation}
with coboundary maps $\delta_L$ and boundary maps $\partial_L$ defined analogously to Eqs.~\eqref{eq: sheaved coboundary map} and \eqref{eq: sheaved boundary map}, where all cells appearing in the sums are taken within $U_\sigma$.

}

In our sheaf code context, we use the \te{block Hamming weight} instead of the standard Hamming weight.

\begin{definition}[Block Hamming weight]
The \emph{block Hamming weight} on a sheaf chain complex $C^\bullet(X,\mathcal{F})$ (or $C_\bullet(X,\mathcal{F})$) is defined as follows: for $\alpha \in C^i(X,\mathcal{F})$,
\begin{equation}
\|\alpha\| \coloneqq \sum_{\sigma \in X(i)} \mathbbm{1}_{\alpha(\sigma) \neq 0}.
\end{equation}
\end{definition}

It is easy to see that, for a sparse complex $X$, the block Hamming weight differs from the usual Hamming weight only by a constant factor. We also find that the block Hamming weight is more convenient to use, so we will adopt it throughout this paper.

With the block Hamming weight, we are ready to define the sheaf code notions of the (co)systolic distance and (co)cycle expansion:

\begin{definition}[Systolic and cosystolic distances]
Let $C_\bullet(X,\mathcal{F})$ be a chain complex with boundary maps $\partial_i$ and associated cochain complex
$C^\bullet(X,\mathcal{F})$ with coboundary maps $\delta^i$.
The \emph{systolic distance} and \emph{cosystolic distance} at degree $i$ are defined respectively as
\begin{equation}
\mu_\partial(i)
\coloneqq
\min_{x \in \ker \partial_i \setminus \operatorname{im} \partial_{i+1}}
\|x\|,
\qquad
\mu_\delta(i)
\coloneqq
\min_{\alpha \in \ker \delta^i \setminus \operatorname{im} \delta^{i-1}}
\|\alpha\|.
\end{equation}
\end{definition}

Then the following parameters are closely related to soundness.

\begin{definition}[Cycle and cocycle expansion]
    The \te{cycle expansion} for $\partial_i$ and the \te{cocycle expansion} for $\delta^i$ are defined respectively by
    \begin{equation}
    \varepsilon_\partial(i)=\min_{x\in C_i(X,\mathcal{F})\setminus\ker\partial_i}\frac{\|\partial_ix\|}{\mathrm{dist}(x,\ker\partial_i)},\quad   \varepsilon_\delta(i)=\min_{\alpha\in C^i(X,\mathcal{F})\setminus\ker\delta^i}\frac{\|\delta^i\alpha\|}{\mathrm{dist}(\alpha,\ker\delta^i)}.
    \end{equation}
\end{definition}

The (co)cycle expansion differs from the soundness of the associated code only by constant factors. For a detailed discussion of the relationship between (co)systolic distance, (co)cycle expansion, distance, and soundness, see e.g.,~\cite[Lemma 2.7]{Dinur2024sheaf}.

\subsubsection{Right derived functor}

\begin{definition}[Sheaf cohomology]
    Let $X$ be a topological space. For any sheaf $\mathcal{F}$ on $X$, the \emph{sheaf cohomology} of $\mathcal{F}$ is defined by
\begin{equation}
    H^p(X,\mathcal{F}) \coloneqq  (R^p \Gamma)(\mathcal{F}),
\end{equation}
where 
\begin{equation}
\Gamma:\F\longmapsto\Gamma(X,\F)\coloneqq \F(X).
\end{equation}
denotes the global sections functor and $R^p\Gamma$ its $p$-th right derived functor.
\end{definition}

In short, to calculate sheaf cohomology, one needs to first find an \te{injective resolution} $\mathcal{I}^\bullet$ of $\F$; that is, each $\mathcal{I}^n$ is an injective sheaf, and there exists an exact sequence
\begin{equation}
0\longrightarrow\F\longrightarrow\mathcal{I}^0\longrightarrow\mathcal{I}^1\longrightarrow\mathcal{I}^2\longrightarrow\cdots
\end{equation}
of sheaves.
Then we delete the term $\F$ and apply the functor $\Gamma$ to obtain a chain complex
\begin{equation}
0\longrightarrow\Gamma(X,\mathcal{I}^0)\longrightarrow\Gamma(X,\mathcal{I}^1)\longrightarrow\Gamma(X,\mathcal{I}^2)\longrightarrow\Gamma(X,\mathcal{I}^3)\longrightarrow\cdots
\end{equation}
The sheaf cohomology groups $H^p(\mathcal{F})$ are then defined as the $p$-th cohomology of this complex. It is well known in homological algebra that $H^p(\mathcal{F})$ is independent of the choice of injective resolution. Moreover, an injective resolution always exists in our setting, since we only consider finite-dimensional vector spaces over a finite field.

In this paper, we will not use the concept of injective sheaves. It is well-known that if each sheaf $\mathcal{I}^\bullet$ is \te{flabby}, then an exact sequence of sheaves can also be used to compute sheaf cohomology in the same manner.

\begin{definition}[Flabby sheaf]\label{def:flabby}
    We say that a sheaf $\F$ over space $X$ is \te{flabby} if, for each open set $U\subseteq V\subseteq X$, the restriction map $\F(V)\rightarrow\F(U)$ is surjective.
\end{definition} 

It is straightforward to obtain the following proposition by definition.

\begin{proposition}
    The sheaf $\F_k$ defined in Proposition \ref{definition of F_k} is flabby.
\end{proposition}
\subsubsection{\v{C}ech Cohomology}

Let $\mathcal{U}=\{U_i\}_{i\in I}$ be an open cover of a topological space $X$. We define $N(\mathcal{U})$ to be the abstract simplicial complex with each vertex corresponding to an open set in $\mathcal{U}$, and each $p$-simplex as a $(p+1)$-tuple $\sigma$ of distinct open sets, $\sigma=[U_{i_0},\cdots,U_{i_p}]$ with $\bigcap_{j=0}^pU_{i_j}\neq \varnothing$. Given a presheaf $\mathcal{F}$ on $X$, the \v{C}ech $p$-cochain $C^p(\mathcal{U},\mathcal{F})$ is defined by
\begin{equation}
C^p(\mathcal{U},\F)\coloneqq \prod_{[U_{i_0},\cdots,U_{i_p}]}\F(U_{i_0}\cap\cdots\cap U_{i_p}),
\end{equation}
where the product is taken over the $p$-simplex in $N(\mathcal{U})$, and we denote $U_\sigma=U_{i_0}\cap\cdots\cap U_{i_p}$ (or simply write as $U_{i_0\cdots i_p}$). A $p$-cochain $f\in C^p(\mathcal{U},\F)$ is a function assigning each simplex an element in $\mathcal{F}(U_\sigma)$. The coboundary operator $\delta^p:C^p(\mathcal{U},\F)\rightarrow C^{p+1}(\mathcal{U},\F)$ is defined by, if $\tau=[U_{i_0},\cdots,U_{i_{p+1}}]$ is a $(p+1)$-simplex, then
\begin{equation}\label{eq: coboundary of sheaved cech}
(\delta^pf)(\tau)=\sum_{j=0}^{p+1}(-1)^jf_{i_0\cdots\hat{i}_j\cdots i_{p+1}}|_{U_\tau},
\end{equation}
where the sum on the right hand side means, first get the value of $f$ on $[U_{i_0},\cdots, \hat{U}_{i_j},\cdots,U_{i_{p+1}}]$, denoted as $f_{i_0\cdots\hat{i}_j\cdots i_{p+1}}|_{U_\tau}$, then apply the restriction map $\F(U_{i_0}\cap\cdots\cap \hat{U}_{i_j}\cap\cdots\cap U_{i_{p+1}})\rightarrow\F(U_\tau)$, and sum up.

\begin{definition}[\v{C}ech cohomology group]
    Let $X$ be a topological space with sheaf $\F$, then the $p$-th \te{\v{C}ech cohomology group} is defined as the direct limit
    \begin{equation}
    \check{H}^p(\F)\coloneqq \varinjlim_{\mathcal{U}}\check{H}^p(\mathcal{U},\F).
    \end{equation}
\end{definition}

\subsubsection{Connection between three cohomology theories}
In this section, we focus on a simplicial complex $X$. A key observation is that there is a particular choice of open cover $\mathcal{V}\coloneqq \{U_v:v\in X(0)\}$ of $X$. For $v_0, v_1, \cdots, v_p\in X(0)$, if $U_{v_0}\cap\cdots\cap U_{v_p}\neq \varnothing$, then $[v_0,\cdots, v_p]$ must be a $p$-simplex in $X$, and
\begin{equation}
U_{v_0}\cap\cdots\cap U_{v_p}=U_{[v_0,\cdots,v_p]}
\end{equation}
Therefore, if $\F$ is a sheaf over $X$, the cellular cochain complex always agrees with the \v{C}ech cochain complex:
\begin{equation}
C^p(X,\F)\cong C^p(\mathcal{V},\F).
\end{equation}
It is easy to verify that the coboundary maps agree. Therefore, we always have
\begin{equation}
H^p(X,\F)\cong \check{H}^p(\mathcal{V},\F).
\end{equation}
Furthermore, note that for any other choice of open cover $\mathcal{U}$ of $X$, $\mathcal{V}$ is always a refinement of $\mathcal{U}$, since there must exist open sets $U_1, U_2,\dots\in \mathcal{U}$ such that each vertices $v_1, v_2,\dots\in X(0)$ belong to at least one of these open sets. Without loss of generality, we assume $v_i\in U_i$. By the definition of topology basis $\mathcal{B}$, $U_i$ is the union of sets in $\mathcal{B}$, and the only open set containing $v_i$ in $\mathcal{B}$ is $U_{v_i}$. Consequently, $U_i\supseteq U_{v_i}$, i.e., $\mathcal{V}$ is always a refinement of $\mathcal{U}$. So we have
\begin{equation}
\check{H}^p(\F)=\varinjlim_{\mathcal{U}}\check{H}^p(\mathcal{U},\F)=\check{H}^p(\mathcal{V},\F)\cong H^p(X,\F),
\end{equation}
meaning that the \v{C}ech cohomology is exactly isomorphic to the cohomology of sheaf codes. As shown in \cite[Section 7.3]{curry2014sheavescosheavesapplications} or \cite[Theorem A.8]{First2022}, the cohomology of sheaf codes is isomorphic to the sheaf cohomology (not necessarily a simplicial complex). As a result, we conclude that the three cohomology groups are all isomorphic:
\begin{equation}
H^p(X,\F)\cong \check{H}^p(\F)\cong (R^p\Gamma)(\F).
\end{equation}

This enables us to switch between different perspectives and obtain more insights and results. For example, the perspective of the right derived functor provides a natural framework to prove the Poincar\'e duality in Theorem~\ref{thm:Poincaré_duality}, and the pullback sheaf is locally acyclic as established in Corollary~\ref{pullback sheaf is locally acyclic}. Meanwhile, the perspective of \v{C}ech cohomology motivates the definition of the tensor product of sheaves on sheaved cellular chain complexes, as discussed in Section~\ref{sec:multiplicative}.


\subsection{Logical multi-controlled-$Z$ gates}\label{sec:controlled-Z}
Finally, we present the essential preliminaries for the study of logical multi-controlled-$Z$ gates on quantum codes, which represent a key application of our theory.

\begin{definition}[Multi-controlled-$Z$ gate]
Let $\mathbb F_q$ be a finite field of characteristic $p$, let $r \ge 2$ be an integer, and let $t \in \mathbb F_q^\ast$.  
The \te{multi-controlled-$Z$ gate} $C^{r-1}Z_q^t : (\mathbb C^{q})^{\otimes r} \to (\mathbb C^{q})^{\otimes r}$ is the $r$-qudit diagonal unitary defined by its action on the computational basis:
\[
C^{r-1}Z_q^t \ket{x_1,\dots,x_r}
=
\exp\!\left(
\frac{2\pi i}{p}\,
\Tr_{\mathbb F_q / \mathbb F_p}
\!\bigl(t\, x_1 x_2 \cdots x_r\bigr)
\right)
\ket{x_1,\dots,x_r},
\]
where $x_1,\dots,x_r \in \mathbb F_q$ and $\Tr_{\mathbb F_q / \mathbb F_p}$ denotes the field trace.
\end{definition}

In this work, we will restrict attention to $p=2$ and omit $q$ when it is clear from the context. When $t=1$, we also omit the superscript $t$. We mainly focus on constructing logical multi-controlled-$Z$ gates on quantum sheaf codes for $r=2,3$, but the framework readily generalizes to arbitrary $r$.

\begin{definition}[Cohomological invariant]\label{def:invariant}
    Given cochain complexes $C^\bullet_{(1)},...C^\bullet_{(r)} $ over $\mathbb{F}_q$, an integer $i$, we say a multilinear form $\xi: C^i_{(1)}\times ... \times C^i_{(r)} \rightarrow \mathbb{F}_q$ is a \te{cohomology invariant} if for every $z_{(h)} + B^i(C^\bullet_{(h)}) \in H^i(C^\bullet_{(h)}), \ \forall h \in [r]$ it holds that $\xi(z'_{(1)},...,z'_{(r)})$ is the same for every choice of representatives $z'_{(h)} \in z_{(h)} + B^i(C^\bullet_{(h)}), \ \forall h \in [r]$. In other words, $\xi$ induces a well-defined multilinear form on the cohomology spaces $H^i(C^\bullet_{(h)})$. 
    
    We say that $\xi$ has \te{sparsity} $\Delta_\xi$ if for every $h \in [r]$, each $i$-cell in $C_{(h)}^i$ is involved in at most $\Delta_\xi$ monomials in $\xi$.
\end{definition}

\begin{definition}[Cohomology subrank]
Given cochain complexes $C^\bullet_{(1)},...C^\bullet_{(r)} $ over $\mathbb{F}_q$, an integer $i$, and a cohomology invariant $\xi: C^i_{(1)}\times ... \times C^i_{(r)} \rightarrow \mathbb{F}_q$,  the \te{cohomology subrank} of $\xi$ is  the largest number $s$ such that there exist $s$ tuples of cohomology elements $(z^j_{(1)} ,..., z^j_{(r)}) \in H^i_{(1)}\times ...\times H^i_{(r)}$ for each $j \in [s]$, satisfying $\xi(z^{j_1}_{(1)} ,..., z^{j_r}_{(r)}) = \mathrm{id}_{j_1=...j_r}$ for every $(j_1,...,j_r) \in [s]^r$. 
\end{definition}

Such multilinear forms give rise to a circuit composed of multi-controlled-$Z$ gates on the $r$ quantum code blocks that leaves the codespaces invariant.

\begin{lemma}[See e.g.\ \protect{\cite[Lemma 3.42]{Golowich_Lin2024}}]\label{lem:cohomology invariant and CCZ}Consider $r$ quantum codes $\mathcal Q_1,...,\mathcal Q_r$ defined by placing qubits at level $i$ (Z checks on $i+1$ and X checks on $i-1$) of the cochain complexes $C^\bullet_{(1)},...C^\bullet_{(r)}$ and a cohomology-invariant multilinear form $\xi$, with constant sparsity and cohomology subrank $s$. Then one can construct a constant-depth quantum circuit composed of multi-controlled-$Z$ gates where each monomial in $\xi$ corresponds to a gate $C^{r-1}Z^t$ which is a logical action on the code spaces. The number of physical gates, denoted $n_{C^{r-1}Z}$, is equal to the number of monomials in $\xi$, and the number of induced logical $C^{r-1}Z$ gate is denoted $k_{C^{r-1}Z} \geq s$. We say the quantum codes $\mathcal{Q}_1,...,\mathcal Q_r$ support a \te{constant-depth $C^{r-1}Z$ circuit inducing $s$ logical $C^{r-1}Z$ gates}. If there is no lower bound on the cohomology subrank, we say the codes support a \te{constant-depth logical $C^{r-1}Z$ circuit}. 
\end{lemma}

It is often desirable to realize a transversal logical circuit, namely a depth-1 circuit, equivalently one in which each variable of $\xi$ is involved in at most $1$ monomial. Indeed, we can always convert a constant-depth logical multi-controlled-$Z$ circuit into a transversal one by concatenating the quantum codes with the repetition code.

\begin{lemma}[\cite{nguyen2025good, Golowich_Lin2024}]\label{lem:constant-depth-to-transversal}
    Suppose the quantum codes $\mathcal Q_1,...,\mathcal Q_r$ support a constant-depth $C^{r-1}Z$ circuit inducing $s$ logical $C^{r-1}Z$ gates. Then we can construct new codes $\mathcal Q'_1,...,\mathcal Q'_r$ that support a \te{transversal $C^{r-1}Z$ circuit} inducing $s$ logical $C^{r-1}Z$ gates. Furthermore, the parameters of the new codes (sparsity, qubits, rate, relative distance, soundness) worsen by at most a constant factor.
\end{lemma}


\section{Poincaré duality for quantum codes}\label{sec:duality}

In this section, we prove the foundational result of this work: the Poincar\'e duality of quantum codes  (Theorem \ref{thm:Poincaré_duality}). Historically, Poincar\'e duality arose in the study of manifolds \cite{Bott:1982xhp,Hatcheralgtop}. For an oriented closed $n$-dimensional manifold $M$, the $i$-th cohomology group is canonically isomorphic to the $(n-i)$-th homology group, with the isomorphism given by the cap product with the fundamental class. This theorem is a cornerstone of algebraic topology and underlies many structural results in geometry and topology.

An analogue of this phenomenon has recently appeared in the study of quantum codes. In \cite{DHLV2022,Dinur2024sheaf,nguyen2025quantum}, it was shown that for $t$-dimensional cubical complexes equipped with specific sheaves, various code parameters including code distance, soundness, (co)cycle expansion, and decoder at degree $i$ are related to the corresponding parameters at degree $t-i$ for another sheaf. These observations naturally raise the question of whether there exists a generic and rigorous formulation of the ``Poincaré duality” principle for quantum codes.

Here, we prove that for locally acyclic sheaves on sparse cell complexes, there is indeed such a duality that relates the logical qubit number, code distance, (co)cycle expansion and decoder, and this duality takes a conceptually cleaner and more uniform form than in previously known examples.
Furthermore, later  in Section~\ref{sec:multiplicative},  we demonstrate that the isomorphism of logical qubits (homology groups) can be established using cap product.

In Section~\ref{sec:strong_sheaf}, we generalize Lemma 6.2 and Lemma 6.3 in \cite{lin2024transversalnoncliffordgatesquantum} from simplicial complexes to cell complexes, and establish the framework of locally acyclic sheaves. In Section~\ref{sec:sheaf_code}, we introduce (dual) sheaves generated by classical local codes. In Section~\ref{sec:duality-proof}, we formally prove our Poincar\'e duality theorem for quantum codes.

\subsection{Strong sheaf axiom and local acyclicity}\label{sec:strong_sheaf}

Given a sheaf $\mathcal{F}$ on a cell complex $X$, recall the sheaf $\mathcal{F}_k$ defined in Proposition~\ref{definition of F_k}, for each $i$-cell $\sigma$ and $k \geq i$ we have 
\begin{align}\label{eq:F_k}
	\mathcal{F}_{k,\sigma} = C^k(X_{\geq \sigma},\mathcal{F}) = \prod_{\tau \in X_{\geq \sigma}(k)} \mathcal{F}_\tau = \prod_{\tau \in U_{\sigma}(k)} \mathcal{F}_\tau.
\end{align}
For any pair of cells $\sigma' \geq \sigma$, there is a natural map $\mathcal{F}_{k,\sigma,\sigma'}: \mathcal{F}_{k,\sigma} \rightarrow \mathcal{F}_{k,\sigma'}$ 
\begin{align}\label{eq:F_k_morphism}
	\mathcal{F}_{k,\sigma,\sigma'}(x) = x \vert_{X_{\geq \sigma'}(k)},
\end{align}
which is a restriction of domain as $X_{\geq \sigma} \supseteq X_{\geq \sigma'}$. Now we define an important double complex. This is a standard construction associated with a complex of sheaves.  Let
\begin{align}
	K^{p,q} \coloneqq  C^p(X,\mathcal{F}_{t-q}) = \prod_{\sigma \in X(p)} C^{t-q}(X_{\geq \sigma},\mathcal{F}).
\end{align}

We now define coboundary operators on the double complex $K^{p,q}$ as follows.

{
We define $d': K^{p,q} \rightarrow K^{p+1,q}$ as the coboundary operator of the sheaved cellular chain complex $C^\bullet(X,\F_{t-q})$. Explicitly, for $\alpha \in K^{p,q}$ and a $(p+1)$-cell $\sigma'$, we set
\begin{align}\label{eq:d'_1}
    (d'\alpha)(\sigma') = \sum_{\sigma \lessdot \sigma'} \mathcal{F}_{k,\sigma,\sigma'}( \alpha(\sigma) ) = \sum_{\sigma \lessdot \sigma'} \alpha(\sigma) \vert_{X_{\geq \sigma'}(t-q)}.
\end{align}
Here $\mathcal{F}_{k,\sigma,\sigma'}$ is simply the restriction map. Consequently, for any $\pi \in X_{\geq \sigma'}(t-q)$ with $\sigma' \geq \sigma$, we have
\begin{align}\label{eq:d'_2}
    (d'\alpha)(\sigma',\pi) = \sum_{\sigma \lessdot \sigma'} \mathcal{F}_{k,\sigma,\sigma'}( \alpha(\sigma) )(\pi) = \sum_{\sigma \lessdot \sigma'} \alpha(\sigma,\pi).
\end{align}
}

To define $d'':K^{p,q}\rightarrow K^{p,q+1}$, we first need to define the \emph{local boundary operator}:
\begin{align}
	C^k(X_{\geq \sigma},\mathcal{F}) \xrightarrow{\partial_L} C^{k-1}(X_{\geq \sigma},\mathcal{F})
\end{align}
by
\begin{align}
	(\partial_L (x) ) (\pi) = \sum_{\sigma \le \pi \lessdot \pi'} \mathcal{F}_{\pi,\pi'}^T (x(\pi')).
\end{align}
Similarly, we can define the local coboundary operator $\delta^L$. Let $\alpha(\sigma) = x$, then we set
\begin{align}\label{eq:d''}
	(d''\alpha)(\sigma) = \partial_L (\alpha(\sigma)) \in K^{p,q+1} = C^p(X,\mathcal{F}_{t-q-1}).
\end{align}
Note that the double complex $K^{p,q}$ is only defined for $p,q \geq 0$ and $p+q \leq t$. We extend the definition to $p,q \in \mathbb{Z}$ by defining the remaining terms and coboundary maps as zero, e.g., we may write $C^t(X,\mathcal{F}_{t-1})=0$.

\begin{proposition}
	The coboundary operators $d'$ and $d''$ commute.
\end{proposition}
\begin{proof}
    This can be verified by direct calculation. Applying $d'$ followed by $d''$, we obtain
    \begin{align}
    \begin{aligned}
    	(d''d'\alpha)(\sigma,\pi) & = (\partial_L(d'\alpha)(\sigma))(\pi)
    	= \sum_{\pi' \gtrdot \pi \gtrdot \sigma} \mathcal{F}_{\pi,\pi'}^T (d'\alpha(\sigma,\pi'))
    	= \sum_{\pi' \gtrdot \pi \gtrdot \sigma} \mathcal{F}_{\pi,\pi'}^T \left( \sum_{\sigma' \lessdot \sigma} \alpha(\sigma', u') \right) \\
    	& = \sum_{\sigma' \lessdot \sigma \lessdot \pi \lessdot \pi'} \mathcal{F}_{\pi,\pi'}^T(\alpha(\sigma',\pi')).
    \end{aligned}
    \end{align}
    On the other hand, applying $d''$ followed by $d'$, we obtain
    \begin{align}
    \begin{aligned}
    	(d'd''\alpha)(\sigma,\pi) & = \sum_{\sigma' \lessdot \sigma}(d''\alpha)(\sigma,\pi)
    	= \sum_{\sigma' \lessdot \sigma} (\partial_L \alpha(\sigma))(u) 
    	= \sum_{\sigma' \lessdot \sigma} \sum_{\pi' \gtrdot \pi \gtrdot \sigma} \mathcal{F}_{\pi,\pi'}^T (\alpha(\sigma',\pi')) \\
    	& = \sum_{\sigma' \lessdot \sigma \lessdot \pi \lessdot \pi'} \mathcal{F}_{\pi,\pi'}^T(\alpha(\sigma',\pi')).
    \end{aligned}
    \end{align}
    This completes the proof.
\end{proof}

In general, the sheaf axiom (Definition \ref{def:sheaf}) are not straightforward to verify directly, so we develop a criterion for checking it.

\begin{definition}[Strong sheaf axiom]\label{def:strong_sheaf}
    We say that a sheaf $\mathcal{F}$ on a cell complex $X$ satisfies the \emph{strong sheaf axiom} if there is an exact sequence for any $\sigma \in X(i)$ and $0 \leq i\leq t-2$
    \begin{equation}
    \begin{tikzcd}[column sep=0.7cm]
    		0 \arrow[r] & \mathcal{F}_{\sigma} \arrow[r,"\delta_L"] 
    		& \displaystyle C^{i+1}(X_{\geq \sigma},\mathcal{F})  = 
    		\hspace{-3mm} \prod_{\sigma_{i+1} \in X_{\geq \sigma}(i+1)} \hspace{-3mm} \mathcal{F}_{\sigma_{i+1}} \arrow[r,"\delta_L"]
    		& \displaystyle C^{i+2}(X_{\geq \sigma},\mathcal{F}) = 
    		\hspace{-3mm} \prod_{\sigma_{i+2}\in X_{\geq \sigma}(i+2)} \hspace{-3mm} \mathcal{F}_{\sigma_{i+2}}.
    \end{tikzcd}
    \end{equation}
    And when $i=t-1$, we only require $\mathcal{F}_{\sigma} \rightarrow \prod_{\sigma_{t}\in X_{\geq \sigma}(t)}\mathcal{F}_{\sigma_{t}}$ to be injective.
\end{definition}

\begin{remark}\label{remark:strong_sheaf}
    Obviously, when $X$ is a simplicial complex, then the sheaf axiom naturally implies the strong sheaf axiom. And by Theorem \ref{thm:strong_sheaf}, the strong sheaf axiom will also imply the sheaf axiom. However, this may not be true for general cell complexes.

    Without having the sheaf axiom, the exactness at $\mathcal{F}_\sigma$ holds when the corresponding local coboundary operator is injective, which is equivalent to requiring the parity-check matrices of the local codes to be of full-rank (see Section \ref{sec:sheaf_code}). Intuitively, by viewing nonzero scalars in $\mathbb{F}$ as rank-1 matrices, full-rank local parity-check matrices naturally generalize scalar coefficients to $\F$. This requirement is not so obvious, but indispensable in building good qLDPC codes \cite{PK2021,PK2022Good,leverrier2022quantum,DHLV2022,Dinur2024sheaf}. Our results in the following reveal its significance in the language of sheaf theory. As a preview, given this property and if $X$ is locally acyclic in Definition \ref{def:acyclic_complex}, then the exactness at $C^{i+1}(X_{\geq \sigma},\mathcal{F})$ and hence strong sheaf axiom follow immediately. In this case, $\mathcal{F}$ on $X$ is also said to be locally acyclic in Definition \ref{def:acyclic_sheaf}.

\end{remark}

For any cell $\sigma$, let 
\begin{align}
	\iota_{\sigma}: \mathcal{F}_{\sigma} \rightarrow \prod_{\tau \in U_{\sigma}(t)} \mathcal{F}_{\tau}
\end{align}
be a map defined by
\begin{align}\label{eq:iota}
	\iota_{\sigma}(g) = \prod_{\tau \in U_{\sigma}(t)}\mathcal{F}_{\sigma,\tau}(g)
\end{align}
for any $g \in \mathcal{F}_{\sigma}$. We modify Lemma 6.3 in \cite{lin2024transversalnoncliffordgatesquantum} and get the following theorem.

\begin{theorem}\label{thm:strong_sheaf}
    Suppose the sheaf $\mathcal{F}$ on a cellular complex $X$ satisfies the strong sheaf axiom, then it is isomorphic to the following sheaf $\mathcal{P}$:
    \begin{align}\label{eq:strong_sheaf}
    	\mathcal{P}(U_{\sigma}) \coloneqq  \iota_\sigma \mathcal{F}_{\sigma} 
    	= \{g \in \prod_{\tau\in U_{\sigma}(t)}\mathcal{F}_{\tau}: \forall \tau'\in U_{\sigma}(t-1),\  g|_{U_{\tau'}(t)} \in \iota_{\tau'} \mathcal{F}_{\tau'}\}.
    \end{align}
    The restriction map for $U_{\sigma} \supseteq U_{\pi}$, $\mathcal{P}_{U_{\sigma},U_{\pi}}$ is given by restriction of domain, i.e. for $g \in \mathcal{P}(U_{\sigma})$, $\mathcal{P}_{U_{\sigma},U_{\pi}}(g)=g|_{U_{\pi}}$.
    Then $\iota = \{\iota_\sigma \}$ is the isomorphism between sheaf $\mathcal{F}$ and $\mathcal{P}$ with the following commutative diagram:
    \begin{equation}
    \begin{tikzcd}
    	\mathcal{F}_{\sigma} \arrow[r, "\iota_{\sigma}"] \arrow[d, "\mathcal{F}_{\sigma,\pi}"'] & 
    	\mathcal{P}(U_{\sigma}) \arrow[d, "\mathcal{P}_{U_{\sigma},U_{\pi}}"] \\
    	\mathcal{F}_{\pi} \arrow[r, "\iota_{\pi}"'] & 
    	\mathcal{P}(U_{\pi})
    \end{tikzcd}
    \end{equation}
\end{theorem}
\begin{proof}
The proof follows from Lemma \ref{lemma:strong_sheaf_1} and \ref{lemma:strong_sheaf_2}.

\begin{lemma}\label{lemma:strong_sheaf_1}
    For each $i$-cell $\sigma$, the map $\iota_{\sigma}: \mathcal{F}_{\sigma}\rightarrow \prod_{\tau\in U_{\sigma}(t)} \mathcal{F}_{\tau}$ is injective, and the following diagram commutes
    \begin{equation}
    	\begin{tikzcd}
    		\mathcal{F}_{\sigma} \arrow[r, "\iota_{\sigma}"] \arrow[d, "\mathcal{F}_{\sigma,\pi}"'] & 
    		\prod_{\tau\in U_{\sigma}(t)}\mathcal{F}_{\tau} \arrow[d, "\cdot|_{U_{\pi}(t)}"] \\
    		\mathcal{F}_{\pi} \arrow[r, "\iota_{\pi}"'] & 
    		\prod_{\tau' \in U_{\pi}(t)}\mathcal{F}_{\tau'}
    	\end{tikzcd}
    \end{equation}
\end{lemma}
\begin{proof}[Proof of Lemma \ref{lemma:strong_sheaf_1}]
	The commutativity follows easily by the presheaf condition. We will prove the injectivity by induction. We note that by Eq.~\eqref{eq:iota} if one of $\mathcal{F}_{\sigma,\tau}$s is injective, then $\iota_\sigma$ is injective. However, it happens that none of $\mathcal{F}_{\sigma,\tau}$s is injective in real case, e.g., morphisms from $(t-1)$-cells to $t$-cells defined by local codes in Section \ref{sec:sheaf_code}, but $\iota_\sigma$ can be injective when the parity-check matrices are of full-rank.  
	
	For $i = t$, the injectivity of $\iota_{\sigma}$ is trivial. For $i=t-1$, the injectivity is merely one requirement of strong sheaf axiom. Suppose this is true for all $i \geq k+1$, and now we are going to prove the case $i = k$. Let $\sigma$ be any $i$-cell, by the strong sheaf axiom, we have the embedding
	\begin{align}\label{eq:iota_3}
		\mathcal{F}_{\sigma} \underset{\cong}{\overset{ \delta_L }{\longrightarrow}} 
		\left\{ \prod_{\pi \in X_{\geq \sigma}(i+1)} g_{\pi} : \forall \rho \in X_{\geq \sigma}(i+2),\sum_{\sigma\lessdot\pi\lessdot\rho}\mathcal{F}_{\pi,\rho}(g_{\pi})=0 \right\},
	\end{align}
	where the RHS condition is simply $\delta_L( \prod_{\pi\in X_{\geq \sigma}(i+1)} g_{\pi} ) = 0$.
	Suppose $h = \iota_{\sigma}g$. By the presheaf condition, $h(\tau) = \mathcal{F}_{\pi,\tau}(\mathcal{F}_{\sigma,\pi}(g))$ for any $\pi \in X_{\geq \sigma}(i+1)$ such that $\sigma<\pi<\tau$. Therefore,
	\begin{align}
		& \prod_{\tau\in U_{\pi}} h(\tau) = \prod_{\tau\in U_{\pi}} \mathcal{F}_{\pi,\tau} (\mathcal{F}_{\sigma,\pi}(g))  \in \iota_{\pi}\mathcal{F}_{\pi} \\
		\implies &
		\iota_{\pi}g_{\pi} = \prod_{\tau\in U_{\pi}}\mathcal{F}_{\pi,\tau}(g_{\pi}) = \prod_{\tau\in U_{\pi}} \mathcal{F}_{\pi,\tau} (\mathcal{F}_{\sigma,\pi}(g)) = \prod_{\tau\in U_{\pi}}h(\tau) \label{eq:iota_1}
	\end{align} 
	for some $g_{\pi} \in \mathcal{F}_{\pi}$. 
	
	On the other hand, for each $\rho \in X_{\geq\sigma}(i+2)$, $\sum_{\sigma\lessdot\pi\lessdot\rho}\mathcal{F}_{\pi,\rho}(g_{\pi})\in \mathcal{F}_{\rho}$ (we only sum over $\pi$ here). As a result,
	\begin{align}
	\begin{aligned}
		\iota_{\rho}(\sum_{\sigma\lessdot\pi\lessdot\rho}\mathcal{F}_{\pi,\rho}(g_{\pi}))
		& =\prod_{\tau\in U_{\rho}}\mathcal{F}_{\rho,\tau}(\sum_{\sigma\lessdot\pi\lessdot\rho}\mathcal{F}_{\pi,\rho}(g_{\pi})) \\
		& =\prod_{\tau\in U_{\rho}}\sum_{\sigma\lessdot\pi\lessdot\rho}\mathcal{F}_{\pi,\tau}(g_{\pi}) \\
		& =\sum_{\sigma\lessdot\pi\lessdot\rho} \prod_{\tau\in U_{\rho}} \mathcal{F}_{\pi,\tau}(g_{\pi}) 
		=\prod_{\tau\in U_{\rho}}\sum_{\sigma\lessdot\pi\lessdot\rho}h(\tau)
	\end{aligned}
	\end{align}
	Assume that the above formula equals to zero, then by the induction hypothesis that $\iota_{\rho}$ is injective, we get $\sum_{\sigma\lessdot\pi\lessdot\rho} \mathcal{F}_{\pi,\rho}(g_{\pi})=0$, hence $\prod_{\pi\in X_{\geq \sigma}(i+1)}g_{\pi}\in \im (\delta_L) \cong \mathcal{F}_{\sigma}$. 
	
	Furthermore, by Eq.~\eqref{eq:iota_1} and by the induction hypothesis that $\iota_{\pi}$ is injective, 
	\begin{align}
		g_{\pi} = \mathcal{F}_{\sigma,\pi}(g) \implies
		\prod_{\pi\in X_{\geq \sigma}(i+1)}g_{\pi} = \delta_L (g).
	\end{align}
    Since $\delta_L$ is an embedding, the map $h \mapsto \prod_{\pi\in X_{\geq \sigma}(i+1)}g_{\pi} \mapsto g$ is a left inverse of $\iota_{\sigma}$, which further indicates that $\iota_{\sigma}$ is injective. Therefore, it suffices to show that $\sum_{\sigma\lessdot\pi\lessdot\rho}h(\tau) = 0$, but this is immediate as we have an even number of $\pi$ in the sum.

\end{proof}

\begin{lemma}\label{lemma:strong_sheaf_2}
	For each cell $\sigma\in X(i)$,
	\begin{align}\label{eq:iota_2}
		\mathcal{F}_{\sigma} \cong \iota_{\sigma}(\mathcal{F}_\sigma) = \{g \in \prod_{\tau\in U_{\sigma}(t)}\mathcal{F}_{\tau}: \forall \tau'\in U_{\sigma}(t-1),\  g|_{U_{\tau'}(t)} \in \iota_{\tau'} \mathcal{F}_{\tau'}\}.
	\end{align}
\end{lemma}
\begin{proof}[Proof of Lemma \ref{lemma:strong_sheaf_2}]
	For $i=t$, we have no additional condition in the definition of the set on the RHS of Eq.~\eqref{eq:iota_2}. For $i = t-1$, the statement holds trivially. We still prove by induction. Suppose this is true for all $i \geq k+1$, let us prove for the case $i = k$. We define $\mathcal{P}(U_{\sigma})$ as the set on the RHS of Eq.~\eqref{eq:iota_2}. Then $0 \in \mathcal{P}(U_\sigma)$ and $\mathcal{P}(U_\sigma)$ is a well-defined vector space because the restriction condition is compatible with the linear structure. Given any $\prod_{\tau \in U_{\sigma}(t)} \mathcal{F}_{\sigma,\tau}(g) \in \iota_{\sigma}(\mathcal{F}_\sigma)$, since  $\mathcal{F}_{\sigma,\tau}(g) \in \mathcal{F}_\tau$, when $\tau' < \tau$ for a fixed $\tau'$
	\begin{align}
		\mathcal{F}_{\sigma,\tau}(g) \vert_{U_{\tau'}(t)} = \mathcal{F}_{\tau',\tau} \mathcal{F}_{\sigma,\tau'} (g).
	\end{align}
    Consequently,
    \begin{align}	
    	\left( \prod_{\tau \in U_{\sigma}(t)} \mathcal{F}_{\sigma,\tau}(g) \right) \vert_{U_{\tau'}(t)} 
    	= \prod_{\tau \in U_{\tau'}(t)} \mathcal{F}_{\tau',\tau} \mathcal{F}_{\sigma,\tau'} (g)
    	= \iota_{\tau'}( \mathcal{F}_{\sigma,\tau'} (g) ).
    \end{align} 
	
	We are left to prove $\im(\iota_{\sigma}) \supseteq  \mathcal{P}(U_\sigma)$. Let $\sigma$ be an $i$-cell. Given any $h \in \mathcal{P}(U_\sigma)$ and an $(i+1)$-cell $\pi \gtrdot \sigma$, we define $h_{\pi} \coloneqq  h|_{U_{\pi}(t)}$. Given any $\tau' \in U_{\pi}(t-1) \subseteq U_{\sigma}(t-1)$, we must have $U_{\tau'}(t) \subseteq U_{\pi}(t)$ and thus
	\begin{align}
		h_{\pi} \vert_{U_{\tau'}(t)} 
		= ( h|_{U_{\pi}(t)} ) \vert_{U_{\tau'}(t)}  = h\vert_{U_{\tau'}(t)} \in \iota_{\tau'} \mathcal{F}_{\tau'}. 
	\end{align}
	By induction, $h_{\pi} \in \im(\iota_{\pi})$. Let $g_{\pi} \in \mathcal{F}_{\pi}$ such that $h_{\pi} = \iota_{\pi}g_{\pi}$, i.e. for each $\tau \in U_{\pi}(t)$, $h(\tau) = \mathcal{F}_{\pi,\tau} (g_{\pi})$. By checking Eq.~\eqref{eq:iota_3} and the fact that $\sum_{\sigma\lessdot\pi\lessdot\rho}\mathcal{F}_{\pi,\rho}(g_{\pi}) = 0$, we can show that $\prod_{\pi\in X_{\geq \sigma}(i+1)}g_{\pi} \in \im(\delta_L)$, which means that we can find some $g \in \mathcal{F}_\sigma$ and
	\begin{align}
		g_{\pi} = \mathcal{F}_{\sigma,\pi} g 
		\implies
		h(\tau) = \mathcal{F}_{\pi,\tau} (g_{\pi}) = \mathcal{F}_{\pi,\tau} \mathcal{F}_{\sigma,\pi} g 
		\implies
		h \in \im(\iota_{\sigma}).
	\end{align}
	This finishes the proof.
\end{proof}

By the above two lemmas, $\mathcal{P}$ is a sheaf isomorphic to $\mathcal{F}$. 
\end{proof}

\begin{definition}[Locally acyclic sheaf]\label{def:acyclic_sheaf}
    Suppose $X$ is a $t$ dimensional cell complex with a sheaf $\mathcal{F}$, we say $\mathcal{F}$ is \emph{locally acyclic} if for each $i$-cell $\sigma$, the cohomology $H^j(X_{\geq \sigma},\mathcal{F})=0$ for $i\leq j\leq t-1$, i.e. the following sequence is exact
    \begin{equation}
    	\begin{tikzcd}[column sep=0.7cm]
    		0\arrow[r] & \mathcal{F}_{\sigma} \arrow[r] & \displaystyle 
    		\hspace{-3mm} \prod_{\sigma_{i+1}\in X_{\geq \sigma}(i+1)}\mathcal{F}_{\sigma_{i+1}} \arrow[r]
    		&\displaystyle
    		\hspace{-3mm} \prod_{\sigma_{i+2}\in X_{\geq \sigma}(i+2)}\mathcal{F}_{\sigma_{i+2}}\arrow[r] &  \cdots \arrow[r] & \displaystyle
    		\hspace{-3mm} \prod_{\sigma_{t}\in X_{\geq \sigma}(t)}\mathcal{F}_{\sigma_{t}}.
    	\end{tikzcd}
    \end{equation}
\end{definition}

Combining the definition with Theorem \ref{thm:strong_sheaf}, we get the following easy but important corollary.

\begin{corollary}\label{coro:acyclic_sheaf}
    Every locally acyclic sheaf satisfies the strong sheaf axiom.
\end{corollary}

We also have the notion of local acyclicity for cell complexes.

\begin{definition}[Locally acyclic cell complex]\label{def:acyclic_complex}
	We say that a cell complex $X$ is \emph{locally acyclic} if for any $k$-cell $f$, the homology $H_i(X_{\leq f},\mathbb{F}_q)=0$, for all $0 < i \leq k$. 
\end{definition}

The local acyclicity is a rather weak requirement for cellular complex. For example, every simplicial complex is locally acyclic, so do the cubical complexes because each $k$-cell $f$ inside is homeomorphic to a sphere.

\begin{proposition}\label{prop:d'_exact}
Suppose $X$ is a locally acyclic cell complex (Definition \ref{def:acyclic_complex}), then the map $d':K^{p,q}\rightarrow K^{p+1,q}$ is exact at any $p>0$.  
\end{proposition}
\begin{proof}
    Note that 
    \begin{align}
        C^i(X,\mathcal{F}_k) = \prod_{f\in X(i)} \prod_{u \in X_{\geq f}(k)} \mathcal{F}_u
        = \prod_{u \in X(k)} \prod_{f \in X_{\leq u}(i)} \mathcal{F}_u 
        = \prod_{u\in X(k)} C^i(X_{\leq u},\mathcal{F}_u),
    \end{align}
    where $C^i(X_{\leq u},\mathcal{F}_u)$ is defined over the constant sheaf with values always in $\mathcal{F}_u$. This indicates that the cochain complex $C^\ast(X,\mathcal{F}_k)$ can be decomposed as
    \begin{align}
    	\cdots \rightarrow 
    	\prod_{u\in X(k)} C^i(X_{\leq u},\mathcal{F}_u) \rightarrow 
    	\prod_{u\in X(k)} C^{i+1}(X_{\leq u},\mathcal{F}_u) \rightarrow \cdots
    \end{align}
    We now show that the coboundary map also splits. For $\alpha \in C^i(X,\mathcal{F}_k)$, recall the coboundary map $(d'\alpha)(f',u) = \sum_{f\lessdot f'} \alpha(f,u)$. By viewing $\alpha (-,u)$ as a vector in $C^i(X_{\leq u},\mathcal{F}_u) \subseteq C^i(X,\mathcal{F}_k)$, then the boundary map satisfies
    \begin{align}
    	(d'\alpha(-,u))(f')=\sum_{f\lessdot f'} \alpha(f,u)=\sum_{f\lessdot f'} (\alpha(-,u))(f).
    \end{align}
    Therefore, $d'$ is restricted to the coboundary operator $\delta_{\leq u}$ of the chain complex $C^i(X_{\leq u},\mathcal{F})$ for a fixed $u$. Then $d' = \prod_{u \in X(k)} \delta_{\leq u}$. By Definition \ref{def:acyclic_complex}, $H^i(X_{\leq u}, \mathbb{F}_q) = 0$ (by transposing the chain complex). Since $\mathcal{F}_u$ here is a constant sheaf, $H^i(X_{\leq u},\mathcal{F}_u) = 0$ for all $i>0$ and thus $d'$ is exact.
\end{proof}

We prove the exactness of $d''$ in Proposition \ref{prop:d'_d''_exact} after introducing sheaf and dual sheaf codes.


\subsection{Sheaf codes and dual sheaf codes}\label{sec:sheaf_code}

By Theorem \ref{thm:strong_sheaf}, the sheaf $\mathcal{P}$ is completely defined by the $(t-1)$-cells: for any $i$-cell $\sigma$ with $i \le t - 1$,
\begin{align}
	\mathcal{P}(U_{\sigma}) 
	= \{g \in \prod_{\tau\in U_{\sigma}(t)}\mathcal{F}_{\tau}: \forall \tau'\in U_{\sigma}(t-1),\  g|_{U_{\tau'}(t)} \in \iota_{\tau'} \mathcal{F}_{\tau'}\}
\end{align} 
contains all vectors as long as their restriction to the $(t-1)$-cells satisfy the above condition. In the context of quantum error correction code, we always assume $\mathcal{F}_{\tau} = \mathbb{F}_q$ for each $\tau \in X(t)$. Hence for each $\tau' \in X(t-1)$, $\iota_{\tau'} \mathcal{F}_{\tau'}$ is a linear subspace of $\mathbb{F}_q^{U_{\tau'}(t)}$, and can be viewed as a classical code $\C_{\tau'} \coloneqq  \iota_{\tau'} \mathcal{F}_{\tau'}$. Specifying these subspaces $\{ \C_{\tau'}\}_{\tau' \in X(t-1)}$ is sufficient to define the sheaf $\mathcal{P}$.

\begin{definition}
    Let $X$ be a $t$-dimensional cell complex and let $\{ \C_{\tau'}\}_{\tau' \in X(t-1)}$ be a family of classical codes, where each $\C_{\tau'} \subseteq \mathbb{F}_q^{U_{\tau'}(t)}$ is a linear subspace. We define the sheaf $\mathcal{F} = \mathcal{F}[\{ \C_{\tau'}\}_{\tau' \in X(t-1)}]$ generated by $\{ \C_{\tau'}\}_{\tau' \in X(t-1)}$ as
    \begin{align}
        \mathcal{F}(U_{\sigma}) = \{c \in \mathbb{F}_q^{U_{\sigma}(t)} : \forall \tau' \in U_{\sigma}(t-1),\  c|_{U_{\tau'}(t)} \in \C_{\tau'}\}.
    \end{align}
    The restriction maps $\mathcal{F}_{\sigma,\sigma'}$ are given by restricting domains. 
\end{definition}

In classical coding theory, for a code $\C \subseteq \mathbb{F}_q^n$, we define its \emph{dual code} as
\begin{align}
	\C^{\perp} \coloneqq  \{y \in \mathbb{F}_q^n: \langle y,x \rangle = 0,\forall x \in \C \}.
\end{align}
Suppose $\C$ has a parity check matrix $h$, and $\C^{\perp}$ has a parity check matrix $h^{\perp}$, then we have the following relationship
\begin{align}
	\C = \ker h = \im(h^{\perp})^T, \quad
	\C^{\perp} = \ker h^{\perp} = \im h^T.
\end{align}
We can also define the dual sheaf of a sheaf generated by classical codes.

\begin{definition}[Dual sheaf]\label{def:sheaf_dual_sheaf}
    For a sheaf $\mathcal{F}$ generated by $\{\C_{\tau'}\}_{\tau' \in X(t-1)}$, we define the \emph{dual sheaf} $\mathcal{F}^{\perp}$ to be the sheaf generated by $\{\C_{\tau'}^{\perp}\}_{\tau' \in X(t-1)}$.
\end{definition}


Given any $i$-cell $\sigma$, by definition, $\mathcal{F}^{\perp}_{\sigma} = \{c \in \mathbb{F}_q^{U_{\sigma}(t)} : \forall \tau'\in U_{\sigma}(t-1),\  c|_{U{\tau'}(t)} \in \C_{\tau'}\}$ is a linear subspace of 
\begin{align}
	\mathbb{F}_q^{U_{\sigma}(t)} = \prod_{\tau \in X_{\geq \sigma}(t)} \mathcal{F}_\tau = \mathcal{F}_{t,\sigma}.
\end{align}
For any $c \in \mathcal{F}_{\sigma}^{\perp}$, the restriction $\mathcal{F}^{\perp}_{\sigma,\tau}$ maps $c$ to one of its components in $\mathcal{F}_\tau = \mathbb{F}_q$. Then $\prod_{\tau \in X_{\geq \sigma }(t)} \mathcal{F}^{\perp}_{\sigma,\tau} c$ is simply $c$ itself. 

Let us define a map $h''$:
\begin{align}
	C^i(X,\mathcal{F}^{\perp}) \xrightarrow{h''} C^i(X,\mathcal{F}_t)
\end{align}
by mapping $c \in \mathcal{F}_{\sigma}^{\perp}$ to $\prod_{\tau \in X_{\geq \sigma }(t)} \mathcal{F}^{\perp}_{\sigma,\tau} c \in \mathcal{F}_{t,\sigma}$. Then the following important properties follow.

\begin{proposition}\label{prop:h_exact}
    Let $\mathcal{F}$ and $\mathcal{F}^\perp$ be generated by classical codes and their dual, respectively. We have an exact sequence for each $0<i\leq t-1$
    \begin{equation}\label{eq:h_exact}
      \begin{tikzcd}
          0 \arrow[r] & C^i(X,\mathcal{F}^\perp) \arrow[r,"h''"] & C^i(X,\mathcal{F}_t) \arrow[r,"d''"] & C^i(X,\mathcal{F}_{t-1}),
      \end{tikzcd}  
    \end{equation}
    For $i = t$, there is a trivial isomorphism $C^t(X,\mathcal{F}^{\perp}) \cong \mathbb{F}_q^{X(t)}\cong C^t(X,\mathcal{F}_t)$.
\end{proposition}
\begin{proof}
	By definition, for each $\mathcal{F}_{\sigma}^{\perp}$, the map $h'': \mathcal{F}^{\perp}_{\sigma} \hookrightarrow \mathcal{F}_{t,\sigma}$ is an inclusion.  
    Since \eqref{eq:h_exact} is the direct sum of the following sequences
    \begin{equation}
    \begin{tikzcd}
    	0 \arrow[r] & \mathcal{F}_{\sigma}^{\perp} \arrow[r,"h''"] & \mathcal{F}_{t,\sigma} \arrow[r,"\partial_L"] & \mathcal{F}_{t-1,\sigma}
    \end{tikzcd}  	
    \end{equation}
    As a result, $h'': C^i(X,\mathcal{F}^{\perp}) \hookrightarrow C^i(X,\mathcal{F}_t)$ is injective and we only need to focus on proving the above sequence is exact. Recall that by definition
    \begin{align}
    	\ker \partial_L=\{c\in \mathbb{F}_q^{U_{\sigma}(t)}:\forall\ \tau'\in U_{\sigma}(t-1),\ \sum_{\tau\gtrdot\tau'} \mathcal{F}_{\tau',\tau}^T(c(\tau)) = 0\}.
    \end{align}
    On the RHS, we sum over $\tau$ and $\sum_{\tau\gtrdot\tau'} \mathcal{F}_{\tau',\tau}^T(c(\tau)) = 0$ is equivalent to say that $\iota_{\tau'}^T (c|_{U_{\tau'}(t)}) = 0$. Furthermore, for any $y \in \mathcal{F}_{\tau'}$,
    \begin{align}
    	\langle y,\iota_{\tau'}^T(c|_{U_{\tau'}(t)}) \rangle
    	= \langle \iota_{\tau'}y,(c|_{U_{\tau'}(t)}) \rangle = 0,
    \end{align}
    which means that $c|_{U_{\tau'}(t)} \in \C_{\tau'}^{\perp}$. Therefore,
    \begin{align}
    	\ker \partial_L = \{c\in \mathbb{F}_q^{U_\sigma(t)}:\forall\ \tau'\in U_{\sigma}(t-1),\ c|_{U_{\tau'}(t)}\in \C_{\tau'}^{\perp}\} = \mathcal{F}_{\sigma}^{\perp}
    \end{align}
    and proof is done.
\end{proof}

\begin{proposition}\label{prop:h_naturality}
    The $h''$ is a cochain map, i.e., it commutes with the coboundary operators and we have the following commutative diagram:
    \begin{equation}
    \begin{tikzcd}
    	0 \arrow[r]  & C^i(X,\mathcal{F}^{\perp}) \arrow[r,"h'' "] \arrow[d,"\delta^{\perp}"] & C^i(X,\mathcal{F}_t) \arrow[r,"d'' "]\arrow[d,"d' "]& C^i(X,\mathcal{F}_{t-1})  \arrow[d,"d' "]  \\
    	0 \arrow[r]  & C^{i+1}(X,\mathcal{F}^{\perp}) \arrow[r,"h'' "] & C^{i+1}(X,\mathcal{F}_t) \arrow[r,"d'' "]& C^{i+1}(X,\mathcal{F}_{t-1}). 
    \end{tikzcd}
    \end{equation}
\end{proposition}
\begin{proof}
We only need to check for each $i$-cell $\sigma$, the following diagram is commutative:
\begin{equation}
    \begin{tikzcd}
        \mathcal{F}^{\perp}_{\sigma} \arrow[r,"h'' "] \arrow[d,"\delta^{\perp}"] & C^t(X_{\geq\sigma},\mathcal{F})\arrow[d,"d'"] = \mathcal{F}_{t,\sigma} \\ \prod_{\sigma'\gtrdot\sigma} \mathcal{F}^{\perp}_{\sigma'}\arrow[r,"h'' "] & \prod_{\sigma'\gtrdot\sigma} C^t(X_{\geq\sigma'},\mathcal{F})
    \end{tikzcd}
\end{equation}
Let $c\in \mathcal{F}^{\perp}_{\sigma}$, then 
\begin{align}
	h''c = \prod_{\tau \in X_{\geq \sigma }(t)} \mathcal{F}^{\perp}_{\sigma,\tau} c, \quad
	d' h''c = \prod_{\sigma'\gtrdot\sigma} \prod_{\tau \in X_{\geq \sigma }(t)} \mathcal{F}^{\perp}_{\sigma,\tau} c.
\end{align}
On the other hand, 
\begin{align}
	\delta^{\perp} c = \prod_{\sigma'\gtrdot \sigma} \mathcal{F}_{\sigma',\sigma}^\perp c, \quad 
	h'' \delta^{\perp}c = \prod_{\tau \in X_{\geq \sigma'}(t)}  \prod_{\sigma'\gtrdot\sigma} \mathcal{F}_{\sigma',\tau}^\perp \mathcal{F}_{\sigma,\sigma'}^\perp c  
	= d' h'' c.
\end{align}
Actually, one can also check that $h$ is a sheaf morphism between $\mathcal{F}^{\perp}$ and $\mathcal{F}_t$ by direct calculation.
\end{proof}

Analogously to Proposition \ref{prop:h_exact} and \ref{prop:h_naturality}, we also have the following results for the chain map.

\begin{proposition}\label{prop:eta_exact}
    Suppose $\mathcal{F}$ is a sheaf generated by classical codes, then we have the following exact sequence for each $i \in \mathbb{Z}$:
    \begin{equation}\label{eq:eta_exact}
    \begin{tikzcd}
    	0 \arrow[r] & C_i(X,\mathcal{F}) \arrow[r,"\eta'"] & C^0(X,\mathcal{F}_i) \arrow[r,"d'"] & C^1(X,\mathcal{F}_{i}),
    \end{tikzcd} 	
    \end{equation}
    where $\eta'$ is defined by $x\mapsto\prod_{\sigma\in X(0)}x|_{X_{\geq\sigma}(i)}$. It is a chain map and the following diagram commutes:
    \begin{equation}\label{eq:eta_natural}
    	\begin{tikzcd}
    		0 \arrow[r]  & C_i(X,\mathcal{F}) \arrow[r,"\eta'"] \arrow[d,"\partial"] & C^0(X,\mathcal{F}_i) \arrow[r,"d'"]\arrow[d,"d''"]& C^1(X,\mathcal{F}_i)  \arrow[d,"d''"]  \\
    		0 \arrow[r]  & C_{i-1}(X,\mathcal{F}) \arrow[r,"\eta'"] & C^0(X,\mathcal{F}_{i-1}) \arrow[r,"d'"] & C^1(X,\mathcal{F}_{i-1}).  
    	\end{tikzcd}
    \end{equation}
\end{proposition}
\begin{proof}
	In our case, each 1-cell $\rho$ contains exactly two 0-cell $\sigma,\sigma'\lessdot \rho$ and hence $X_{\geq\rho}(i) = X_{\geq\sigma}(i)\cap X_{\geq\sigma'}(i)$, which further indicates \eqref{eq:eta_exact} above is exact:
    \begin{equation}
    \begin{tikzcd}
    	0 \arrow[r] & C_i(X,\mathcal{F}) \arrow[r,"\eta'"] & C^0(X,\mathcal{F}_i) \arrow[r,"d'"] & C^1(X,\mathcal{F}_{i}) \\ [-4ex]
    	& x\arrow[r, mapsto] & \displaystyle\prod_{\sigma\in X(0)}x|_{X_{\geq\sigma}(i)} \\ [-4ex]
    	& & \displaystyle \prod_{\sigma \in X(0)}x_{\sigma} \arrow[r,mapsto] &\displaystyle \prod_{\rho \in X(0), \sigma,\sigma' < \rho} (x_{\sigma}|_{X_{\geq \rho}(i)} - x_{\sigma'}|_{X_{\geq \rho}(i)} ),
    \end{tikzcd}  	
    \end{equation}
    where $x_\sigma$ in the last line is from $\mathcal{F}_{i,\sigma}$.  Any $i$-cell contains both $\sigma$ and $\sigma'$ must also contain $\rho > \sigma,\sigma'$. Whenever $\prod_{\sigma\in X(0)}x_{\sigma}$ is mapped to zero, agreements of $x_{\sigma}$ and $x_{\sigma'}$ in their intersection at each 1-cell $\rho$ ensures that we can always glue them together to some $x \in C_i(X,\mathcal{F})$.
     
    To prove that $\eta'$ is a chain map, Note that for $x\in C_i(X,\mathcal{F})$, $\eta' (x)=\prod_{\sigma\in X(0)}x|_{X_{\geq\sigma}(i)}$, $d''(\eta' (x))=\prod_{\sigma\in X(0)}\partial_L(x|_{X_{\geq\sigma}(i)})$. On the other hand, $\eta'(\partial(x))=\prod_{\sigma\in X(0)}(\partial x)|_{X_{\geq\sigma}(i-1)}$. For each $\tau\in X_{\geq\sigma}(i-1)$
    \begin{align}
    	(\partial x)|_{X_{\geq \sigma}(i-1)}(\tau) 
    	= \sum_{\tau'\gtrdot\tau} \mathcal{F}_{\tau,\tau'}^T(x(\tau')),
    \end{align}
    while 
    \begin{align}
    	(\partial_L(x|_{X_{\geq\sigma}(i)}))(\tau) 
    	= \sum_{\tau'\gtrdot\tau} \mathcal{F}_{\tau,\tau'}^T(x|_{X_{\geq\sigma}(i)}(\tau'))
    	= \sum_{\tau'\gtrdot\tau} \mathcal{F}_{\tau,\tau'}^T(x(\tau')).
    \end{align}
    which verifies $\eta' \partial = d''\eta'$, and $\eta'$ is a chain map.
\end{proof}

As a summary of Proposition \ref{prop:d'_exact}, \ref{prop:h_exact},  \ref{prop:h_naturality} and \ref{prop:eta_exact}, we conclude that

\begin{proposition}\label{prop:d'_d''_exact}
Let $\mathcal{F}$ and $\mathcal{F}^\perp$ locally acyclic sheaves generated by classical codes $\{C_{\sigma}\}_{\sigma\in X(t-1)}$ and their dual, then we have the following exact sequences:
\begin{equation}\label{eq:eta_d'}
	\begin{tikzcd}[column sep=0.6cm, row sep=0.7cm]
		0 \arrow[r] & C_i(X,\mathcal{F}) \arrow[r,"\eta'"] & C^0(X,\mathcal{F}_i) \arrow[r,"d' "] & C^1(X,\mathcal{F}_{i}) \arrow[r,"d' "] & \cdots \arrow[r,"d' "] & C^t(X, \mathcal{F}_i) \arrow[r] & 0
	\end{tikzcd}.
\end{equation}
and
\begin{equation}\label{eq:h''_d''}
	\begin{tikzcd}[column sep=0.6cm, row sep=0.7cm]
		0 \arrow[r] & C^p(X, \mathcal{F}^\perp) \arrow[r,"h'' "]
		& C^p(X, \mathcal{F}_t) \arrow[r,"d'' "] & C^p(X, \mathcal{F}_{t-1}) \arrow[r,"d'' "] & \cdots \arrow[r,"d'' "] & C^p(X, \mathcal{F}_0) \arrow[r] & 0
	\end{tikzcd}.
\end{equation}
\end{proposition}

There is a nice and important observation in the following. It is essential to prove our main result (Theorem \ref{thm:Poincaré_duality}), where the exactness is crucial for the ``local to global" method.

\begin{proposition}\label{prop:flabby_resolution}
    There is an exact sequence of sheaves
\begin{equation}
	\begin{tikzcd}
		0 \arrow[r] & \mathcal{F}^{\perp}\arrow[r,"h'' "]
		& \mathcal{F}_t \arrow[r,"\partial_L"] & \mathcal{F}_{t-1}\arrow[r,"\partial_L"] & \cdots \arrow[r,"\partial_L"] & \mathcal{F}_1 \arrow[r,"\partial_L"] & \mathcal{F}_0\arrow[r] & 0
	\end{tikzcd}.
\end{equation}
\end{proposition}

\begin{proof}
	Since $h''$ and $d''$ are sheaf morphisms, it suffices to check the exactness on stalks, i.e., for each $k$-cell $\rho \in X$, the following sequence is exact
    \begin{equation}
    \begin{tikzcd}
    	0 \arrow[r] & \mathcal{F}^{\perp}_{\rho} \arrow[r,"h'' "] 
    	& \mathcal{F}_{t,\rho} \arrow[r,"\partial_L"] & \mathcal{F}_{t-1,\rho} \arrow[r,"\partial_L"] & \cdots \arrow[r,"\partial_L"] & \mathcal{F}_{k+1,\rho} \arrow[r,"\partial_L"] & \mathcal{F}_{k,\rho} \arrow[r] & 0
    \end{tikzcd}.
    \end{equation}
    The exactness at $\mathcal{F}_{\sigma}^{\perp}$ and $\mathcal{F}_t$ is guaranteed by Proposition \ref{prop:h_exact}. While the rest is guaranteed by the local acyclicity and universal coefficient theorem.
\end{proof}


\subsection{Proof of the duality}\label{sec:duality-proof}

\begin{theorem}[Poincaré duality on sheaf codes]\label{thm:Poincaré_duality}
	Suppose $X$ is a $t$-dimensional sparse cell complex with a locally acyclic sheaf $\mathcal{F}$, then we have a duality between $\mathcal{F}$ and $\mathcal{F}^\perp$ in terms of logical qubits, code distances, (co)cycle expansions and decoders. To be precise, for any $0 \leq i \leq t$, there is an isomorphism:
	\begin{align}
		D: H^i(X,\mathcal{F}^\perp)\overset{\cong}{\longrightarrow} H_{t-i}(X,\mathcal{F}).
	\end{align}
	The (co)systolic distance are bounded by each other linearly:
	\begin{align}
		\mu_\partial(i)=\Theta(\mu_{\delta^\perp}(t-i)).
	\end{align}
	There exists a decoder $\mathcal{O}^\perp_{t-i}$ for the cochain $C^{t-i}(X,\mathcal{F}^\perp)$ if and only if there exists a decoder $\mathcal{O}_i$ for the chain $C_i(X,\mathcal{F})$. Their running time differs by a constant, and so is the decoding radius
	\begin{align}
		R(\mathcal{O}_i) = \Theta(R(\mathcal{O}^\perp_{t-i})).
	\end{align}
	For the (co)cycle expansion, we have
	\begin{align}
		\varepsilon_{\partial}(i) = 
		\Theta(\varepsilon_{\delta^{\perp}}(t-i)).
	\end{align}
	Furthermore, when $X$ can be subdivided into a simplicial complex, then the isomorphism $D$ is given by the cap product (see Definition~\ref{def:cap-III}) with a cycle $[X] \in C_t(X,\mathcal{F}^\perp \!\! \otimes \mathcal{F})$.
	\begin{align}
		D[\alpha] = [\alpha] \frown [X].
	\end{align}
\end{theorem}

We divide the proof into the following parts.

\subsubsection*{Duality of (co)homologies/logical qubits:}

\noindent\emph{\textbf{Method 1: Flabby resolution}}. By Proposition \ref{prop:flabby_resolution}, we have an exact sequence of sheaves as follows
\begin{equation}
	\begin{tikzcd}
		0 \arrow[r] & \mathcal{F}^{\perp}\arrow[r,"h'' "]
		& \mathcal{F}_t \arrow[r,"\partial_L"] & \mathcal{F}_{t-1}\arrow[r,"\partial_L"] & \cdots \arrow[r,"\partial_L"] & \mathcal{F}_1 \arrow[r,"\partial_L"] & \mathcal{F}_0\arrow[r] & 0.
	\end{tikzcd}
\end{equation}
Note that since each sheaf $\F_i$ is flabby (see Definition \ref{def:flabby}), this is a flabby resolution, and hence can be used to compute sheaf cohomology, i.e. $(R^p\Gamma)(\F^\perp)$ is isomorphic to the $p$-th cohomology of the following cochain
\begin{equation}
	\begin{tikzcd}
		0 \arrow[r] & \Gamma(X,\mathcal{F}_t) \arrow[r] & \Gamma(X,\mathcal{F}_{t-1})\arrow[r] & \cdots \arrow[r] & \Gamma(X,\mathcal{F}_1) \arrow[r] & \Gamma(X,\mathcal{F}_0)\arrow[r] & 0.
	\end{tikzcd}
\end{equation}
Note that this is exactly the following chain complex
\begin{equation}
	\begin{tikzcd}[column sep=0.7cm]
		0 \arrow[r] & C_t(X,\mathcal{F}) \arrow[r] & C_{t-1}(X,\mathcal{F})\arrow[r] & \cdots \arrow[r] & C_{1}(X,\mathcal{F}) \arrow[r] & C_0(X,\mathcal{F})\arrow[r] & 0.
	\end{tikzcd}
\end{equation}
Therefore, we proved the generalized Poincaré duality between sheaf codes $H^i(X,\F^\perp)\cong H_{t-i}(X,\F)$.

\noindent\emph{\textbf{Method 2: Spectral sequence}}. There is also an alternative proof using spectral sequence proposed in \cite{lin2024transversalnoncliffordgatesquantum}. Let us rewrite the double complex $K^{p,q}$ defined in the beginning of Section \ref{sec:strong_sheaf} as $E_0^{p,q}$. And we can draw a table as follows.

\begin{equation*}
	E_0= \centering
	\resizebox{0.9\linewidth}{!}{%
		\begin{tikzpicture}[xscale={3},baseline={(X.base)}]
			\draw (0,4)node{$C^0(X,\mathcal{F}_0)$} (1,4)node{0} (2,4)node{0} (3,4)node{$\cdots$} (4,4)node{0};
			\draw (0,3)node{$C^0(X,\mathcal{F}_1)$} (1,3)node{$C^1(X,\mathcal{F}_1)$} (2,3)node{0} (3,3)node{$\cdots$} (4,3)node{0};
			\draw (0,2)node(X){$C^0(X,\mathcal{F}_2)$} (1,2)node{$C^1(X,\mathcal{F}_2)$} (2,2)node{$C^2(X,\mathcal{F}_2)$} (3,2)node{$\cdots$} (4,2)node{0};
			\draw (0,1)node{$\vdots$} (1,1)node{$\vdots$} (2,1)node{$\vdots$} (3,1)node{$\vdots$} (4,1)node{$\vdots$};
			\draw (0,0)node{$C^0(X,\mathcal{F}_t)$} (1,0)node{$C^1(X,\mathcal{F}_t)$} (2,0)node{$C^2(X,\mathcal{F}_t)$} (3,0)node{$\cdots$} (4,0)node{$C^t(X,\mathcal{F}_t)$};
			\draw[->] (-0.6,-0.5)--(4.5,-0.5)node[anchor=north]{$p$};
			\draw[->] (-0.6,-0.5)--(-0.6,4.5)node[anchor=east]{$q$};
		\end{tikzpicture}
	}
\end{equation*}

Then we can calculate the spectral sequence of $E_0$. For a more comprehensive guide on spectral sequences for double complexes, refer to \cite[Theorem 14.14]{Bott:1982xhp}. We can calculate the page-1 of type-I spectral sequence by ${}^\text{I}E_1=H_{d''}E_0$. By Proposition \ref{prop:h_exact} and acyclicity of $\mathcal{F}$, we know that 
\begin{equation}
	{}^\text{I} E_1^{p,q} =
	\begin{cases}
		C^{p}(X,\mathcal{F}^\perp) & ,q=0 \\
		0  & ,q\neq 0
	\end{cases}.
\end{equation}
Since $h''$ is a chain map, we know that the second page of type-I spectral sequence ${}^\text{I} E_2=H_{d'}{}^\text{I} E_1$ is
\begin{equation}
	{}^\text{I} E_2^{p,q} =
	\begin{cases}
		H^p(X,\mathcal{F}) & ,q=0 \\
		0  & ,q\neq 0
	\end{cases}.
\end{equation}
The differential $d_2:{}^\text{I}E_2^{p,q}\rightarrow {}^\text{I}E_2^{p+2,q+1}$ is zero. So do all differentials $d_r=0$ for $r \geq 2$. Therefore, the spectral sequence converges at the second page, i.e. ${}^{\text{I}}E_\infty={}^{\text{I}}E_2$

Similarly, we can calculate the page-1 of type-II spectral sequence by ${}^\text{II}E_1=H_{d'}E_0$. By Proposition \ref{prop:d'_exact} and \ref{prop:eta_exact}, we know that 
\begin{equation}
	{}^\text{II} E_1^{p,q} =
	\begin{cases}
		C_{t-q}(X,\mathcal{F}) & ,p=0 \\
		0  & ,p\neq 0
	\end{cases}.
\end{equation}
Note that since $\eta'$ is a chain map, the second page of the type-II spectral sequence ${}^\text{II}E_2=H_{d''}{}^\text{II}E_1$ is
\begin{align}
	{}^\text{II} E_2^{p,q} =
	\begin{cases}
		H_{t-q}(X,\mathcal{F}) & ,p=0 \\
		0  & ,p\neq 0
	\end{cases}.
\end{align}
Also the spectral sequence converges at the second page. Note that in our case the type-I and type-II spectral sequence both converge to the cohomology of total complex, i.e.
\begin{align}
	H^n(\text{Tot}(K))=\bigoplus_{p+q=n}{}^\text{I} E_{\infty}^{p,q}=\bigoplus_{p+q=n}{}^\text{II} E_{\infty}^{p,q}.
\end{align}
Therefore, we conclude with $H^p(X,\mathcal{F}^\perp)\cong H_{t-p}(X,\mathcal{F})$.

\subsubsection*{Explicit isomorphism of Poincaré duality:}

It is natural to seek an explicit isomorphism between the cohomology and homology. The technique of spectral sequence cannot help us find the answer, but we can use the method of diagram chasing from homological algebra to work it out. This method is also used in \cite{Dinur2024sheaf} to bound the code distance and (co)cycle expansion. To be precise, by Proposition \ref{prop:h_naturality}, we have the following commutative diagram:
\[
\begin{tikzcd}[column sep=2.2em, row sep=1.8em,scale=0.9, transform shape]
      C_0(X,\F) \arrow[r,"\eta'"] 
    & C^0(X,\mathcal{F}_0) \arrow[r,"d'"] 
    & 0  \arrow[r,"d'"] 
    & 0  \arrow[r,"d'"] 
    & \cdots  \arrow[r,"d'"] 
    & 0 \\
      C_1(X,\F) \arrow[r,"\eta'"]  \arrow[u,"\partial"]
    & C^0(X,\mathcal{F}_1) \arrow[u,"d''"] \arrow[r,"d'"] 
    & C^1(X,\mathcal{F}_1) \arrow[u,"d''"]  \arrow[r,"d'"] 
    & 0  \arrow[r,"d'"]  \arrow[u,"d''"] 
    & \cdots  \arrow[r,"d'"] 
    & 0  \arrow[u,"d''"]\\
      C_2(X,\F) \arrow[r,"\eta'"] \arrow[u,"\partial"]
    & C^0(X,\mathcal{F}_2) \arrow[u,"d''"] \arrow[r,"d'"] 
    & C^1(X,\mathcal{F}_2) \arrow[u,"d''"] \arrow[r,"d'"] 
    & C^2(X,\mathcal{F}_2) \arrow[u,"d''"]  \arrow[r,"d'"] 
    & \cdots  \arrow[r,"d'"] 
    & 0 \arrow[u,"d''"]\\
       \vdots  \arrow[u,"\partial"]
    &\vdots \arrow[u,"d''"] 
    & \vdots \arrow[u,"d''"] 
    & \vdots \arrow[u,"d''"] 
    & \ddots 
    & \vdots \arrow[u,"d''"] \\
   C_t(X,\mathcal{F}) \arrow[r,"\eta'"] \arrow[u,"\partial"]
    &C^0(X,\mathcal{F}_t) \arrow[u,"d''"] \arrow[r,"d'"] 
    & C^1(X,\mathcal{F}_t) \arrow[u,"d''"] \arrow[r,"d'"] 
    & C^2(X,\mathcal{F}_t) \arrow[u,"d''"] \arrow[r,"d'"] 
    & \cdots \arrow[r,"d'"] 
    & C^{t}(X,\mathcal{F}_t) \arrow[u,"d''"]\\
    0\arrow[u]\arrow[r]
    &C^0(X,\F^\perp) \arrow[u,"h''"]\arrow[r,"\delta^\perp"]
    &C^1(X,\F^\perp) \arrow[u,"h''"]\arrow[r,"\delta^\perp"]
    &C^2(X,\F^\perp) \arrow[u,"h''"]\arrow[r,"\delta^\perp"]
    &\cdots\arrow[r,"\delta^\perp"]
    &C^t(X,\F^\perp) \arrow[u,"h''"]
\end{tikzcd}
\]
Given any cocycle $\alpha \in C^p(X,\mathcal{F}^\perp)$, let $\alpha^{p,t} \coloneqq  h''\alpha \in C^p(X,\mathcal{F}_t)$. Since $h''$ is a chain map, we have
\begin{align}\label{eq:zig-zag_1}
	d' \alpha^{p,t} = d' h'' \alpha = h'' \delta^\perp \alpha = 0.
\end{align}
By exactness of $d'$, we may find $\alpha^{p-1,t} \in C^{p-1}(X,\mathcal{F}_t)$ such that $\alpha^{p,t} = d'\alpha^{p-1,t}$. Now, let $\alpha^{p-1,t-1} \coloneqq  d''\alpha^{p-1,t}$,
\begin{align}\label{eq:zig-zag_2}
	d' \alpha^{p-1,t-1} = d' d''\alpha^{p-1,t} = d'' d' \alpha^{p-1,t} = d'' \alpha^{p,t} = d'' h'' \alpha = 0
\end{align}
indicates that there exists some $\alpha^{p-2,t-1} \in C^{p-2}(X,\mathcal{F}_{t-1})$ such that $d' \alpha^{p-2,t-1} = \alpha^{p-1,t-1}$. We set $\alpha^{p-2,t-2} \coloneqq  d''\alpha^{p-2,t-1}$. By repeating this "zig-zag" argument, we obtain $\alpha^{k,t-p+k+1} \in C^k(X,\mathcal{F}_{t-p+k+1})$ for $0 \leq k \leq p-1$ and $\alpha^{k,t-p+k} \in C^k(X,\mathcal{F}_{t-p+k})$ for $0 \leq k \leq p$ such that 
\begin{equation}\label{eq:zig-zag_3}
     d'\alpha^{k-1,t-p+k} = \alpha^{k,t-p+k} = d''\alpha^{k,t-p+k+1}, \quad \alpha^{p,t} = h'' \alpha.
\end{equation}
In the last step, after we find $\alpha^{0,t-p+1}$ and define $\alpha^{0,t-p} \coloneqq  d'' \alpha^{0,t-p+1}$, we have $d'\alpha^{0,t-p} = 0$ like Eq.~\eqref{eq:zig-zag_1} and \eqref{eq:zig-zag_2}. By Proposition \ref{prop:eta_exact}, the exactness of $\eta'$ ensures that there exists some $D\alpha \in C_{t-p}(X,\mathcal{F})$ such that $\eta' (D\alpha) = \alpha^{0,t-p}$. We may call these $\alpha^{*,*}$ as \emph{explaining sequence} in the language of \cite{Dinur2024sheaf}. At first glance, $D\alpha$ depends on the choice of the sequence, but we are going to show that the homology class $[D\alpha]$ is uniquely determined by the cohomology class $[\alpha]$ and hence induces a map
\begin{align}
	D: H^p(X,\mathcal{F}^\perp)\longrightarrow H_{t-p}(X,\mathcal{F}),
\end{align}
by $D[\alpha]\coloneqq [D\alpha]$. 

We first show that the homology class $[D\alpha]$ is independent of the the choices of $\alpha^{*,*}$ when the initial cocycle $\alpha$ is fixed. Recall that $d'\alpha^{p-1,t} = \alpha^{p,t}$. Then by the exactness of $d'$, the only freedom on $\alpha^{p-1,t}$ is to add a coboundary $d'\beta^{p-2,t}$ where $\beta^{p-2,t} \in C^{p-2}(X,\mathcal{F}_t)$. Let
\begin{align}\label{eq:zig-zag_4}
	\tilde{\alpha}^{p-1,t} & \coloneqq  \alpha^{p-1,t} + d'\beta^{p-2.t}, \\
	\tilde{\alpha}^{p-1,t-1} & \coloneqq  d''(\alpha^{p-1,t} + d'\beta^{p-2,t}) = \alpha^{p-1,t-1} + d' d''\beta^{p-2,t}.
\end{align} 
Obviously, we still have $d'\tilde{\alpha}^{p-1,t-1} = 0$ like Eq.~\eqref{eq:zig-zag_2}. To find an element that can be mapped to it via $d'$, since $d'(\alpha^{p-2,t-1} + d''\beta^{p-2,t}) = \tilde{\alpha}^{p-1,t-1}$, we define
\begin{align}
	\tilde{\alpha}^{p-2,t-1} & \coloneqq  \alpha^{p-2,t-1} + d''\beta^{p-2,t} + d'\beta^{p-3,t-1}, \\
	\tilde{\alpha}^{p-2,t-2} & \coloneqq  d''\tilde{\alpha}^{p-2,t-1} = \alpha^{p-2,t-2} + d' d''\beta^{p-3,t-1}
\end{align}
as before by taking a boundary operator $d'\beta^{p-3,t-1}$ as the only freedom. Then we can recursively define $\tilde{\alpha}^{*,*}$'s with max freedom. By \eqref{eq:eta_exact} in Lemma \ref{prop:eta_exact}, it ends up with
\begin{align}
	\tilde{\alpha}^{1,t-p+1} & \coloneqq  d''\alpha^{1,t-p+2} + d' d''\beta^{0,t-p+2}, \\
	\tilde{\alpha}^{0,t-p} & \coloneqq  d'' \alpha^{0,t-p+1} + d'' \eta' \beta = \eta'( D\alpha + \partial \beta), \label{eq:zig-zag_5}
\end{align}
where $\partial: C_{t-p+2}(X,\mathcal{F}) \to C_{t-p+1}(X,\mathcal{F})$ and $\beta \in C_{t-p+2}(X,\mathcal{F})$. 
Since $\eta'$ is injective, the element constructed from $\tilde{\alpha}^{*,*}$ is $\tilde{D}\alpha = D\alpha + \partial\beta$. This verifies that $[D\alpha]$ is indeed fixed when $\alpha$ is given.

Next we show that $[D\alpha]$ is well-defined with respect to the cohomology class $[\alpha]$. Suppose $\alpha'$ is homologous to $\alpha$, i.e., there exist $\gamma\in C^{p-1}(X, \mathcal{F}^{\perp})$ such that $\alpha' = \alpha + \delta^\perp \gamma$. Then we construct two explaining sequences $\alpha^{*,*}$ and $\alpha'^{\,*,*}$ via $\alpha$ and $\alpha'$, respectively. In sequences, we must have
\begin{align}
	\alpha'^{\,p-1,t} = \alpha^{p-1,t} + h'' \gamma + d'\beta^{p-2,t}
\end{align}
for some coboundary $d'\beta^{p-2,t}$ as in Eq.~\eqref{eq:zig-zag_4}. Applying $d''$ to define $\alpha'^{\,p-1,t}$ and $\alpha^{p-1,t-1}$, $d'' h'' \gamma = h'' \delta^{\perp} \gamma$ vanishes. Together with Eq.~\eqref{eq:zig-zag_5}, we must have $[D\alpha] = [D\alpha']$ and the map $D: H^p(X,\mathcal{F}^\perp) \to H_{t-p}(X,\mathcal{F})$ is indeed well-defined.

Finally, let us prove that $D$ is an isomorphism. We can easily construct the inverse map of $D$. Note that for each cycle $x \in C_{t-p}(X,\mathcal{F})$, we can construct a explaining sequence $x^{*,*}$ similarly (by using the exactness of $d''$ in Proposition \ref{prop:d'_d''_exact}) such that there exist $x^{k,t-p+k} \in C^k(X,\mathcal{F}_{t-p+k})$ for $0\leq k\leq p$ and $x^{k,t-p+k+1} \in C^k(X,\mathcal{F}_{t-p+k+1})$ for $0\leq k\leq p-1$ such that 
\begin{equation}\label{eq:zig-zag_6}
     d'x^{k-1,t-p+k} = x^{k,t-p+k} = d''x^{k,t-p+k+1}, \quad x^{0,t-p} = \eta' x.
\end{equation}
Since $d''x^{p,t} = 0$, there exists a cocycle $\alpha_x \in C^p(X,\mathcal{F}^\perp)$ such that $h\alpha_x = x^{p,t}$ (cf. the definition of $D\alpha$). It is straightforward to see that the map $[x] \mapsto [\alpha_x]$ is the inverse map of $D$.\\

Later in Theorem \ref{thm:Cap induced poincare dual map} we will show that when $X$ can be subdivided into a simplicial complex, then $D[\alpha]=[\alpha]\frown [X]$, which is exactly the original form of Poincar\'e duality for manifolds.

\subsubsection*{Duality of code distances:}

One significant reason for explicitly building the isomorphism $D$ is that it enables us to bound the code distance and (co)cycle expansion and even yields a decoder. The basic idea is to carefully analyze the relationship of norms of (co)cycles in diagram chase. For convenience of presentation, let us denote $i= t-p$. By Eq.~\eqref{eq:zig-zag_3}, for each $0 \leq k \leq t-i-1$ we have $x^{k,i+k} = d'x^{k-1,i+k}$. By the definition of $d'$ in Eq.~\eqref{eq:d'_2}, this implies that 
\begin{equation}\label{eq:explain seq middle <}
	\|x^{k.i+k}\| \leq \text{max}_{\tau \in X(k-1)}|X_{\geq \tau}(k)|\cdot\|x^{k-1,i+k}\|=M^k_{k-1}\|x^{k-1,i+k}\|,
\end{equation}
where for each $j \geq i$ we define $M^{j}_i \coloneqq  \text{max}_{\tau \in X(i)}|X_{\geq \tau}(j)|$. The norm used here is the block Hamming weight, which records the number of nonzero local vectors of support at $x^{k.i+k}(\sigma,-)$. Similarly, for later use, we will also define $m^j_i\coloneqq \text{max}_{\pi \in X(i)}|X_{\leq \pi}(j)|$ with $j \leq i$.  Note that if $X$ is a sparse complex, then $M^*_*$ and $m^*_*$'s are always constants. For $x^{0,i} = \eta' x$, we also have an obvious estimation on the norms:
\begin{equation}\label{eq:explain seq initial <}
	\|x^{0,i}\| \leq \text{max}_{\pi \in X(i)}|X_{\leq \pi}(0)|\cdot\|x\| = m^0_i\|x\|.
\end{equation}
By using this inequality, we are going to bound all other terms appearing in the diagram chasing by $\|x\|$. For $x^{0,i+1}$ with $d'' x^{0,i+1} = x^{0,i}$, if $x^{0,i+1}(\sigma,-)=0$ for some $\sigma \in X(0)$, then $x^{0,i}(\sigma,-) = 0$ by the definition of $d''$ in Eq.~\eqref{eq:d''}. As a result, $\|x^{0,i}\| \leq \|x^{0,i+1}\|$, but this is not enough to bound $\|x^{0,i+1}\|$ via Eq.~\eqref{eq:explain seq initial <}. Since we have the freedom to choose $x^{0,i+1}$, We want to choose some $\bar{x}^{0,i+1}$ for which $d'' \bar{x}^{0,i+1} = x^{0,i}$ and $\bar{x}^{0,i+1}(\sigma,-)=0$ whenever $x^{0,i}(\sigma,-) = 0$ for some $\sigma \in X(k)$. This yields $\|x^{0,i}\| = \|\bar{x}^{0,i+1}\|$.

To this end, we can use the exactness of $d''$ in Proposition \ref{prop:d'_d''_exact} by adding some $d''\zeta^{0,i+2}$ to $x^{0,i+1}$. Then we also need to redefine $x^{1,i+2}$ and $x^{1,i+1}$ by (cf. Eq.~\eqref{eq:zig-zag_4})
\begin{align}
	\bar{x}^{1,i+2} = x^{1,i+2} + d'\zeta^{0,i+2}, \quad \bar{x}^{1,i+1} = x^{1,i+1} + d' d''\zeta^{0,i+2}. 
\end{align}
This guarantees $d' \bar{x}^{0,i+1} = \bar{x}^{1,i+1} = d'' \bar{x}^{1,i+2}$. By Eq.~\eqref{eq:explain seq middle <}, $\| \bar{x}^{1,i+1} \|$ can be bounded by $\| \bar{x}^{0,i+1}\| = \|x^{0,i}\|$, which is further controlled by $\|x\|$. Then we modify $\bar{x}^{1,i+2}$ and redefine all the remaining terms inductively. We still denote them by $x^{k,i+k}$, $x^{k,i+k+1}$ for simplicity, they are homological equivalent to the original explaining sequence. By Eq.~\eqref{eq:explain seq middle <} and \eqref{eq:explain seq initial <}, we have
\begin{equation}\label{eq:bound x^t-i,t by x}
    \|x^{t-i,t}\| \leq (\prod_{j=0}^{t-i-1}M_j^{j+1}) \cdot \|x^{0,i}\| \leq m_i^0(\prod_{j=0}^{t-i-1}M_j^{j+1}) \cdot \|x\|. 
\end{equation}
Note that since $h''$ is injective, we always have $\|x^{t-i,t}\| = \|\alpha_x\|$. Therefore, as long as
\begin{equation}\label{eq:x_small}
    \|x\| < \frac{\mu_{\delta^\perp}(t-i) }{ m_i^0(\prod_{j=0}^{t-i-1}M_j^{j+1}) },
\end{equation}
then $\|\alpha_x\| < \mu_{\delta^\perp}(t-i)$. By assumption, $\alpha_x$ can only be a coboundary and we can find some $\beta_x \in C^{t-i-1}(X,\mathcal{F}^\perp)$ such that $\delta^\perp \beta_x = \alpha_x$. In turn, we are going to show that $x$ must be a boundary of some element in $C_{i+1}(X,\mathcal{F})$ by using $\beta_x$. Therefore,
\begin{equation}
	\mu_{\partial}(i) \geq \frac{ \mu_{\delta^\perp}(t-i) }{ m_i^0(\prod_{j=0}^{t-i-1}M_j^{j+1}) }.
\end{equation}

To prove that $x$ is a boundary, let $\tilde{x}^{t-i-1,t} \coloneqq  x^{t-i-1,t} + h'' \beta_x \in C^{t-i-1}(X,\mathcal{F}_t)$. Then
\begin{align}
	d'\tilde{x}^{t-i-1,t} = x^{t-i,t} + d' h'' \beta_x = x^{t-i,t} + d'' \delta^\perp \beta_x = 0,
\end{align}
hence we can find $\tilde{x}^{t-i-2,t} \in C^{t-i-2}(X,\mathcal{F}_t)$ such that $d'\tilde{x}^{t-i-2,t} = \tilde{x}^{t-i-1,t}$. One may have the feeling that we are going to use the diagram chase again, but we now start from $(t-i-1,t)$. By definition,
\begin{align}
	d'' d'\tilde{x}^{t-i-2,t} = x^{t-i-1,t-1} + d'' h'' \beta_x = x^{t-i-1,t-1}.
\end{align}
Then we define $\tilde{x}^{t-i-2,t-1} \coloneqq  x^{t-i-2,t-1} + d''\tilde{x}^{t-i-2,t}$. We have
\begin{align}
	d'\tilde{x}^{t-i-2,t-1} = d' x^{t-i-2,t-1} + d' d''\tilde{x}^{t-i-2,t} = 0.
\end{align}
We can do this inductively and define $\tilde{x}^{k,i+k+1} \in C^k(X,\mathcal{F}_{i+k+1})$ for each $0\leq k\leq t-i-1$ and $\tilde{x}^{k-1,i+k+1}$ for each $1\leq k\leq t-i-1$ such that 
\begin{align}
    & \tilde{x}^{k,i+k+1} = x^{k,i+k+1}+d''\tilde{x}^{k,i+k+2}, \\ 
    & d'\tilde{x}^{k,i+k+1} = 0, \\
    & d'\tilde{x}^{k-1,i+k+1} = \tilde{x}^{k,i+k+1}. 
\end{align}
Finally, we have $d'\tilde{x}^{0,i+1} = 0$. Then we can define $\tilde{x}\in C_{i+1}(X,\mathcal{F})$ by $\eta'\tilde{x} = \tilde{x}^{0,i+1}$. Since  we also have $d''\tilde{x}^{0,i+1} = x^{0,i}$, $x = \partial\tilde{x}$ by \eqref{eq:eta_natural}. A similar argument shows that $\mu_{\delta^\perp}(t-i)\geq \Theta(\mu_{\partial}(i))$, which concludes the proof of code distance duality.

\subsubsection*{Duality of decoders:} 

By modifying the strategy of proof of distance duality a little bit, we can prove this duality. Suppose we have a decoder $\mathcal{O}_{t-i}^\perp$ with decoding radius $R(\mathcal{O}_{t-i}^\perp)$ when we put qubits in the cochain $C^{t-i}(X,\mathcal{F}^\perp)$. Then consider an error $e\in C_i(X,\mathcal{F})$ with syndrome $s=\partial e\in C_{i-1}(X,\mathcal{F})$. As usual, we construct an explaining sequence $s^{*,*}$ for $s$, and each element in the sequence is known to the decoder. Let $z^{0,i}\coloneqq \eta' e$. Note that since $d''(z^{0,i}+s^{0,i})=0$, we may find $z^{0,i+1}\in C^0(X,\mathcal{F}_{i+1})$ such that $d''z^{0,i+1}=\eta' e+s^{0,i}$. Then we define $z^{1,i+1}\coloneqq d'z^{0,i+1}$. The direct calculation gives $d''z^{1,i+1}=s^{1,i}$. Hence, we may further find $z^{1,i+2}\in C^1(X,\mathcal{F}_{i+2})$ such that $d''z^{1,i+2}=z^{1,i+1}+s^{1,i+1}$, and define $z^{2,i+2}\coloneqq d'z^{1,i+2}$, which satisfies $d''z^{2,i+2}=s^{2,i+1}$. By induction, for each $0\leq k\leq t-i-1$, there exist $z^{k,k+i+1}\in C^k(X,\mathcal{F}_{k+i+1})$, and for each $1\leq k\leq t-i$, there exist $z^{k,k+i}\in C^k(X,\mathcal{F}_{k+i})$ such that 
\begin{equation}
    z^{k+1,k+i+1}=d'z^{k,k+i+1},d''z^{k+1,k+i+1}=s^{k,k+i},d''z^{k+1, k+i+2}=z^{k+1,k+i+1}+s^{k+1,k+i+1}.
\end{equation}
Since $d''(z^{t-i,t}+s^{t-i,t})=0$, we can find $\beta_z\in C^{t-i}(X,\mathcal{F}^\perp)$ such that 
\begin{equation}
d''\beta_z=z^{t-i,t}+s^{t-i,t}.
\end{equation}
Direct calculation gives $\delta^\perp\beta_z = \alpha_s$. If we also have $\|\beta_s\|<R(\mathcal{O}_{t-i}^\perp)$, then we can use the decoder $\mathcal{O}_{t-i}^\perp$ to decode it, and finally decode $e$. 

From $s^{0,i-1}=d''\eta' e$, we deduce that $\|s^{0,i}\|=\|s^{0,i-1}\|\leq\|\eta' e\|\leq m^0_i\|e\|$. From Eq.~\eqref{eq:bound x^t-i,t by x} we know that 
\begin{equation}
\|s^{t-i,t}\|\leq(\prod_{j=0}^{t-i-1}M_j^{j+1})\cdot \|s^{0,i}\|\leq m^0_i(\prod_{j=0}^{t-i-1}M_j^{j+1})\cdot\|e\|.
\end{equation}
We begin to bound $\beta_z$ from
\begin{equation}
\|z^{0,i+1}\|\leq\|\eta' e\|+\|s^{0,i}\|\leq 2m^0_i\|e\|.
\end{equation}
There are useful inequalities in the intermediate step:
\begin{equation}
\|z^{k+1,i+k+1}\|=\|d'z^{k,i+k+1}\|\leq M^{k+1}_k\|z^{k,i+k+1}\|,
\end{equation}
\begin{equation}
\|z^{k+1,k+i+2}\|=\|z^{k+1,k+i+1}+s^{k+1, k+i+1}\|\leq M^{k+1}_k(\|z^{k,k+i+1}\|+\|s^{k,k+i}\|).
\end{equation}
By using them inductively, we get
\begin{equation}
\|z^{t-i,t}\|\leq(\prod_{j=0}^{t-i-1}M_j^{j+1})(\|z^{0,i+1}\|+(t-i-1)\|s^{0,i}\|),
\end{equation}
and hence
\begin{align*}
      \|\beta_z\|&\leq \|z^{t-i,t}\|+\|s^{t-i,t}\|\leq(\prod_{j=0}^{t-i-1}M_j^{j+1})(\|z^{0,i+1}\|+(t-i)\|s^{0,i}\|)\leq (t-i+2)(\prod_{j=0}^{t-i-1}M_j^{j+1})m^0_i\|e\|.
\end{align*}
Therefore, as long as 
\begin{equation}
\|e\|<\frac{R(\mathcal{O}_{t-i}^\perp)}{(t-i+2)(\prod_{j=0}^{t-i-1}M_j^{j+1})m^0_i},
\end{equation}
then we can use the decoder $\mathcal{O}_{t-i}^\perp$ to decode $\alpha_s$ and get $\tilde{\beta}_z\in C^{t-i}(X,\mathcal{F}^\perp)$ such that $\delta^\perp\tilde{\beta_z}=\alpha_s$ and $\tilde{\beta}_z$ is homologous to $\beta_z$. This allows us to find an element $\tilde{e}\in C_i(X,\F)$ homologous to $e$ as follows. Let $\tilde{z}^{t-i,t}\coloneqq s^{t-i.t}+h''\tilde{\beta}_z$, which is known to the decoder $\mathcal{O}^\perp_{t-i}$.  Obviously, $d'\tilde{z}^{t-i,t}=0$, so we can find $\tilde{z}^{t-i-1,t}$ such that $d'\tilde{z}^{t-i-1,t}=\tilde{z}^{t-i,t}$.
Suppose $\tilde{\beta}_z=\beta_z+\delta^\perp\gamma$ for some $\gamma\in C^{t-i+1}(X,\F)$, then $\tilde{z}^{t-i,t}=z^{t-i,t}+d'd''\gamma$. Therefore, $d'(\tilde{z}^{t-i-1,t}+d''\gamma+z^{t-i-1,t})=0$, which also provides some $\gamma^{t-i-2,t}$ such that $\tilde{z}^{t-i-1,t}=z^{t-i-1,t}+d''\gamma+d'\gamma^{t-i-2,t}$. Then we define $\tilde{z}^{t-i-1,t-1}\coloneqq d'' \tilde{z}^{t-i-1,t}=z^{t-i-1,t-1}+d'd''\gamma^{t-i-2,t}$. By doing this iteratively, we acquire $\tilde{z}^{0,i}=z^{0,i}+\eta'\partial \gamma^{-1,i-1}$ for some $\gamma^{-1.i+1}\in C_{i+1}(X,\F)$ and $d'\tilde{z}^{0,i}=0$. As a result, we get $\tilde{e}\in C_i(X,\F)$ such that $\eta'\tilde{e}=\tilde{z}^{0,i}$, and $\tilde{e}$ is known to the decoder and homologous to $e$.

Therefore, we claim that when we put qubits in $C_i(X,\mathcal{F})$, there is a decoder $\mathcal{O}_i$ such that the coding radius satisfies
\begin{equation}
R(\mathcal{O}_i)\geq\frac{R(\mathcal{O}_{t-i}^\perp)}{(t-i+2)(\prod_{j=0}^{t-i-1}M_j^{j+1})m^0_i}.
\end{equation}

Note that each process of diagram chasing takes constant time, hence the running time of $\mathcal{O}_i$ is constant time plus the running time of $\mathcal{O}_{t-i}^\perp$. A similar argument will show that we can construct a decoder $\mathcal{O}^\perp_{t-i}$ given $\mathcal{O}_i$ with $R(\mathcal{O}^\perp_{t-i})\geq \Theta(R(\mathcal{O}_i))$, which concludes the proof.

\subsubsection*{Duality of (co)cycle expansions:} 

By combining the proof of duality of code distance and decoder together, we can easily get the duality of (co)cycle expansion by carefully analyzing the relation of norms in the diagram chase. We use the same notation with the proof of duality of decoders. Without loss of generality, we may assume $\|s\|=\|\partial e\|=\varepsilon_{\partial}(i)\cdot\text{dist}(e,\ker\partial_i)$. Since $\alpha_s$ is a coboundary, we can choose another $\beta\in C^{t-i+1}(X,\F)$ such that $\delta^\perp\beta=\alpha_s$ and  
\begin{equation}
\forall\,\gamma\in\ker\delta^\perp_{ t-i}, \|\beta\|\leq\|\beta+\gamma\|.
\end{equation}
In other words, $\text{dist}(\beta,\ker \delta^\perp_{t-i})=\|\beta\|$. This can always be done by greedy algorithm. When $\alpha_s=0$, we can simply choose $\beta=0$. Otherwise $\delta^\perp\beta\neq0$, hence we have
\begin{equation}
\|\beta\|\leq\frac{\|\alpha_s\|}{\varepsilon_{\delta^\perp}(t-i)},
\end{equation}

Given $\beta$, we can imitate the process of constructing $\tilde{x}$ from $\beta_x$ in the proof of duality of code distance, i.e., we may start with the definition of
$\tilde{e}^{t-i,t}\in C^{t-i}(X,\mathcal{F}_t)$ by
\begin{equation}
\tilde{e}^{t-i,t}\coloneqq e^{t-i,t}+h''\beta.
\end{equation}
By a similar induction, We can define $\tilde{e}^{k,i+k}\in C^k(X,\mathcal{F}_{i+k})$ for each $0\leq k\leq t-i$ and $\tilde{e}^{k-1,i+k}$ for each $1\leq k\leq t-i$ such that 
\begin{equation}
    \tilde{e}^{k,i+k}=e^{k,i+k}+d''\tilde{e}^{k,i+k+1},\, d'\tilde{e}^{k,i+k}=0,\,d'\tilde{e}^{k-1,i+k}=\tilde{e}^{k,i+k}, 
\end{equation}
and $\tilde{e}\in C_{i}(X,\mathcal{F})$ by
\begin{equation}
\eta'\tilde{e}=\tilde{e}^{0,i}.
\end{equation}
Still, $\partial\tilde{e}=s$. We want to give the upper bound of $\|\tilde{e}\|$. This can be done as follows. We start with the observation that 
\begin{equation}
\|\tilde{e}^{t-i,t}\|=\|s^{t-i-1,t}+h''\beta\|\leq\|s^{t-i,t}\|+\|h''\beta\|=\|s^{t-i,t}\|+\|\beta\|.
\end{equation}
Then we note that 
\begin{equation}
\|\tilde{e}^{t-i-1,t-1}\|=\|s^{t-i-1,t-1}+d''\tilde{e}^{t-i-1,t}\|\leq\|s^{t-i-1,t-1}\|+\|\tilde{e}^{t-i-1,t}\|.
\end{equation}
We know $d'\tilde{e}^{t-i-1,t}=\tilde{e}^{t-i,t}$. By the previous method, $\|\tilde{e}^{t-i,t}\|\leq M^{t-i}_{t-i-1}\|\tilde{e}^{t-i-1,t}\|$. To give an upper bound of $\|\tilde{e}^{t-i-1,t}\|$ via $\|\tilde{e}^{t-i,t}\|$, we apply the following lemma

\begin{lemma}\label{dz=x upper bound for z using x}
    For each $1\leq i\leq k\leq t$, $\xi\in C^i(X,\mathcal{F}_k)$, there exist $\zeta\in C^{i-1}(X,\mathcal{F}_k)$ such that $d'\zeta=\xi$ and $\|\zeta\|\leq m^{i-1}_kM^k_i\|\xi\|$
\end{lemma}
\begin{proof}
    Recall in Proposition \ref{prop:d'_exact}, for each $u\in X(k)$, we define the map
    \begin{equation}
    \delta_{\leq u}: C^{i-1}(X_{\leq u},\mathcal{F}_u)\longrightarrow C^i(X_{\leq u},\mathcal{F}_u),
    \end{equation}
    which induces the cohomology of the subcomplex $X_{\leq u}$ with coefficients in $\mathcal{F}_u$. We have
    \begin{equation}
    \delta_{\leq u}\zeta(-,u)=\xi(-,u)
    \end{equation}
    when $\xi(-,u)=0$ for some $u\in X(k)$ and we can choose $\zeta(-,u)=0$. Otherwise, 
    \begin{equation}
    \frac{\|\xi(-,u)\|}{\|\zeta(-,u)\|}\geq\frac{1}{\|\zeta(-,u)\|}\geq\frac{1}{|X_{\leq u}(i-1)|}.
    \end{equation}
    Therefore, in either case, we have 
    \begin{equation}
    \|\zeta(-,u)\|\leq m^{i-1}_k \|\xi(-,u)\|.
    \end{equation}
    Hence, 
    \begin{equation}
    \|\zeta\|\leq \sum_{u\in X(k)}\|\zeta(-,u)\|\leq\sum_{u\in X(k)}m^{i-1}_k\|\xi(-,u)\|\leq m^{i-1}M^k_i\|\xi\|,
    \end{equation}
    where the last inequality is by the trivial bound $M^k_i\|\xi\|\geq\sum_{u\in X(k)}\|\xi(-,u)\|$.
\end{proof}

Substituting $\tilde{e}^{t-i-1,t}$ for $\zeta$ in the above lemma, we have 
\begin{align}
\begin{aligned}
    \|\tilde{e}^{t-i-1,t-1}\|&\leq \|s^{t-i-1,t-1}\|+m^{t-i-1}_tM^t_{t-i}\|\tilde{e}^{t-i,t}\|\\
    &\leq \|s^{t-i-1,t-1}\|+m^{t-i-1}_tM^t_{t-i}(\|s^{t-i,t}\|+\frac{\|\alpha_s\|}{\varepsilon_{\delta^\perp}(t-i)})\\
    &\leq(1+m^{t-i-1}_tM^t_{t-i}M^{t-i}_{t-i-1})\|s^{t-i-1,t-1}\|+m^{t-i-1}_tM^t_{t-i}\frac{\|\alpha_s\|}{\varepsilon_{\delta^\perp}(t-i)}.
\end{aligned} 
\end{align}
Then we can combine the inequality $\|\tilde{e}^{t-i-2,t-2}\|\leq\|s^{t-i-2,t-2}\|+\|\tilde{e}^{t-i-2,t-1}\|$ and Lemma \ref{dz=x upper bound for z using x} to give an upper bound of $\|\tilde{e}^{t-i-2,t-2}\|$ by $\|s^{t-i-2,t-2}\|$ and $\|\alpha_s\|$. By induction, we finally get
\begin{equation}
\|\tilde{e}^{0,i}\|\leq(1+\sum_{j=0}^{t-i-1}\prod_{k=0}^jm^k_{k+i+1}M^{k+i+1}_{k+1}M^{k+1}_k)\|s^{0,i}\|+(\prod_{j=0}^{t-i}m^j_{j+i}M^{j+i}_{j+1})\frac{\|\alpha_s\|}{\varepsilon_{\delta^\perp}(t-i)}.
\end{equation}
Note that 
\begin{align}
\begin{aligned}
    \|\tilde{e}\|&=\sum_{\sigma\in X(i)}\mathbbm{1}_{\tilde{x}(\sigma)\neq0}\leq\sum_{v\in X(0)}\sum_{\sigma\in X_{\geq v}(i)}\mathbbm{1}_{\tilde{x}(\sigma)\neq0}\\
    &=\sum_{v\in X(0)}\sum_{\sigma\in X_{\geq v}(i)}\mathbbm{1}_{\tilde{x}^{0,i+1}(v,\sigma)\neq0}\\
    &\leq\sum_{v\in X(0)}M^{i}_0\mathbbm{1}_{\tilde{x}^{0,i}(v-)\neq0}=M^{i}_0\|\tilde{x}^{0,i}\|.
\end{aligned}    
\end{align}
Together with the fact that $\|s^{0,i}\|=\|s^{0,i-1}\|\leq m^0_{i-1}\|s\|$ and Eq.~\eqref{eq:bound x^t-i,t by x}, we obtain
\begin{equation}
    \|\tilde{e}\|\leq m^0_{i-1}M^{i}_0 \Big(1+\sum_{j=0}^{t-i-1}\prod_{k=0}^jm^k_{k+i+1}M^{k+i+1}_{k+1}M^{k+1}_k+\frac{1}{\varepsilon_{\delta^\perp}(t-i)}\prod_{j=0}^{t-i}m^j_{j+i}M^{j+i}_{j+1}M^{j+1}_j\Big) \cdot\|s\|.
\end{equation}
Note that we always have an upper bound of (co)cycle expansion as follows. Suppose $\varphi\in C^{t-i}(X,\F^\perp)\setminus\ker\delta^\perp_{t-i}$ satisfies
\begin{equation}
\|\delta^\perp\varphi\|=\varepsilon_{\delta}^\perp(t-i)\cdot\|\varphi\|.
\end{equation}
Then the trivial bound
\begin{equation}
\|\delta^\perp\varphi\|\leq M^{t-i+1}_{t-i}\|\varphi\|   
\end{equation}
gives
\begin{equation}
\varepsilon_{\delta}^\perp(t-i)\leq M^{t-i+1}_{t-i}.
\end{equation} 
Therefore,
\begin{align}
\begin{aligned}
    \|\tilde{e}\|\leq m^0_{i-1}M^{i}_0M^{t-i+1}_{t-i} \Big( 1 & + \sum_{j=0}^{t-i-1}\prod_{k=0}^jm^k_{k+i+1}M^{k+i+1}_{k+1}M^{k+1}_k \\
    & + \frac{1}{M^{t-i+1}_{t-i}}\prod_{j=0}^{t-i}m^j_{j+i}M^{j+i}_{j+1}M^{j+1}_j\Big) \cdot\frac{\|s\|}{\varepsilon_{\delta^\perp}(t-i)}.
\end{aligned}
\end{align}
Note that $\tilde{e}-e\in\ker\partial_i$, and thus
\begin{equation}
\|\tilde{e}\|=\|e-(e-\tilde{e})\|\geq \text{dist}(e,\ker\partial_i).
\end{equation}
Together with $\|s\|=\|\partial e\|=\varepsilon_{\partial}(i)\cdot\text{dist}(e,\ker\partial_i)$ gives
\begin{equation}
\varepsilon_{\partial}(i)\geq\frac{\varepsilon_{\delta^\perp}(t-i)}{m^0_{i-1}M^{i}_0M^{t-i+1}_{t-i}\big(1+\sum_{j=0}^{t-i-1}\prod_{k=0}^jm^k_{k+i+1}M^{k+i+1}_{k+1}M^{k+1}_k+\frac{1}{M^{t-i+1}_{t-i}}\prod_{j=0}^{t-i}m^j_{j+i}M^{j+i}_{j+1}M^{j+1}_j\big)}.
\end{equation}
A similar argument gives $\varepsilon_{\delta^\perp}(t-i)\geq\Theta(\varepsilon_{\partial}(i))$, which completes the proof of (co)cycle expansion.


\section{Multiplicative structures and logical gates}\label{sec:multiplicative}


In algebraic topology, Poincaré duality is deeply intertwined with multiplicative structures such as the cup and cap products. It is therefore natural to ask whether sheaf codes admit such multiplicative structures. We show that this is indeed the case, and that the resulting structures substantially enrich the theory of sheaf codes. 

{Importantly for quantum computation, these multiplicative structures underlie the construction of multi-controlled-$Z$ gate in previous studies \cite{10.21468/SciPostPhys.14.4.065,Chen_2023,Wang_2024,,Breuckmann:2024cupandgate,lin2024transversalnoncliffordgatesquantum,Golowich_Lin2024,zhu2025transversalnoncliffordgatesqldpc}.
However, some of these constructions cannot be generalized to sheaf codes, and some would require extra structures, such as the multi-orthogonality on local codes. The multi-orthogonality is not likely to be compatible with a key product-expansion property on local codes~\cite{kalachev2025maximallyextendableproductcodes}, which is crucial for achieving the (nearly) optimal code parameters. 
Consequently, there is currently no known construction of constant-depth multi-controlled-$Z$ gate on (nearly) optimal qLDPC codes.} 

{By using both cup and cap products, we are able to write down a variety of multi-linear cohomological invariants on general sheaf codes {without} imposing any additional assumption on local codes. As a corollary of Theorem \ref{thm:Poincaré_duality}, we give the first example of transversal logical $\CZ$ gates on good qLDPC codes. 
We can also show that our method can yield logical $\CCZ$ gates that are provably logical, although the current technique is not yet sufficient to establish nontrivial lower bounds on $k_{\CCZ}$.
Our definition of cup product differs  from that in \cite{lin2024transversalnoncliffordgatesquantum}. First, that our cup product is exactly how it is defined using \v{C}ech cohomology as in Section \ref{sec:sheaf_cohomology}, while the one in  \cite{lin2024transversalnoncliffordgatesquantum} is not, hence we expect that a deeper mathematical investigation into our cup product will be useful to bound $k_{\CCZ}$. Second, our Poincar\'e duality isomorphism induced by the cap product is defined using tensor products of sheaves, which is fundamentally different from the construction in \cite{lin2024transversalnoncliffordgatesquantum}.}


In Section~\ref{subsec:cup_cap_pairing}, we define the cup product on simplicial complexes based on the key observation that the cohomology of sheaf codes is equivalent to \v{C}ech cohomology. We also introduce cap products as a new technique for studying sheaf codes. In Section~\ref{subsec:subdivision}, we extend all these concepts to the more general setting of cell complexes, which encompasses the cubical complexes considered in previous studies of qLDPC codes \cite{leverrier2022quantum,DHLV2022,Dinur2024sheaf}. In Section~\ref{subsec:duality-cap-product}, we complete the proof of the isomorphism of Poincar\'e duality induced by cap product. In Section~\ref{subsec:explicit}, we provide explicit calculations for cubical complexes. In Section~\ref{sec:explicit gate construction}, we build constant-depth logical $\CZ$ gates on good qLDPC codes, and present the conjectures about multi-linear cohomological invariants and non-Clifford gates on almost good qLTCs.

\subsection{Cup products, cap products and pairing on simplicial complexes}\label{subsec:cup_cap_pairing}
In {the theory of \v{C}ech cohomology}, given an open cover $\mathcal{U}=\{U_i\}_{i\in I}$ of $X$ with three sheaves $\F$, $\mathcal{G}$ and $\mathcal{H}$. Suppose we have a bilinear map of sheaves
\begin{equation}
\mu:\F\times\mathcal{G}\longrightarrow\mathcal{H},
\end{equation}
then for $\alpha\in C^i(\mathcal{U},\F)$, $\beta\in C^j(\mathcal{U},\mathcal{G})$, one may define $\alpha\smile\beta\in C^{i+j}(\mathcal{U},\mathcal{H})$ by 
\begin{equation}
(\alpha\smile\beta)_{i_0\cdots i_{p+q}}\coloneqq \mu(\alpha_{i_0\cdots i_p}|_{U_{i_0\cdots i_{p+q}}},\beta_{i_p\cdots i_{p+q}}|_{U_{i_0\cdots i_{p+q}}}).
\end{equation}

A canonical choice of $\mu$ is the tensor product defined as follows

\begin{definition}[Tensor product of sheaves]
    Let $X$ be a topological space with two sheaves $\F$ and $\mathcal{G}$, the tensor product of the two sheaves $\F\otimes\mathcal{G}$ is defined to be the sheafification of the following presheaf, for any open set $U\subseteq X$,
    \begin{equation}
    U\longmapsto\F(U)\otimes \mathcal{G}(U).
    \end{equation}
\end{definition}
Since sheafification preserves the stalks, it is not hard to see that when $X$ is a poset with Alexandrov topology, the stalk of the tensor sheaf at $\sigma\in X$
\begin{equation}
(\F\otimes\mathcal{G})_\sigma=\F_\sigma\otimes\mathcal{G}_\sigma,
\end{equation}
and when $\sigma\leq \tau$, the restriction map is
\begin{equation}
(\F\otimes\mathcal{G})_{\sigma,\tau}=\F_{\sigma,\tau}\otimes\mathcal{G}_{\sigma,\tau}.
\end{equation}

Now we are able to define cup products on simplicial complexes by translating \v{C}ech theory into the context of sheaf codes.

\begin{definition}[Cup product]
	Given cochain complexes $C^\bullet(X, \mathcal{F})$ and $C^\bullet(X,\mathcal{G})$ defined on a simplicial complex $X$, we define the \te{cup product}
	\begin{align}
		\smile: C^p(X,\mathcal{F}) \times C^q(X,\mathcal{G}) \rightarrow C^{p+q}(X,\mathcal{F} \otimes \mathcal{G})
	\end{align}
	as follows: for $\alpha\in C^p(X,\F)$, $\beta\in C^{q}(X,\mathcal{G})$, $\sigma=[ v_0,...,v_{p+q}]$ a $(p+q)$-simplex, let ${}_p\sigma=[v_0,v_1,...,v_p]$ be the \te{former p-face} and $\sigma_q=[v_p,\cdots, v_{p+q}]$ be the \te{later q-face},
	\begin{align}
		(\alpha \smile \beta)(\sigma) \coloneqq  \mathcal{F}_{{}_p\sigma, \sigma}(\alpha( {}_p\sigma ) ) \otimes \mathcal{G}_{\sigma_q, \sigma}(\beta( \sigma_q ) ) .
	\end{align}
\end{definition}

\begin{proposition}\label{prop:cup_Leibniz}
	The cup product satisfies the \te{Leibniz rule}:
	\begin{align}
		\delta(\alpha \smile \beta) = (\delta \alpha) \smile \beta + \alpha \smile (\delta \beta). 
	\end{align}
\end{proposition}
\begin{proof}
For a $(p+q+1)$-cell $\sigma = [v_0,...,v_{p+q+1}]$,
\begin{align}
& ( \delta(\alpha \smile \beta) )(\sigma) 
		= \sum_{\sigma' \lessdot \sigma} (\mathcal{F} \otimes \mathcal{G})_{\sigma',\sigma}[ (\alpha \smile \beta)(\sigma') ] \notag \\
		= & \sum_{i = 0}^p (\mathcal{F} \otimes \mathcal{G})_{\sigma \setminus v_i,\sigma}[ 
		\mathcal{F}_{{}_{p+1}\sigma \setminus v_i, \sigma \setminus v_i}(\alpha( {}_{p+1}\sigma \setminus v_i ) ) \otimes \mathcal{G}_{\sigma_q, \sigma \setminus v_i}(\beta( \sigma_q ) )  ] \notag \\
		& + \sum_{i = p+1}^{p+q+1} (\mathcal{F} \otimes \mathcal{G})_{\sigma \setminus v_i,\sigma}[ 
		\mathcal{F}_{{}_p\sigma, \sigma \setminus v_i}(\alpha( {}_p\sigma ) ) \otimes \mathcal{G}_{\sigma_{q+1} \setminus v_i, \sigma \setminus v_i }(\beta( \sigma_{q+1} \setminus v_i) )  ] \\ 
		= & \sum_{i = 0}^p 
		\mathcal{F}_{{}_{p+1}\sigma \setminus v_i, \sigma }(\alpha( {}_{p+1}\sigma \setminus v_i ) ) \otimes \mathcal{G}_{\sigma_q, \sigma }(\beta( \sigma_q ) ) 
		+ \sum_{i = p+1}^{p+q+1} 
		\mathcal{F}_{{}_p\sigma, \sigma }(\alpha( {}_p\sigma ) ) \otimes \mathcal{G}_{\sigma_{q+1} \setminus v_i, \sigma }(\beta( \sigma_{q+1} \setminus v_i) )  \notag \\ 
		= & \sum_{i = 0}^{p+1}
		\mathcal{F}_{{}_{p+1}\sigma \setminus v_i, \sigma }(\alpha( {}_{p+1}\sigma \setminus v_i ) ) \otimes \mathcal{G}_{\sigma_q, \sigma }(\beta( \sigma_q ) ) 
		+ \sum_{i = p}^{p+q+1} 
		\mathcal{F}_{{}_p\sigma, \sigma }(\alpha( {}_p\sigma ) ) \otimes \mathcal{G}_{\sigma_{q+1} \setminus v_i, \sigma }(\beta( \sigma_{q+1} \setminus v_i) ), \notag
\end{align} 
where in the last term, we add 
\begin{equation}
\mathcal{F}_{{}_{p}\sigma, \sigma }(\alpha( {}_{p}\sigma ) ) \otimes \mathcal{G}_{\sigma_q, \sigma }(\beta( \sigma_q ) )
\end{equation}
twice and hence the equality holds. On the other hand,
\begin{align}
	\begin{aligned}
		( (\delta \alpha) \smile \beta) (\sigma) 
		= & \mathcal{F}_{{}_{p+1}\sigma, \sigma}( (\delta\alpha) ( {}_{p+1}\sigma ) ) \otimes \mathcal{G}_{\sigma_q, \sigma}(\beta( \sigma_q ) ) \\
		= & \sum_{i = 0}^{p+1} \mathcal{F}_{{}_{p+1}\sigma, \sigma}\, \mathcal{F}_{{}_{p+1}\sigma \setminus v_i, {}_{p+1}\sigma} (\alpha ( {}_{p+1}\sigma \setminus v_i ) )\otimes \mathcal{G}_{\sigma_q, \sigma}(\beta( \sigma_q ) ) \\
		=  & \sum_{i = 0}^{p+1} \mathcal{F}_{{}_{p+1}\sigma \setminus v_i, \sigma} (\alpha ( {}_{p+1}\sigma \setminus v_i ) ) \otimes \mathcal{G}_{\sigma_q, \sigma}(\beta( \sigma_q ) ) 
	\end{aligned}
\end{align}
Similarly,
\begin{align}
	(\alpha \smile (\delta \beta))(\sigma) = 
	\sum_{i = p}^{p+q+1} 
	\mathcal{F}_{{}_p\sigma, \sigma }(\alpha( {}_p\sigma ) ) \otimes \mathcal{G}_{\sigma_{q+1} \setminus v_i, \sigma }(\beta( \sigma_{q+1} \setminus v_i) )
\end{align}
and this finishes the proof.
\end{proof}

\begin{proposition}\label{prop:cup_assoc}
	The cup product is associative:
	\begin{align}
		(\alpha \smile \beta) \smile \gamma = \alpha \smile (\beta \smile \gamma)   
	\end{align}
	for arbitrary $\alpha \in C^p(X, \mathcal{F})$, $\beta \in C^q(X, \mathcal{G})$ and $\gamma \in C^r(X, \mathcal{H})$.
\end{proposition}
\begin{proof}
Suppose $\sigma=[v_0,v_1,\cdots,v_{p+q+r}]\in X(p+q+r)$, then by definition
\begin{align}
    \begin{aligned}
        ((\alpha\smile\beta)\smile\gamma)(\sigma)&=((\mathcal{F}\otimes\mathcal{G})_{{}_{p+q}\sigma,\sigma}(\alpha\smile\beta)({}_{p+q}\sigma))\otimes\mathcal{H}_{\sigma_r,\sigma}\beta(\sigma_r)\\
        &=(\mathcal{F}\otimes\mathcal{G})_{{}_{p+q}\sigma,\sigma}(\mathcal{F}_{{}_p\sigma,{}_{p+q}\sigma}\alpha({}_p\sigma)\otimes\mathcal{G}_{({}_{p+q}\sigma)_q,{}_{p+q}\sigma}\beta(({}_{p+q}\sigma)_q))\otimes\mathcal{H}_{\sigma_r,\sigma}(\sigma_r)\\        &=\mathcal{F}_{{}_p\sigma,\sigma}\alpha({}_p\sigma)\otimes\mathcal{G}_{({}_{p+q}\sigma)_q,\sigma}\beta(({}_{p+q}\sigma)_q)\otimes\mathcal{H}_{\sigma_r,\sigma}(\sigma_r).
    \end{aligned}
\end{align}
Similarly, we have
\begin{equation}
(\alpha\smile(\beta\smile\gamma))(\sigma)=\mathcal{F}_{{}_p\sigma,\sigma}\alpha({}_p\sigma)\otimes\mathcal{G}_{{}_{q}(\sigma_{q+r}),\sigma}\beta({}_{q}(\sigma_{q+r}))\otimes\mathcal{H}_{\sigma_r,\sigma}(\sigma_r)
\end{equation}
Note that $({}_{p+q}\sigma)_q={}_{q}(\sigma_{q+r})=[v_p,v_{p+1},\cdots,v_q]$, hence the proof is done.
\end{proof}

Recall that the \emph{pairing} between the $i$-th cochain $C^i(X,\mathcal{F})$ and chain $C_i(X,\mathcal{F})$ 
	\begin{align}
		\langle -, - \rangle: C^i(X,\mathcal{F}) \times C_i(X,\mathcal{F}) \rightarrow \mathbb{F}
	\end{align}
	is defined as by
	\begin{align}
		\langle \alpha,x \rangle = \sum_{\sigma \in X(i)} \langle \alpha(\sigma), x(\sigma) \rangle 
	\end{align}
    where $\langle \alpha(\sigma), x(\sigma) \rangle$ is the standard inner product in $\F_\sigma$.
	
    Given any cycle $x$,
	\begin{align}
		\langle -, x \rangle: C^i(X,\mathcal{F}) \times C_i(X,\mathcal{F}) \rightarrow \mathbb{F}_2
	\end{align}
	is a single-linear cohomological invariant because suppose $\alpha = \delta \beta$ as a boundary, then
	\begin{equation}
    \begin{aligned}
        \langle\delta\beta,x\rangle&=\sum_{\sigma\in X(i)}\sum_{\sigma'\lessdot\sigma}\langle\mathcal{F}_{\sigma',\sigma} \beta(\sigma'),x(\sigma)\rangle\\
        &=\sum_{\sigma'\in X(i-1)}\sum_{\sigma\gtrdot\sigma'}\langle\beta(\sigma'),\mathcal{F}_{\sigma',\sigma}^T x(\sigma)\rangle\\
        &=\langle\beta,\partial x\rangle=0.
    \end{aligned}
\end{equation}

Therefore, the cup product together with the pairing is sufficient to construct multi-linear cohomological invariants. Actually, we can build more multiplicative structures. We define the following cap product as a ``partial pairing".

\begin{definition}[Cap product]\label{def:cap-III}
We define the \te{cap product}
\begin{equation}
\frown:C^p(X,\mathcal{F})\times C_{p+q}(X,\mathcal{F}\otimes\mathcal{G})\longrightarrow C_q(X,\mathcal G)
\end{equation}
as follows: for $\alpha \in C^{p}(X,\mathcal{F}), x\in C_{p+q}(X,\mathcal{F}\otimes\mathcal{G})$,
\begin{align}
	\alpha\frown x\coloneqq \sum_{\sigma\in X(p+q)}\mathcal{G}^T_{\sigma_q,\sigma}\langle\mathcal{F}_{{}_p\sigma,\sigma}\alpha({}_p\sigma),x(\sigma)\rangle_{\mathcal{F}}\cdot \sigma_q.
\label{eq:cap_other_definition}
\end{align}
\end{definition}

\begin{proposition}[Leibniz rule for cap product]
    For $\alpha\in C^p(X,\mathcal{F})$ and $x\in C_{p+q+1}(X,\mathcal{F}\otimes\mathcal{G})$, we have $\alpha\frown x\in C_{q+1}(X,\mathcal{G})$ and  $\partial(\alpha\frown x)=(\delta\alpha)\frown x+\alpha\frown (\partial x)$
\end{proposition}

\begin{proof}
Consider a single term $x=x(\sigma)\cdot\sigma$, $\sigma=[v_0,\cdots,v_{p+q+1}]$, then
    \begin{align}
            &\alpha\frown (\partial x)=\alpha\frown \sum_{i=0}^{p+q+1}(\mathcal F\otimes\mathcal G)^T_{\sigma\setminus v_i,\sigma}x(\sigma)\cdot(\sigma\setminus v_i) \notag \\
            &=\sum_{i=0}^{p+q+1}\mathcal G^T_{(\sigma\setminus v_i)_q,\sigma\setminus v_i}\langle\mathcal F_{{}_{p}(\sigma\setminus v_i),\sigma\setminus v_i}\alpha({}_p(\sigma\setminus v_i)),(\mathcal F\otimes\mathcal G)^T_{\sigma\setminus v_i,\sigma}x(\sigma)\rangle_{\mathcal F}\cdot(\sigma\setminus v_i)_q \notag \\
            &=\sum_{i=0}^{p+q+1}\mathcal G^T_{(\sigma\setminus v_i)_q,\sigma}\langle\mathcal F_{{}_{p}(\sigma\setminus v_i),\sigma}\alpha({}_p(\sigma\setminus v_i)),x(\sigma)\rangle_{\mathcal F}\cdot(\sigma\setminus v_i)_q\\
            &=\sum_{i=0}^p\mathcal G^T_{\sigma_q,\sigma}\langle\mathcal F_{{}_{p+1}\sigma\setminus v_i,\sigma}\alpha({}_{p+1}\sigma\setminus v_i),x(\sigma)\rangle_\mathcal F\cdot\sigma_q+\sum_{i=p+1}^{p+q+1}\mathcal G^T_{\sigma_{q+1}\setminus v_i,\sigma}\langle\mathcal F_{{}_p\sigma,\sigma}\alpha({}_p\sigma),x(\sigma)\rangle_{\mathcal F}\cdot(\sigma_{q+1}\setminus v_i) \notag \\
            &=\sum_{i=0}^{p+1}\mathcal G^T_{\sigma_q,\sigma}\langle\mathcal F_{{}_{p+1}\sigma\setminus v_i,\sigma}\alpha({}_{p+1}\sigma\setminus v_i),x(\sigma)\rangle_\mathcal F\cdot\sigma_q+\sum_{i=p}^{p+q+1}\mathcal G^T_{\sigma_{q+1}\setminus v_i,\sigma}\langle\mathcal F_{{}_p\sigma,\sigma}\alpha({}_p\sigma),x(\sigma)\rangle_{\mathcal F}\cdot(\sigma_{q+1}\setminus v_i), \notag
    \end{align}
    where the last equality is derived by adding
    \begin{equation}
        \mathcal G^T_{\sigma_q,\sigma}\langle\mathcal F_{{}_p\sigma,\sigma}\alpha({}_p\sigma),x(\sigma)\rangle_{\mathcal F}\cdot\sigma_q
    \end{equation}
    twice. Note that
    \begin{equation}
        \begin{aligned}
            (\delta\alpha)\frown x&=\mathcal G^T_{\sigma_q,\sigma}\langle\mathcal F_{{}_{p+1}\sigma,\sigma}((\delta\alpha)({}_{p+1}\sigma)),x(\sigma)\rangle_{\mathcal F}\cdot\sigma_q\\
            &=\mathcal G^T_{\sigma_q,\sigma}\langle\mathcal F_{{}_{p+1}\sigma,\sigma}\sum_{i=0}^{p+1}\mathcal F_{{}_{p+1}\sigma\setminus v_i, {}_{p+1}\sigma}\alpha({}_{p+1}\sigma\setminus v_i),x(\sigma)\rangle_{\mathcal F}\cdot\sigma_q\\
            &=\sum_{i=0}^{p+1}\mathcal G^T_{\sigma_q,\sigma}\langle\mathcal F_{{}_{p+1}\sigma\setminus v_i,\sigma}\alpha({}_{p+1}\sigma\setminus v_i),x(\sigma)\rangle_{\mathcal F}\cdot \sigma_{q}.
        \end{aligned}
    \end{equation}
On the other hand,
\begin{equation}
    \begin{aligned}
        \partial(\alpha\frown x)&=\partial (\mathcal G^T_{\sigma_{q+1},\sigma}\langle \mathcal F_{{}_p\sigma,\sigma}\alpha({}_p\sigma),x(\sigma)\rangle_{\mathcal F}\cdot \sigma_{q+1})\\
        &=\sum_{i=p}^{p+q+1}\mathcal G^T_{\sigma_{q+1},\sigma}\langle\mathcal F_{{}_p\sigma,\sigma}\alpha({}_p\sigma),x(\sigma)\rangle_{\mathcal F}\cdot (\sigma_{q+1}\setminus v_i).
    \end{aligned}
\end{equation}
Therefore, we have
\begin{equation}
\partial(\alpha\frown x)=\alpha\frown (\partial x)+(\delta \alpha)\frown x.
\end{equation}
\end{proof}

The cap product is of crucial importance for a more simplified description of the isomorphism $H^i(X,\mathcal{F}^\perp) \cong H_{t-i}(X,\mathcal{F})$ in Theorem~\ref{thm:Cap induced poincare dual map}.


\subsection{Subdivision and pullback sheaf}\label{subsec:subdivision}

Given that the existing (almost) good quantum codes are based on cubical complexes rather than simplicial complexes,  we need to extend the definition of cup and cap product to cubical complexes. Inspired by previous studies~\cite{FreedmanHastings2021,Portnoy2023,lin2024transversalnoncliffordgatesquantum}, we can achieve this by subdividing a cubical complex into a simplicial complex and define cup and cap products through certain (co)chain maps. To this end, we introduce various powerful tools from sheaf theory tools~\cite{curry2014sheavescosheavesapplications}. 

\begin{definition}[Subdivision]
    Let $X$ be a cell complex. A \te{subdivision} of $X$ is a cell complex $\t{X}$ which is the same topological space as $X$ but with a different cell decomposition, such that each cell of $X$ is a union of cells in $\t{X}$.
\end{definition}

For $\tilde{\sigma} \in \tilde{X}(i)$, let $\tau_{\tilde{\sigma}}$ be the minimal cell in $X$ that contains $\tilde{\sigma}$. We will write $\t{U}_{\t{\sigma}}\coloneqq \t{X}_{\geq \t{\sigma}}$ as the elements in the topology basis of $\t{X}$.

\begin{definition}[Pullback sheaf]
    Let $f:X\rightarrow Y$ be a continuous map, and $\F$ is a sheaf on $Y$. Then the \te{pullback sheaf} $f^{*}\F$ is the sheafification of the following presheaf, for each open set $U\subseteq X$,
    \begin{equation}
    U\longmapsto \varinjlim_{V\supseteq f(U)}\F(V)
    \end{equation}
\end{definition}

Given a subdivision, we can construct a map $s:\tilde{X}\rightarrow X$ between posets by, mapping each cell in $\tilde{X}$ to the minimal cell containing it in $X$. It is easy to see that this map preserves partial order of poset. $s$ is a continuous map with respect to Alexandrov topology, because for $U_\sigma$ a basis element of $X$, if $s(\tilde{\tau})\in U_\sigma$, then $s(\tilde{U}_{\tilde{\tau}})\subseteq U_{\sigma}$, hence $s^{-1}(U_{\sigma})$ is open. Therefore, given a cellular sheaf $\F$ on $X$, we can use the pullback to construct a new sheaf $s^*\F$. Since sheafification preserves the stalks, we have, for $\tilde{\sigma}\in\tilde{X}$,
\begin{equation}
(s^*\F)_{\tilde{\sigma}}=\varinjlim_{V\supseteq s(\tilde{U}_{\tilde{\sigma}})}\F(V)=\F_{s(\tilde{\sigma})}=\F_{\tau_{\tilde{\sigma}}}.
\end{equation}
And for each $\t{\sigma}\leq\t{\rho}$, the restriction map is simply
\begin{equation}
(s^*\F)_{\t{\sigma},\t{\rho}}=\F_{\tau_{\t{\sigma}},\tau_{\t{\rho}}}.
\end{equation}

We need one more definition to prove the next proposition.

\begin{definition}[Pushforward sheaf]
    Let $f:X\rightarrow Y$ be a continuous map and let $\F$ be a sheaf on $X$. Then we can define the \te{pushforward sheaf} $f_*\F$ on $Y$ by, for an open set $U\subseteq Y$,
    \begin{equation}
    f_*\F(U)\coloneqq \F(f^{-1}(U)).
    \end{equation}
\end{definition}
Note that when $X$ and $Y$ are cell poset, then the formula on stalk is given by, for $\sigma\in Y$ a cell,
\begin{equation}
(f_*\F)_\sigma=\F(f^{-1}(U_\sigma))=\varprojlim_{V\subseteq f^{-1}(U_\sigma)}\F(V)=\varprojlim_{f(\rho)\geq \sigma} \F_\rho.
\end{equation}

The crucial use of this functor is that for constant map to a point $p:X\rightarrow\star$, we can express the global section as
\begin{equation}
(p_*\F)_\star\cong \F(X)=\Gamma(X,\F).
\end{equation}

These sheaf-theoretic tools yield the following results:

\begin{proposition}[\protect{\cite[Theorem 7.3.9]{curry2014sheavescosheavesapplications}}]\label{prop: subdivided sheaf has same cohomology}
    Suppose $\F$ is a sheaf on $X$ which can be subdivided into $\t{X}$, and $s:\t{X}\rightarrow X$ be the inclusion relation from subdivision, then
    \begin{equation}
    H^\bullet(X,\F)\cong H^\bullet(\t{X},s^*\F).
    \end{equation}
\end{proposition}
\begin{proof}
Let $\star$ be a point. Let $p_X:X\rightarrow\star$ and $p_{\t{X}}:\t{X}\rightarrow\star$ be constant function mapping every cell to the point. Then we have the following commutative diagram
\begin{equation}
\begin{tikzcd}
\t{X} \arrow[r, "s"] \arrow[dr, "p_{\t{X}}"'] & X \arrow[d, "p_X"] \\
& \star
\end{tikzcd}
\end{equation}
$p_{\t{X}}=p_X\circ s$ gives $(p_{\t{X}})_*=(p_X)_*\circ s_*$. Therefore
\begin{equation}
(p_{\t{X}})_*s^*\F=(p_X)_*\circ s_*s^*\F.
\end{equation}
If we can prove $s_*s^*\F\cong \F$, then the proof is done. By definition
\begin{equation}
 (s_*s^*\F)_y=\varprojlim_{s(x)\geq y}(s^*\F)_{x}=\varprojlim_{s(x)\geq y}\F_{s(x)}.
\end{equation}
Note that $s:\tilde{X}\rightarrow X$ is surjective, therefore
\begin{equation}
(s_*s^*\F)_y=\varprojlim_{x\geq y}\,\F_x=\F_y.
\end{equation}
Hence $s_*s^*\F\cong \F$, and the proof is done.
\end{proof}

\begin{corollary}\label{pullback sheaf is locally acyclic}
    Suppose $\F$ is a locally acyclic sheaf on $X$, then $s^*\F$ is a locally acyclic sheaf on $\t{X}$.
\end{corollary}
\begin{proof}
Let $\t{\sigma}\in \t{X}$ be an arbitrary cell, then we note that there is a commutative diagram
\begin{equation}
\begin{tikzcd}
\t{U}_{\t{\sigma}} \arrow[r, "s"] \arrow[dr, "p_{\t{U}_{\t{\sigma}}}"'] & U_{s(\t{\sigma})} \arrow[d, "p_{U_{s(\t{\sigma})}}"] \\
& \star
\end{tikzcd}
\end{equation}
Note that $s|_{\t{U}_{\t{\sigma}}}:\t{U}_{\t{\sigma}}\rightarrow U_{s(\t{\sigma})}$ is still surjective, because for each $\rho>s(\t{\sigma})$, $s(\t{\sigma})$ lies in the boundary of $\rho$. Since $\t{X}$ is the subdivision of $X$, there must be a cell $\t{\rho}\in\t{X}$ as a part of $\rho$ containing the part of boundary where $\t{\sigma}$ lies. Therefore, the method in Proposition \ref{prop: subdivided sheaf has same cohomology} still applies, giving for each $p<t$,
\begin{equation}
H^p(\t{U}_{\t{\sigma}},s^*\F)\cong H^p(U_{s(\t{\sigma})},\F)=0
\end{equation}
which means that $s^*\F$ is also locally acyclic.
\end{proof}

However, we really need a chain map rather than a solely isomorphism between cohomology groups in order to construct logical gate. Therefore we introduce the following method, as well to provide an alternative way to understand the pullback sheaf when $\F$ is locally acyclic, which is proposed in \cite{lin2024transversalnoncliffordgatesquantum}.

Note that any $t$-simplex is uniquely contained in a $t$-cell. As a result, any $\tilde{\rho} \in \tilde{X}_{\geq \tilde{\sigma}}(t)$ is uniquely contained in a cell $\rho \in X_{\geq \tau_{\tilde{\sigma}}}(t)$. Then the following map is well-defined:
\begin{align}
	I: \tilde{X}_{\geq \tilde{\sigma}}(t) \rightarrow X_{\geq \tau_{\tilde{\sigma}}}(t).
\end{align}
It is also surjective as any $\rho \in X_{\geq \tau_{\tilde{\sigma}}}(t)$ must contain some $t$-simplex where $\tilde{\sigma}$ belongs to. Applying the functor $\text{Hom}(-,\mathbb{F})$, we obtain the following injective linear map:
\begin{align}
	I^\ast: \mathbb{F}^{U_{\tau_{\tilde{\sigma}}}(t) } \rightarrow \mathbb{F}^{\tilde{U}_{\tilde{\sigma}}(t) },
\end{align} 
where for a set $V$, we write $\mathbb{F}^V$ as the linear space of functions $V\rightarrow\mathbb{F}$.

Now given $\tilde{\sigma} \leq \tilde{\sigma}' \leq \tilde{\sigma}''$, then by definition, $\tau_{\tilde{\sigma}} \leq \tau_{\tilde{\sigma}'} \leq \tau_{\tilde{\sigma}''}$ and $U_{\tau_{\tilde{\sigma}}}(t) \supseteq U_{\tau_{\tilde{\sigma}'}}(t) \supseteq U_{\tau_{\tilde{\sigma}''}}(t)$. The following commutative diagram holds by the inclusion relation
\begin{equation}
	\begin{tikzcd}
		U_{\tau_{\tilde{\sigma}}}(t) &
		U_{\tau_{\tilde{\sigma}'}}(t) \arrow[l,hook'] &
		U_{\tau_{\tilde{\sigma}''}}(t) \arrow[l,hook'] \\
		\tilde{U}_{\tilde{\sigma}}(t) \arrow[u,"I"] &
		\tilde{U}_{\tilde{\sigma}'}(t) \arrow[l,hook'] \arrow[u,"I"] &
		\tilde{U}_{\tilde{\sigma}''}(t) \arrow[l,hook'] \arrow[u,"I"]
	\end{tikzcd}
\end{equation}
Again by applying the functor $\text{Hom}(-,\mathbb{F})$, we have
\begin{equation}
	\begin{tikzcd}
		\mathbb{F}^{U_{\tau_{\tilde{\sigma}}}(t)} \arrow[r] \arrow[d,"I^\ast"] &
		\mathbb{F}^{U_{\tau_{\tilde{\sigma}'}}(t)}  \arrow[r] \arrow[d,"I^\ast"] &
		\mathbb{F}^{U_{\tau_{\tilde{\sigma}''}}(t)}   \arrow[d,"I^\ast"] \\	
		\mathbb{F}^{\tilde{U}_{\tilde{\sigma}}(t)} \arrow[r] &
		\mathbb{F}^{\tilde{U}_{\tilde{\sigma}'}(t)} \arrow[r] &
		\mathbb{F}^{\tilde{U}_{\tilde{\sigma}''}(t)}
	\end{tikzcd}
\end{equation}
where we all the horizontal maps are restriction of domain of functions.

Suppose $\mathcal{F}$ satisfies the strong sheaf axiom. Then the local coefficient space $\mathcal{F}_{\tau_{\tilde{\sigma}}}$ here is a subspace of $\mathbb{F}^{U_{\tau_{\tilde{\sigma}}}(t)}$ and the map $\mathcal{F}_{\tau_{\tilde{\sigma}}, \tau_{\tilde{\sigma}'}}$ is defined by restriction. For any $\tilde{\sigma}$, we define $\tilde{\mathcal{F}}_{\tilde{\sigma}} \coloneqq  I^\ast \mathcal{F}_{\tau_{\tilde{\sigma}}}$ and set $\tilde{\mathcal{F}}_{\tilde{\sigma},\tilde{\sigma}'}$ by taking restriction. Then we have the following commutative diagram:
\begin{equation}
\begin{tikzcd}
	\mathcal{F}_{\tau_{\tilde{\sigma}}} \arrow[r, "\mathcal{F}_{\tau_{\tilde{\sigma}}, \tau_{\tilde{\sigma}'}}"] \arrow[d,"I^\ast"] &
	\mathcal{F}_{\tau_{\tilde{\sigma}'}}  \arrow[r,"\mathcal{F}_{\tau_{\tilde{\sigma}'}, \tau_{\tilde{\sigma}''}}"] \arrow[d,"I^\ast"] &
	\mathcal{F}_{\tau_{\tilde{\sigma}''}} \arrow[d,"I^\ast"] \\	
	\tilde{\mathcal{F}}_{\tilde{\sigma}} \arrow[r, "\tilde{\mathcal{F}}_{\tilde{\sigma},\tilde{\sigma}'}"] &
	\tilde{\mathcal{F}}_{\tilde{\sigma}'} \arrow[r, "\tilde{\mathcal{F}}_{\tilde{\sigma}',\tilde{\sigma}''}"] &
	\tilde{\mathcal{F}}_{\tilde{\sigma}''} 
\end{tikzcd}
\end{equation}
Since $I^*$ is injective, $\tilde{\mathcal{F}}_{\tilde{\sigma}}\cong \mathcal{F}_{\tau_{\tilde{\sigma}}}$. From now on, we further restrict the definition of $I^*$, and use this notion only for this isomorphism, and treat it as an invertible matrix. Note that by definition $\t{\F}\cong s^*\F$.

\begin{definition}[Cellular map]
	Suppose $f: X \rightarrow Y$ is a continuous map between two cell complexes. It is called a \te{cellular map} if it maps the $p$-skeleton $X^p$ of $X$ to that of $Y$ for any $p \geq 0$: $f(X^p) \subseteq Y^p$. 
\end{definition}

The following well-known fact about cellular homology is useful for our later proofs.

\begin{theorem}
	Any cellular map $f: X \rightarrow Y$ induces a chain map between the cellular chain complexes:
	\begin{align}
		f_{\#}: C_\bullet(X, \mathbb{F}) \rightarrow C_\bullet(Y, \mathbb{F}).
	\end{align}  
\end{theorem}

For a cellular map $f$, we may further assume that if $f_\#e_\alpha=\sum_\beta \lambda_{\alpha\beta}\tilde{e}_\beta$, then each $\lambda_{\alpha,\beta}\in \{0,\pm 1\}$ for simplicity. Actually, a general coefficient is allowed and all the results remain. We will also write $f(e_\alpha)\coloneqq \{\t{e}_\beta:\lambda_{\alpha\beta}\neq 0\}$. For cohomology,
\begin{equation}
f^\#: C^\bullet(Y, \mathbb{F}) \rightarrow C^\bullet(X, \mathbb{F}),
\end{equation}
is defined by applying the functor $\text{Hom}(-,\mathbb{F})$ to $f_\#$, or simply saying that $f^\#$ is the matrix transpose of $f_\#$. Both of them induce maps on the (co)homology, respectively:
\begin{equation}
f_*:H_\bullet(X,\mathbb{F})\rightarrow H_\bullet(\t{X},\mathbb{F}),\quad f^*:H^\bullet(\t{X},\mathbb{F})\rightarrow H^\bullet(X,\mathbb{F}).
\end{equation}

Let $S:X\rightarrow\t{X}$ be the identity map between topological spaces. Then $S$ is a cellular map. We will call this \te{subdivision map}, which induces a chain map $S_\#$. However, the identity map $\t{X}\rightarrow X$ is not cellular, since we created more cells in $\t{X}$ doing subdivision. By cellular approximation theorem, this map is homotopic to a cellular map. However, this will not be enough for our purpose. Thus we introduce the following definition

\begin{definition}[Approximate inverse]\label{approximate inverse}
	An \te{approximate inverse} of the subdivision map $S$ is a cellular map $A: \tilde{X} \rightarrow X$ such that for any $\tilde{\sigma}$ and any $\sigma \in A(\tilde{\sigma})$, $\sigma \leq \tau_{\tilde{\sigma}}$. Furthermore, $A_\#S_\#=\id_{C_\bullet(X,\mathbb{F})}$. We will say $\t{X}$ is a \te{simplicial approximation} of $X$ if such $A$ exists.
\end{definition}

The approximate inverse always exists. Intuitively, for each cell $\sigma\in X$, we only need to choose a particular cell $\t{\sigma}\in\t{X}$ contained in $\sigma$ and stretch it to cover $\sigma$, then map the remaining part to the boundary of $\sigma$. The special case on cubical complexes is given in Section \ref{subsec:explicit}. A proof of the existence statement in the general setting is given in \cite[Proposition~3.11]{LLL2026nontrivial}.

Now we want to extend the definition of chain map $S_\#$ and $A_\#$ to sheaved chain complexes. We abuse the notation (note that this will not cause confusion since the maps are uniquely determined by their domains) and give the definition. 
\begin{equation}
S_\#: C_\bullet(X,\mathcal{F}) \longrightarrow C_\bullet(\tilde{X},\tilde{\mathcal{F}}),
\end{equation}
by, for $\sigma\in X$ and $x(\sigma)\in \F_\sigma$,

\begin{align}\label{eq:S_sharp}
	S_\# (x(\sigma) \sigma) = \sum_{\tilde{\sigma} \in S(\sigma)} \Big[ (I^{\ast T})^{-1} \mathcal{F}_{\tau_{\tilde{\sigma}},\sigma}^T x(\sigma ) \Big] \tilde{\sigma}.
\end{align}
Notice that $\mathcal{F}_{\tau_{\tilde{\sigma}},\sigma}^T$ is well-defined because $\sigma$ is a cell containing $\tau_{\t{\sigma}}$.  We also define
\begin{equation}
A_\#:C_\bullet(\t{X},\t{\F})\longrightarrow C_\bullet(X,\F)
\end{equation}
by, for $\t{\sigma}\in \t{X}$ and $\t{x}(\t{\sigma})\in \t{\F}_{\t{\sigma}}$,
\begin{align}\label{eq:A_sharp}
	A_\# (\tilde{x}(\tilde{\sigma}) \tilde{\sigma}) = \sum_{\sigma \in A(\tilde{\sigma})} \Big[ \mathcal{F}_{\sigma,\tau_{\tilde{\sigma}}}^T I^{\ast T} \tilde{x}(\tilde{\sigma}) \Big] \sigma.
\end{align}
Here $\mathcal{F}_{\sigma,\tau_{\tilde{\sigma}}}$ is well-defined due to the definition of approximate inverse.

\begin{proposition}\label{prop:A_S_chain_maps}
	The $S_\#, A_\#$ defined above are chain maps. 
\end{proposition}
\begin{proof}
	We can prove the statements on the standard basis vectors from the chain. That is, we start by taking $x = x(\sigma) \cdot \sigma$. By definition,
	\begin{align}
		\begin{aligned}
			S_\# \partial (x(\sigma) \sigma) & = S_\# (\sum_{\sigma' \lessdot \sigma} \F_{\sigma',\sigma}^T x(\sigma) \sigma' ) \\
			& = \sum_{\sigma' \lessdot \sigma} \sum_{\tilde{\sigma}' \in S(\sigma')} \Big[ (I^{\ast T})^{-1} \mathcal{F}_{\tau_{\tilde{\sigma}'},\sigma'}^T  (\F_{\sigma',\sigma}^T x(\sigma) ) \Big] \tilde{\sigma}' \\
			& = \sum_{\sigma' \lessdot \sigma} \sum_{\tilde{\sigma}' \in S(\sigma')} \Big[ (I^{\ast T})^{-1} \mathcal{F}_{\tau_{\tilde{\sigma}'},\sigma}^T  x(\sigma) \Big] \tilde{\sigma}'.
		\end{aligned}
	\end{align}
	On the other hand,
	\begin{align}
		\begin{aligned}
			\partial S_\# (x(\sigma) \sigma) & = \partial \Big( \sum_{\tilde{\sigma} \in S(\sigma)} \Big[ (I^{\ast T})^{-1} \mathcal{F}_{\tau_{\tilde{\sigma}},\sigma}^T x(\sigma ) \Big] \tilde{\sigma} \Big) \\
			& = \sum_{\tilde{\sigma} \in S(\sigma)} \sum_{\tilde{\sigma}' \lessdot \tilde{\sigma}} \Big[ \mathcal{F}_{\tilde{\sigma}',\tilde{\sigma}}^T (I^{\ast T})^{-1} \mathcal{F}_{\tau_{\tilde{\sigma}},\sigma}^T x(\sigma ) \Big] \tilde{\sigma}' \\
			& = \sum_{\tilde{\sigma} \in S(\sigma)} \sum_{\tilde{\sigma}' \lessdot \tilde{\sigma}} \Big[ (I^{\ast T})^{-1} \mathcal{F}_{\tau_{\tilde{\sigma}'},\tau_{\tilde{\sigma}}}^T \mathcal{F}_{\tau_{\tilde{\sigma}},\sigma}^T x(\sigma ) \Big] \tilde{\sigma}'\\
            &=\sum_{\tilde{\sigma}\in S(\sigma)}\sum_{\tilde{\sigma}'\lessdot\tilde{\sigma}}\Big[(I^{*T})^{-1}\mathcal{F}^T_{\tau_{\tilde{\sigma}'},\sigma}x(\sigma)\Big]\tilde{\sigma}',
		\end{aligned}
	\end{align}
	where we used the fact that $\mathcal{F}_{\tilde{\sigma}',\tilde{\sigma}}^T (I^{\ast T})^{-1} =  (I^{\ast T})^{-1} \mathcal{F}_{\tau_{\tilde{\sigma}'},\tau_{\tilde{\sigma}}}^T$ by commutativity. Note that, since $S_\#:C_\bullet(X,\mathbb{F})\rightarrow C_\bullet(\tilde{X},\mathbb{F})$ is a chain map, we have
	\begin{align}
		\sum_{\sigma'\lessdot\sigma}\sum_{\tilde{\sigma}'\in S(\sigma')}\tilde{\sigma}'=\sum_{\tilde{\sigma}\in S(\sigma)}\sum_{\tilde{\sigma}'\lessdot\tilde{\sigma}}\tilde{\sigma}'.
	\end{align}
	Therefore, $S_\#\partial(x(\sigma)\cdot\sigma)=\partial S_\#(x(\sigma)\cdot\sigma)$, and $S_\#:C_\bullet(X,\mathcal{F})\rightarrow C_\bullet(\tilde{X},\tilde{\mathcal{F}})$ is a chain map.
	
	For $A_\ast$,
	\begin{align}
		\begin{aligned}
			A_\# \partial (\tilde{x}(\tilde{\sigma}) \tilde{\sigma})  
			& = A_\# ( \sum_{\tilde{\sigma}' \lessdot \tilde{\sigma}} \tilde{\mathcal{F}}_{\tilde{\sigma}',\tilde{\sigma}}^T \tilde{x}(\tilde{\sigma}) \tilde{\sigma} )  \\
			& = \sum_{\tilde{\sigma}' \lessdot \tilde{\sigma}} \sum_{\sigma' \in A(\tilde{\sigma}')} \Big[ \mathcal{F}_{\sigma',\tau_{\tilde{\sigma}'}}^T I^{\ast T} \tilde{\mathcal{F}}_{\tilde{\sigma}',\tilde{\sigma}}^T \tilde{x}(\tilde{\sigma}) \Big] \sigma' \\
			& = \sum_{\tilde{\sigma}' \lessdot \tilde{\sigma}} \sum_{\sigma' \in A(\tilde{\sigma}')} \Big[ \mathcal{F}_{\sigma',\tau_{\tilde{\sigma}'}}^T \mathcal{F}_{\tau_{\tilde{\sigma}'},\tau_{\tilde{\sigma}}}^T  I^{\ast T} \tilde{x}(\tilde{\sigma}) \Big] \sigma'\\
            &=\sum_{\tilde{\sigma}'\lessdot\tilde{\sigma}}\sum_{\sigma'\in A(\tilde{\sigma}')}\Big[\mathcal{F}^T_{\sigma',\tau_{\tilde{\sigma}}}I^{*T}\tilde{x}(\tilde{\sigma})\Big]\sigma'
		\end{aligned}
	\end{align}
	And
	\begin{align}
		\begin{aligned}
			\partial A_\# (\tilde{x}(\tilde{\sigma}) \tilde{\sigma}) & = \partial \Big( \sum_{\sigma \in A(\tilde{\sigma})} \Big[ \mathcal{F}_{\sigma,\tau_{\tilde{\sigma}}}^T I^{\ast T} \tilde{x}(\tilde{\sigma}) \Big] \sigma \Big) \\
			& = \sum_{\sigma \in A(\tilde{\sigma})} \sum_{\sigma' \lessdot \sigma} \Big[ \mathcal{F}_{\sigma',\sigma}^T \mathcal{F}_{\sigma,\tau_{\tilde{\sigma}}}^T I^{\ast T} \tilde{x}(\tilde{\sigma}) \Big] \sigma' \\
			& = \sum_{\sigma \in A(\tilde{\sigma})} \sum_{\sigma' \lessdot \sigma} \Big[ \mathcal{F}_{\sigma',\tau_{\tilde{\sigma}}}^T I^{\ast T} \tilde{x}(\tilde{\sigma}) \Big] \sigma'.
		\end{aligned}
	\end{align}
	Again, since $A_\#: C_\bullet(\tilde{X},\mathbb{F}) \rightarrow C_\bullet(X,\mathbb{F})$ is a chain map,
	\begin{align}
		\sum_{\tilde{\sigma}' \lessdot \tilde{\sigma}} \sum_{\sigma' \in A(\tilde{\sigma}')} \sigma' = \sum_{\sigma \in A(\tilde{\sigma})} \sum_{\sigma' \lessdot \sigma} \sigma'
	\end{align}
	and this completes the proof.
\end{proof}

We also need to extend the definition to tensor sheaf. Let $\mathcal{F}, \mathcal{G}$ be sheaves generated by some classical codes. Then we can define $\tilde{\mathcal{F}}$ and $\tilde{\mathcal{G}}$ separately. Then we extend the definition as follows
	\begin{align}
		S_\#: C_\bullet(X,\mathcal{F} \otimes \mathcal{G} ) \rightarrow C_\bullet(\tilde{X},\tilde{\mathcal{F}} \otimes \tilde{\mathcal{G}} )
	\end{align}
	is defined by
	\begin{align}
		S_\# (x(\sigma) \sigma) = \sum_{\tilde{\sigma} \in S(\sigma)} \Big[ \Big( (I^{\ast T})^{-1} \otimes (I^{\ast T})^{-1} \Big) \Big( \mathcal{F}_{\tau_{\tilde{\sigma}},\sigma}^T \otimes \mathcal{G}_{\tau_{\tilde{\sigma}},\sigma}^T \Big) x(\sigma) \Big] \tilde{\sigma}, 
	\end{align}
	Here $x(\sigma)\in \F_\sigma\otimes\mathcal{G}_\sigma$ can be spanned by simple tensors like $x_{\mathcal{F}}(\sigma) \otimes x_{\mathcal{G}}(\sigma)$ and $I^\ast$ is defined independent of the choice of the sheaves $\mathcal{F}$ and $\mathcal{G}$ (both as subspaces of $\mathbb{F}^{X_{\geq \tau_{\tilde{\sigma}}}(t)}$). Similarly,
	\begin{align}
		A_\#: C_\ast(\tilde{X},\tilde{\mathcal{F}} \otimes \tilde{\mathcal{G}} ) \rightarrow C_\ast(X,\mathcal{F} \otimes \mathcal{G} )
	\end{align}
	is defined by
	\begin{align}
		A_\# (\tilde{x}(\tilde{\sigma}) \tilde{\sigma}) = \sum_{\sigma \in A(\tilde{\sigma})} \Big[ \Big( \mathcal{F}_{\sigma,\tau_{\tilde{\sigma}}}^T \otimes \mathcal{G}_{\sigma,\tau_{\tilde{\sigma}}}^T \Big) \Big( I^{\ast T} \otimes I^{\ast T} \Big) \tilde{x}(\tilde{\sigma}) \Big] \sigma.
	\end{align}
	
Now, we are able to define cup and cap products on cell complexes and. Let $\t{X}$ be a simplicial subdivision of $X$, then the cup product on $X$ is defined by the following commutative diagram
\begin{equation}
    \begin{tikzcd}
    C^p(X,\mathcal{F}) \times C^q(X,\mathcal{G})
  \arrow[r, "\smile"]
  \arrow[d, "A^\# \times A^\#"']
& C^{p+q}(X,\mathcal{F}\otimes\mathcal{G})
   \\
C^p(\tilde{X},\t{\mathcal{F}}) \times C^q(\tilde{X},\t{\mathcal{G}})
  \arrow[r, "\smile"]
& C^{p+q}(\tilde{X},\t{\mathcal{F}}\otimes\t{\mathcal{G}}) \arrow[u, "S^\#" ]
     \end{tikzcd}
\end{equation}
One may check by direct calculation that the Leibniz rule still holds by commutativity of the chain maps with the coboundary operator
\begin{align}
	\delta(\alpha \smile \beta) 
		=\delta \alpha \smile \beta + \alpha \smile \delta \beta.
\end{align}

The cap product is defined by the following commutative diagram
\begin{equation}
    \begin{tikzcd}
    C^p(X,\mathcal{F}) \times C_{p+q}(X,\mathcal{F}\otimes\mathcal{G})
  \arrow[r, "\frown"]
  \arrow[d, "A^\# \times S_\#"']
& C_q(X,\mathcal{G})
   \\
C^p(\tilde{X}, \tilde{\mathcal{F}}) \times C_{p+q}(\tilde{X}, \tilde{\mathcal{F}} \otimes \tilde{\mathcal{G}} )
  \arrow[r, "\frown"]
& C_q(\tilde{X},\t{\mathcal{G}}) \arrow[u, "A_\#" ]
     \end{tikzcd}
\end{equation}
Note that we still have
\begin{align}
	\partial(\alpha \frown x)=\delta \alpha \frown x + \alpha \frown \partial x.
\end{align}

Note that the chain-level definition of cup and cap products inevitably depends on the choice of approximate inverse. Nevertheless, the (co)homology-level definitions are independent of this choice, since $H^\bullet(X,\F)$ and $H^\bullet(\t X,\t\F)$ are isomorphic finite dimensional vector spaces and $A^*$ is a right inverse of $S^*$.


The following properties ensure that our construction of logical operation has constant-depth.

\begin{lemma}
	If the cell complex $X$ is sparse, then $\tilde{X}$ is also sparse.
\end{lemma}
\begin{proof}
	We only consider the case for cubical complexes. The subdivision of each cube has a constant number of simplices. Then by the inclusion relation and sparsity of $X$, we can prove the claim.  
\end{proof}

\begin{lemma}
	Suppose both $X$ and $\tilde{X}$ are sparse and each cell in $X$ is subdivided into constant many simplices, then the matrix representations of the operators $S_\#, S^\#$, $A_\#, A^\#$ are sparse. 
\end{lemma}
\begin{proof}
	This is also quite straightforward, as the matrix representation of $I^\ast$ is sparse due to the sparsity of both $X$ and $\tilde{X}$.
\end{proof}

The importance of our Definition \ref{approximate inverse} is that, together with Proposition \ref{prop: subdivided sheaf has same cohomology}, $S^*$ will be an isomorphism between (co)homology groups
\begin{equation}
    H^\bullet(\t{X},\t{\F})\overset{S^*}{\underset{\cong}{\longrightarrow}} H^\bullet(X,\F),
\end{equation}
where $A^*=(S^*)^{-1}$ is the inverse. This will be crucial in providing a lower bound of $k_{\CZ}$ in the construction of logical gate later.


\subsection{Poincar\'e duality via cap product}\label{subsec:duality-cap-product}


The cap product further uncovers the isomorphism $H^i(X,\mathcal{F}^\perp) \cong H_{t-i}(X,\mathcal{F})$ in Theorem \ref{thm:Poincaré_duality}. This insight will also be used to build the logical $\CZ$ and $\CCZ$ circuits, and to evaluate the lower bound on their logical action subrank. 
\begin{theorem}[Cap product induces Poincar\'e duality map]
	\label{thm:Cap induced poincare dual map}
	Suppose $X$ is a cell complex that admits a simplicial approximation. Let $[X]\in C_t(X,\mathcal{F}^\perp\!\!\otimes\mathcal{F})$ be the fundamental class defined by $[X]=\sum_{\sigma\in X(t)}\sigma$, then for each $0\leq i\leq t$, there is an isomorphism $D$ 
	\begin{equation}
		D:H^i(X,\mathcal{F}^\perp)\longrightarrow H_{t-i}(X,\mathcal{F}),
	\end{equation}
	given by $D[\alpha]=[\alpha]\frown [X]$.
\end{theorem}
Before proving the theorem, we want to provide some remarks. In algebraic topology, there is a well-known Poincar\'e duality for manifold \cite{Bott:1982xhp,Hatcheralgtop}. Suppose $M$ is a $n$-dimensional closed manifold, then there exist a homology class $[M] \in H_n(M,\mathbb{Z})$ which induces an isomorphism 
\begin{equation}
\begin{tikzcd}
    H^i(M,\mathbb{Z})\arrow[r,"\frown {[M]}"]&H_{n-i}(M,\mathbb{Z}).
\end{tikzcd}
\end{equation}

As an extension to Theorem \ref{thm:Poincaré_duality}, we want to find an analog for quantum error-correcting codes. For cubical complexes, the ``fundamental class" defined in Theorem \ref{thm:Cap induced poincare dual map} is naturally a cycle in $C_t(X,\mathcal{F}^\perp\!\!\otimes\mathcal{F})$. Together with a simple intuition from the isomorphism $H^0(X,\mathcal{F}^\perp) \cong H_t(X,\mathcal{F})$ when $X$ is a simplicial complex, this observation motivates the proof of Theorem \ref{thm:Cap induced poincare dual map}. To be precise, suppose $\alpha\in C^0(X,\mathcal{F}^\perp)$ is a cocycle, then $\delta^\perp\alpha=0$. Recall the diagram chase in Section \ref{sec:duality-proof}, by Eq.~\eqref{eq:d'_1}, Proposition \ref{prop:h_exact} and \ref{prop:h_naturality}, $d' h'' \alpha = 0$. Then by exactness of $h''$, we can find an element $x \in C_t(X,\mathcal{F})$ such that $\eta' x = h'' \alpha$. Since $d' h'' \alpha=0$, $x$ is a cycle in $C_t(X,\mathcal{F})$. As in the proof of Theorem \ref{thm:Poincaré_duality}, it is easy to check that the map $[\alpha]\mapsto[x]$ is a well-defined isomorphism between $H^0(X,\mathcal{F}^\perp)$ and $H_t(X,\mathcal{F})$. Formally, we can write $x=\sum_{\tau\in X(t)}\mathcal{F}^\perp_{\sigma(\tau),\tau}\alpha(\sigma(\tau))\cdot\tau$, where $\sigma(\tau)\in X(0)$ is a function of $\tau$. It can be chosen to be any element in $X(0)$ as long as it is under $\tau$. As a result, it is natural to choose $\sigma(\tau)={}_0\tau$, and then $x=\sum_{\tau\in X(t)}\mathcal{F}^\perp_{{}_0\tau,\tau}\alpha({}_0\tau)\cdot \tau$ is exactly $\alpha\frown [X]$. This inspires us to explore whether the argument holds for all dimensions, and it turns out to be true.

\begin{proof}[Proof of Theorem \ref{thm:Cap induced poincare dual map}]
We first prove the theorem for the case when $X$ is a simplicial complex. Before, we used a zig-zag method to construct the isomorphism $D$. Now we are going to prove that this map is equal to taking the cap product in Eq.~\eqref{eq:cap_other_definition}. Recall that by zig-zag method, for a cocycle $\alpha\in C^i(X,\mathcal{F}^\perp)$, we can find $\alpha^{k,t-i+k}\in C^k(X,\mathcal{F}_{t-i+k})$ for $0\leq k\leq i$ and $\alpha^{k,t-i+k+1}\in C^k(X,\mathcal{F}_{t-i+k+1})$ for $0\leq k\leq i-1$ such that 
    \begin{equation}
    d'\alpha^{k-1,t-i+k}=\alpha^{k,t-i+k}=d''\alpha^{k,t-i+k+1},\,\alpha^{i,t}=h'' \alpha,\,\alpha^{0,t-i}=\eta' D\alpha,\,\alpha^{i,t}=h'' \alpha .   \end{equation}
    By definition, we have 
    \begin{equation}
    \alpha\frown[X]=\sum_{\rho^t\in X(t)}\mathcal{F}^T_{\rho^t_{t-1},\rho^t}\,\mathcal{F}^\perp_{{}_i\rho^t,\rho^t}\alpha({}_i\rho^t)\cdot\rho^t_{t-1}.
    \end{equation}
    On the other hand, since for each $\rho^0\in X(0), \rho^{t-i}\in X(t-i), \rho^0\leq\rho^{t-i}$, we have $(D\alpha)(\rho^t)=\alpha^{0,t-i}(\rho^0,\rho^{t-i})$. Note that for fixed $\rho^{t-i}$ the value of $\alpha^{0,t-i}(\rho^0,\rho^{t-i})$ would be the same  and different choices of $\rho^0$ as long as $\rho^0\leq \rho^{t-i}$, therefore we may write $\rho^0=\rho^0(\rho^{t-i})$ as a function of $\rho^{t-i}$. Hence we may write
    \begin{equation}
    D\alpha=\sum_{\rho^{t-i}\in X(t-i)}\alpha^{0,t-i}(\rho^0(\rho^{t-1}),\rho^{t-i})\cdot \rho^{t-i}.
    \end{equation}
    Since $\alpha^{0,t-i}=d''\alpha^{0,t-i+1}$, we have
    \begin{align}
    \begin{aligned}
        D\alpha&=\sum_{\rho^{t-i}\in X(t-i)}\sum_{\rho^{t-i+1}\gtrdot\rho^{t-i}}\mathcal{F}^T_{\rho^{t-i},\rho^{t-i+1}}\alpha^{0,t-i+1}(\rho^0(\rho^{t-i}),\rho^{t-i+1})\cdot\rho^{t-i}\\
    &=\sum_{\rho^{t-i+1}\in X(t-i+1)}\sum_{\rho^{t-i}\lessdot\rho^{t-i+1}}\mathcal{F}^T_{\rho^{t-i},\rho^{t-i+1}}\alpha^{0,t-i+1}(\rho^0(\rho^{t-i}),\rho^{t-i+1})\cdot\rho^{t-i}\\
    &=\sum_{\rho^{t-i+1}\in X(t-i+1)}\sum_{{}_0\rho^{t-i+1}\lessdot\rho^{t-i}\lessdot\rho^{t-i+1}}\mathcal{F}^T_{\rho^{t-i},\rho^{t-i+1}}\alpha^{0,t-i+1}(({}_0\rho^{t-i+1}),\rho^{t-i+1})\cdot\rho^{t-i}\\
    &+\sum_{\rho^{t-i+1}\in X(t-i+1)}\mathcal{F}^T_{\rho^{t-i+1}_{t-i},\rho^{t-i+1}}\alpha^{0,t-i+1}({}_0(\rho^{t-i+1}_{t-i}),\rho^{t-i+1})\cdot\rho^{t-i+1}_{t-i}.
    \end{aligned}
    \end{align}
Where we set $\rho^0(\rho^{t-i})={}_0\rho^{t-i+1}$ when $\rho^{t-i}\geq {}_0\rho^{t-i+1}$, and $\rho^0(\rho^{t-i})={}_0(\rho^{t-i+1}_{t-1})$ otherwise. An important observation is that, from the fact that $d'\alpha^{0,t-i+1}=\alpha^{1,t-i+1}$, we get
\begin{equation}
\alpha^{0,t-i+1}({}_0\rho^{t-i+1},\rho^{t-i+1})+\alpha^{0,t-i+1}({}_0(\rho^{t-i+1}_{t-1}),\rho^{t-i+1})=\alpha^{1,t-i+1}({}_1\rho^{t-i+1},\rho^{t-i+1})
\end{equation}
Hence we get
\begin{align}
    \begin{aligned}
    D\alpha&=\sum_{\rho^{t-i+1}\in X(t-i+1)}\sum_{\rho^{t-i}\lessdot\rho^{t-i+1}}\mathcal{F}^T_{\rho^{t-i},\rho^{t-i+1}}\alpha^{0,t-i+1}(({}_0\rho^{t-i+1}),\rho^{t-i+1})\cdot\rho^{t-i}\\
    &+\sum_{\rho^{t-i+1}\in X(t-i+1)}\mathcal{F}^T_{\rho^{t-i+1}_{t-i},\rho^{t-i+1}}\alpha^{1,t-i+1}({}_1\rho^{t-i+1},\rho^{t-i+1})\cdot\rho^{t-i+1}_{t-i}\\
    &=\sum_{\rho^{t-i+1}\in X(t-i+1)}\partial(\alpha^{0,t-i+1}(({}_0\rho^{t-i+1}),\rho^{t-i+1})\cdot\rho^{t-i+1})\\
    &+\sum_{\rho^{t-i+1}\in X(t-i+1)}\mathcal{F}^T_{\rho^{t-i+1}_{t-i},\rho^{t-i+1}}\alpha^{1,t-i+1}({}_1\rho^{t-i+1},\rho^{t-i+1})\cdot\rho^{t-i+1}_{t-i}.    
    \end{aligned}
\end{align}
Hence gives the equality of homology class
\begin{equation}
[D\alpha] = \Big[ \sum_{\rho^{t-i+1}\in X(t-i+1)}\mathcal{F}^T_{\rho^{t-i+1}_{t-i},\rho^{t-i+1}}\alpha^{1,t-i+1}({}_1\rho^{t-i+1},\rho^{t-i+1})\cdot\rho^{t-i+1}_{t-i} \Big].
\end{equation}
Now we are going to use an inductive method. Suppose for $\forall 1\leq j\leq k$, we have
\begin{equation}
[D\alpha] = \Big[ \sum_{\rho^{t-i+j}\in X(t-i+j)}\mathcal{F}^T_{\rho^{t-i+j}_{t-i},\rho^{t-i+j}}\alpha^{j,t-i+j}({}_j\rho^{t-i+j},\rho^{t-i+j})\cdot\rho^{t-i+j}_{t-i} \Big].
\end{equation}
We are going to prove that the same form holds for $j=k+1$. Similarly, since $\alpha^{k,t-i+k}=d''\alpha^{k,k-i+k+1}$, we have
\begin{equation}
[D\alpha] = \Big[ \hspace{-1cm} \sum_{\rho^{t-i+k+1}\in X(t-i+k+1)}\sum_{\rho^{t-i+k}\lessdot\rho^{t-i+k+1}}\mathcal{F}^T_{\rho^{t-i+k}_{t-i}.\rho^{t-i+k+1}}\alpha^{k,t-i+k+1}({}_k\rho^{t-i+k},\rho^{t-i+k+1})\cdot\rho^{t-i+k}_{t-i} \Big].
\end{equation}
Now we consider a fixed $\rho^{t-i+k+1}=[v_0,v_1,\cdots,v_{t-i+k+1}]$. Then the set $\{\rho^{t-i+k}:\rho^{t-i+k}\lessdot\rho^{t-i+k+1}\}$ can be divided into two parts $\{\rho^{t-i+k+1}\setminus v_s:0\leq s\leq k\}\sqcup\{\rho^{t-i+k+1}\setminus v_s:k+1\leq s\leq t-i+k+1\}$. For $0\leq s\leq k$, we have 
\begin{equation}
(\rho^{t-i+k+1}\setminus v_s)_{t-i}=(\rho^{t-i+k+1})_{t-i},\quad {}_k(\rho^{t-i+k+1}\setminus v_s)=({}_{k+1}(\rho^{t-i+k+1}))\setminus v_s.
\end{equation}
And for $k+1\leq s\leq t-i+k+1$, we have 
\begin{equation}
(\rho^{t-i+k+1}\setminus v_s)_{t-i}=(\rho^{t-i+k+1}_{t-i+1})\setminus v_s,\quad {}_k(\rho^{t-i+k+1}\setminus v_s)={}_k(\rho^{t-i+k+1}).
\end{equation}

Therefore,
\begin{align}
    \begin{aligned}
    &\sum_{\rho^{t-i+k}\lessdot\rho^{t-i+k+1}}\mathcal{F}^T_{\rho^{t-i+k}_{t-i}.\rho^{t-i+k+1}}\alpha^{k,t-i+k+1}({}_k\rho^{t-i+k},\rho^{t-i+k+1})\cdot\rho^{t-i+k}_{t-i}\\
    &=\sum_{s=0}^{s=k}\mathcal{F}^T_{\rho^{t-i+k+1}_{t-i},\rho^{t-i+k+1}}\alpha^{k,t-i+k+1}(({}_{k+1}\rho^{t-i+k+1})\setminus v_s,\rho^{t-i+k+1})\cdot \rho^{t-i+k+1}_{t-i}\\
    &+\sum_{s=k+1}^{s=t-i+k+1}\mathcal{F}^T_{(\rho^{t-i+k+1}_{t-i+1})\setminus v_s,\rho^{t-i+k+1}}\alpha^{k,t-i+k+1}({}_{k}\rho^{t-i+k+1},\rho^{t-i+k+1})\cdot(\rho^{t-i+k+1}_{t-i+1})\setminus v_s\\
    &=\sum_{\rho^k\in X_{\leq{}_{k+1}\rho^{t-i+k+1}}(k)}\mathcal{F}^T_{\rho^{t-i+k+1}_{t-i},\rho^{t-i+k+1}}\alpha^{k,t-i+k+1}(\rho^k,\rho^{t-i+k+1})\cdot \rho^{t-i+k+1}_{t-i}\\
    &+\partial(\mathcal{F}^T_{\rho^{t-i+k+1}_{t-i+1},\rho^{t-i+k+1}}\alpha^{k,t-i+k+1}({}_{k}\rho^{t-i+k+1},\rho^{t-i+k+1})\cdot\rho^{t-i+k+1}_{t-i+1}),    
    \end{aligned}
\end{align}
where the last equality is done by adding
\begin{equation}
\mathcal{F}^T_{\rho^{t-i+k+1}_{t-i+1}\setminus v_k,\rho^{t-i+k+1}}\alpha^{k,t-k+i+1}({}_k\rho^{t-i+k+1},\rho^{t-i+k+1})\cdot\rho^{t-i+k+1}_{t-i+1}\setminus v_k
\end{equation}
twice. Note that the condition $d'\alpha^{k,t-i+k+1}=\alpha^{k+1,t-i+k+1}$ gives us the equation
\begin{equation}
\sum_{\rho^k\in X_{\leq{}_{k+1}\rho^{t-i+k+1}}(k)}\alpha^{k,t-i+k+1}(\rho^k,\rho^{t-i+k+1})=\alpha^{k+1,k-i+k+1}({}_{k+1}\rho^{t-i+k+1},\rho^{t-i+k+1}).
\end{equation}
Therefore, 
\begin{equation}
[D\alpha] = \Big[\sum_{\rho^{t-i+k+1}\in X(t-i+k+1)}\mathcal{F}^T_{\rho^{t-i+k+1}_{t-i},\rho^{t-i+k+1}}\alpha^{k+1,t-i+k+1}({}_{k+1}\rho^{t-i+k+1},\rho^{t-i+k+1})\cdot\rho^{t-i+k+1}_{t-i} \Big]
\end{equation}
By induction, we get
\begin{align}
\begin{aligned}
    [D\alpha]& = \Big[\sum_{\rho^t\in X(t)}\mathcal{F}^T_{\rho^t_{t-i},\rho^t}\alpha^{i,t}({}_i\rho^t,\rho^t)\cdot\rho^t_{t-1} \Big]
  & = \Big[\sum_{\rho^t\in X(t)}\mathcal{F}^T_{\rho^t_{t-i},\rho^t}\,\mathcal{F}^{\perp}_{{}_i\rho^t,\rho^t}\alpha({}_i\rho^t)\cdot\rho^t_{t-i} \Big] \\
  &=[\alpha]\frown[X],
\end{aligned}
\end{align}
where the second line is obtained by $\alpha^{i,t}=h'' \alpha$.

Now, let us consider a general cell complex $X$ with simplicial approximation $\t{X}$, the following commutative diagram completes the proof 
\begin{equation}
    \begin{tikzcd}
        H^i(X,\F^\perp)\arrow[r,"\frown {[X]}"]\arrow[d,"A^*"',"\cong"] & H_{t-i}(X,\F)\\
        H^i(\t{X},\t{\F}^\perp)\arrow[r,"\frown {[\t{X}]}","\cong"'] & H_{t-i}(\t{X},\t{\F}) \arrow[u,"A_*","\cong"']
    \end{tikzcd}
\end{equation}
\end{proof}

\begin{definition}[Dual pairing]
  Given two $\mathbb{F}$-linear space $V$ and $W$, a \te{dual pairing} between them is a map $\phi:V\times W\rightarrow\mathbb{F}$ such that $\phi(V)\cong W^*, \phi(W)\cong V^*$.
\end{definition}

\begin{proposition}\label{dual pairing between homology groups}
For a chain complex $C_*$ and a cochain complex $C^*$, the pairing $\langle-,-\rangle$ between (co)homology groups $H^i$ and $H_i$ is a dual pairing.
\end{proposition}
\begin{proof}
    Let $\alpha$ be a cocycle in $C^i$, and suppose that $\alpha(x)=\langle\alpha,x\rangle=0$ for all $x\in \ker \partial_{i}$. We show that $[\alpha]=0$, i.e., that $\alpha$ is a coboundary. To see this, define $\beta\in C^{i-1}$ by setting $\beta(\partial y)\coloneqq \alpha(y)$ for each $y\in C_i$. This $\beta$ is well-defined: if $y,y'\in C_i$ satisfy $\partial y'=\partial y$, then $y'-y\in \ker \partial_{i}$, and consequently $\beta(\partial y')-\beta(\partial y)=\alpha(y'-y)=0$. Moreover, $(\delta \beta)(y)=\beta(\partial y)=\alpha(y)$ for all $y$, so $\alpha = \delta\beta$ and $[\alpha]=0$. 

    Thus, the map $[\alpha]\mapsto\langle[\alpha],-\rangle\in H_i^*$ is injective. Since $\dim H^i=\dim H_i^*$, this map is actually bijective. A symmetric argument shows that $[x]\mapsto\langle-,[x]\rangle$ is also bijective. Consequently, $\langle-,-\rangle$ defines a dual pairing between $H^i$ and $H_i$.
\end{proof}

\begin{corollary}\label{Poincare dual pairing}
    There is a dual pairing $P$
    \begin{equation}
    P:H^i(X,\mathcal{F}^\perp)\times H^{t-i}(X,\mathcal{F})\longrightarrow\mathbb{F},
    \end{equation}
    given by $P([\alpha],[\beta])=\langle[\alpha]\smile[\beta],[X]\rangle$.
\end{corollary}
\begin{proof}

    By direct calculation,
    \begin{equation}
    \langle\alpha\smile\beta,[X]\rangle=\sum_{\sigma\in X(t)}\mathcal{F}^\perp_{{}_p\sigma,\sigma}\alpha({}_p\sigma)\cdot\mathcal{F}_{\sigma_{t-p},\sigma}\beta(\sigma_{t-p})=\langle\beta,\alpha\frown[X]\rangle.
    \end{equation}
    As a result, the following diagram commutes
\begin{equation}
\begin{tikzcd}
    H^p(X, \mathcal{F}^\perp) \times H^{t-p}(X, \mathcal{F}) \arrow[r, "\smile"] \arrow[d,"D\times \text{id}"'] & 
    H^t(X, \mathcal{F}^\perp \otimes \mathcal{F}) \arrow[d,"\langle-{,}\,{[X]}\rangle"] \\ 
    H_{t-p}(X, \mathcal{F}) \times H^{t-p}(X, \mathcal{F}) \arrow[r,"\langle-{,}\,-\rangle"'] & 
    \mathbb{F}
\end{tikzcd}
\end{equation}
    Since $\langle-,-\rangle$ is a dual pairing on $H^{n-p}(X,\mathcal{F})\times H_{n-p}(X,\mathcal{F})$, $P$ is also a dual pairing.
\end{proof}

One can formulate the above discussion more generally via the following commutative diagram, although we will not make use of it in what follows.
\begin{equation}
\begin{tikzcd}
    H^p(X, \mathcal{F}^\perp) \times H^{q}(X, \mathcal{F}) \arrow[r, "\smile"] \arrow[d,"D\times \text{id}"'] & 
    H^{p+q}(X, \mathcal{F}^\perp \otimes \mathcal{F}) \arrow[d,"\langle-{,}\,{[X]}\rangle"] \\ 
    H_{t-p}(X, \mathcal{F}) \times H^{q}(X, \mathcal{F}) \arrow[r,"\langle-{,}\,-\rangle"'] & 
    H_{t-p-q}(X,\mathbb{F})
\end{tikzcd}
\end{equation}

The dual pairing in Corollary~\ref{Poincare dual pairing} will be crucial for the bound of $k_{\CZ}$ in Theorem~\ref{thm: CZ on good qLDPC} and~\ref{thm: CZ on almost good qLTC}.


\subsection{Explicit approximate inverse}\label{subsec:explicit}

We now give an explicit construction of the approximate inverse $A$ for cubical complexes.

\begin{example}\label{example:2D_CZ}
    We give an explicit calculation for a 2-dimensional cubical complex. Let $G$ be a set with pairwise commutative permutation sets $A_1$ and $A_2$, and let $X$ be the cubical complex generated by $\{G, A_1, A_2\}$. Following the notation of \cite{Dinur2024sheaf,lin2024transversalnoncliffordgatesquantum, nguyen2025quantum}, we denote a 2-cube by $[g;a_1,a_2]$. For simplicity, we write $g$ for the point $[g;0,0]$, $g a_1$ for $[g a_1;1,0]$, $g a_2$ for $[g a_2;0,1]$, and $g a_1a_2$ for $[g a_1 a_2;1,1]$. The triangulation is obtained by adding a segment from $[g;0,0]$ to $[g  a_1 a_2; 1,1]$, and we name the two 2-simplices as $[g, g  a_1, g  a_1 a_2]$ and $[g, g  a_2, g a_1 a_2]$. The 1-simplices are named according to the boundary map. The approximate inverse $A$ is defined by expanding the simplex $[g,ga_2,ga_1a_2]$ to fill the entire square $[g;a_1,a_2]$. This automatically collapses the other 2-simplex $[g, g  a_1, g  a_1 a_2]$ to the union of the two 1-simplices $[g, g  a_1]$ and $[g  a_1, g  a_1 a_2]$.
\end{example}

\begin{example}
   We now give a general formula for $t$-dimensional cubical complexes. Let $G$ be a set equipped with pairwise commutative permutation sets $A_1,A_2,\dots,A_t$, and let $X$ be the cubical complex generated by $\{ G, A_1, A_2,\dots,A_t \}$. For elements $a_i\in A_i$ ($i=1,\dots,t$), we write $ga_{i_1}\cdots a_{i_j}$ to denote the point $[ga_{i_1}\cdots a_{i_j};b_1,\dots, b_t]$, where $b_k=1$ if $k\in\{i_1,\dots,i_j\}$ and $b_k=0$ otherwise. The $t$-cube $[g;a_1,a_2,\dots,a_t]$ is decomposed into $t!$ many $t$-simplices as follows:
    \begin{equation}
        \tilde{X}(t) =  \Bigl\{[g,ga_{\pi(1)},\dots,ga_{\pi(1)}a_{\pi(2)}\cdots a_{\pi(t)}]:
        g\in G,\;a_1\in A_1,\dots,a_t\in A_t,\;\pi\in S_t \Bigr\},
    \end{equation}
    where $S_t$ is the permutation group. Simplices of lower dimensions are obtained by applying the boundary map. The approximate inverse $A$ is defined by extending the simplex 
    \begin{equation}
        [g,ga_t,g a_ta_{t-1},\dots, g a_ta_{t-1}\cdots a_1]
    \end{equation}
    to fill the whole $t$-cube $[g;a_1, \dots, a_t]$. For convenience, set $a_0$ to be the identity. Then
    \begin{equation}
        A_\#[ga_1a_2\cdots a_i,ga_1a_2\cdots a_{i+j}] = \sum_{k=0}^{j-1}[ga_1\cdots a_{i+k},ga_1\cdots a_{i+k+1}],
    \end{equation}
    and $A_\#$ is uniquely determined by its action on 1-simplices. Although explicit formulas for $A_\#$ are difficult to write down, they are theoretically feasible to compute and can be implemented for concrete examples.
\end{example}


\subsection{Explicit constructions with nontrivial subrank lowerbound}\label{sec:explicit gate construction}
This section gives instantiations of the local codes so that we can prove a nonzero lower bound on $k_{\CZ}$ and $k_{\CCZ}$. First, we give the constructions of logical $\CZ$ on good qLDPC and qLTC.

\begin{theorem}\label{thm: CZ on good qLDPC}
    There exist $[\![n,\Theta(n),\Theta(n)]\!]$ quantum LDPC codes with transversal disjoint logical $\CZ$ gate where $k_{\CZ}=\Theta(n)$. 
\end{theorem}

\begin{theorem}\label{thm: CZ on almost good qLTC}
    There exist $[\![n,\Theta(n),\Theta(n/(\log\,n)^3)]\!]$ quantum locally testable codes with soundness $1/(\log\, n)^3$ and transversal disjoint logical $\CZ$ gate where $k_{\CZ}=\Theta(n)$.
\end{theorem}

\begin{proof}[Proof of Theorem \ref{thm: CZ on good qLDPC} and \ref{thm: CZ on almost good qLTC}]
    By Corollary \ref{Poincare dual pairing}, when $X$ is a sparse cell complex that admits simplicial approximation and $\F$ is a locally acyclic sheaf on it, then we have a dual pairing
     \begin{equation}
    P:H^i(X,\mathcal{F}^\perp)\times H^{t-i}(X,\mathcal{F})\longrightarrow\mathbb{F},
    \end{equation}
    given by $P([\alpha],[\beta])=\langle[\alpha]\smile[\beta],[X]\rangle$. In Section \ref{subsec:explicit}, we showed that a $t$-dimensional cubical complex admits simplicial approximation. When $t=2$, $i=1$, $X$ is the left-right Cayley square complex in \cite{DHLV2022} and $\F$ is the associated sheaf, then $\F$ is locally acyclic, and the bilinear map $P$ above gives the constant-depth logical $\CZ$ in this code. Note that $H^1(X,\F^\perp)\cong H^1(X,\F)$ by Poincar\'e duality, and $P$ is a dual pairing, therefore we can choose a basis $\{e_i\}$ for $H^1(X,\F^\perp)$ and a basis $\{f_j\}$ for $H^1(X,\F)$ such that $P(e_i, f_j)=\delta_{i,j}$, i.e., they are dual basis of each other. Therefore the number of logical $\CZ$ is exactly the number of logical qubits, i.e., $k_{\CZ}= k =\Theta(n)$, and the gates are disjoint. Using Lemma~\ref{lem:constant-depth-to-transversal}, we can convert the code such that the physical $\CZ$ circuit becomes transversal. This proves Theorem \ref{thm: CZ on good qLDPC}.

    When $t=4$, $i=2$, $X$ is the 4-dimensional cubical complex in \cite{Dinur2024sheaf}, and $\F$ is the associated sheaf, then $\F$ is locally acyclic. A similar argument proves Theorem \ref{thm: CZ on almost good qLTC}.
\end{proof}

We also provide an alternative proof, showing $k_{\CZ}>0$ instead of $k_{\CZ}=\Theta(n)$, serving as an evidence that one might use the technique of planting all-ones vectors into local code to prove $k_{\CCZ}$ in the future.

\begin{theorem}[$\CZ$ on square complex codes]\label{thm:cz}
    Consider the square complex sheaf code construction in~\cite{DHLV2022} with the following modification. First, we use the square complex $X$ construction in~\cite[Theorem 14]{Golowich2024_NLTSPlantedqLDPC} where $|X(2)|$ is an odd number. Second, we choose the sheaf $\mathcal{F}$ such that the classical codes $\im h_1^T$ and $\ker h_2=\im (h_2^\perp)^T$ contain the all-ones vector. Then there exists a cycle $x\in C_2(X,\mathcal{F}^\perp\!\!\otimes\mathcal{F})$ which induces a bilinear function
    \begin{equation}
    f_x:C^1(X,\mathcal{F}^\perp)\times C^1(X,\mathcal{F})\longrightarrow\mathbb{F},
    \end{equation}
    given by 
    \begin{equation}
    f_x(\alpha,\beta)=\langle\alpha\smile\beta,x\rangle.
    \end{equation}
    This bilinear function $f_x$ induces transversal $\CZ$ gate which is not logical identity. Furthermore, the code parameters are asymptotically good.
\end{theorem}
\begin{proof}
    We plant the all-ones vector in $\im h_1^T$ and $\ker h_2=\im (h_2^\perp)^T$. By Theorem~\ref{all-ones vector planted product expanding codes}, there exist such choices so that the local codes are two-way product-expanding and hence, the code parameters are asymptotically good according to~\cite{DHLV2022}.
    Notice that for each cell $\sigma$, 
    \begin{equation}
    \iota_\sigma\mathcal{F}_\sigma=\bigotimes_{j\notin\type (\sigma)}\im h^T_j.
    \end{equation}
    Hence $\1_\sigma\coloneqq \otimes_{j\notin\type(\sigma)}\1_j\in  \iota_\sigma\mathcal{F}_\sigma$ is the all-ones vector. Therefore, we can write down some codewords explicitly. For each set $S\subseteq[t]$, we define 
    $\1_S\coloneqq \sum_{\text{type}(\sigma)=S}\1_\sigma\cdot\,\sigma.$
Then $\1_S$ is a cocycle in $C^{|S|}(X,\mathcal{F})$
    Then $\1_{\{1\}}\in C^1(X,\mathcal{F}^\perp)$ and $\1_{\{2\}}\in C^1(X,\mathcal{F})$ are cocycles. The the cup product gives
    \begin{equation}
    \1_{\{1\}}\smile\1_{\{2\}}=\1_{\{1,2\}}\in C^2(X,\mathcal{F}^\perp\!\!\otimes\mathcal{F}) 
    \end{equation}
    Hence if we choose $|X(2)|$ to be an odd number, then
    \begin{equation}
    \langle \1_{\{1\}}\smile\1_{\{2\}},[X]\rangle=|X(2)|\neq0.
    \end{equation} 
    We only need to show that $\1_{\{1,2\}}\in C^2(X,\mathcal{F}^\perp\!\!\otimes\mathcal{F})$ is not a coboundary. Choose an arbitrary $y\in C^1(X,\mathcal{F}^\perp\!\!\otimes\mathcal{F})$, without loss of generality, we may assume $y=y(\sigma)\cdot \sigma$ and $y(\sigma)=v\otimes w$ for some $\sigma\in X(1)$ and $v\in C^\perp_\sigma$, $w\in C_{\sigma}$. Then we have
    \begin{equation}
    \langle \delta y,[X]\rangle=\sum_{\tau\in U_\sigma(2)} \mathcal{F}^\perp_{\sigma,\tau}v\otimes\mathcal{F}_{\sigma,\tau}w=\langle v, w\rangle=0,
    \end{equation}
    where the last equation is done by the fact that $C^\perp_\sigma$ is the dual code of $C_\sigma$, hence $v$ and $w$ are orthogonal. Therefore, $\1_{\{1,2\}}$ is a nontrivial cohomology class in $C^2(X,\mathcal{F}^\perp\!\!\otimes\mathcal{F})$, and by dual pairing (Proposition \ref{dual pairing between homology groups}), we can always choose a cyle $x\in C_2(X,\mathcal{F}^\perp\!\!\otimes\mathcal{F})$ such that $f_x$ gives a logical $\CZ$ gate which is not the logical identity, i.e. $k_{\CZ} >0$.
\end{proof}

In the next construction, we use a spectral-expanding 3-dimensional cubical complex~\cite{Dinur2024sheaf} and an appropriate choice of sheaves to obtain a quantum LDPC code (that is one-sided locally testable) with a nontrivial $\CCZ$ gate and a conjectured polynomial distance.

\begin{definition}
    Let $X$ be the 3-dimensional Abelian cubical complex from~\cite{Dinur2024sheaf}. Let $C_1,C_2,C_3$ be classical linear codes and $C_\mathrm{rep}$ the repetition code of the same length.  We define quantum codes $\mathcal Q_1, \mathcal Q_2, \mathcal Q_3$ by placing qubits on level 1, Z checks on level 2, and X checks on level 0 of the sheaf complexes $C^\bullet(X,\mathcal F), C^\bullet(X,\mathcal G), C^\bullet(X,\mathcal H)$, respectively, where $\mathcal{F}$ is the sheaf generated by $\{C_\sigma\}_{\sigma\in X(2)}$, where
\begin{equation}
    C_\sigma = 
    \begin{cases}
    C_1, & \type(\sigma)=\{2,3\},\\
    C_2, & \type(\sigma)=\{1,3\},\\
    C_{\text{rep}}, &\type(\sigma)=\{1,2\}.
    \end{cases}
\end{equation}
$\mathcal{G}$ is the sheaf generated by $\{C'_\sigma\}_{\sigma\in X(2)}$, where
\begin{equation}
    C'_\sigma = 
    \begin{cases}
    C_{\text{rep}}, & \type(\sigma)=\{2,3\},\\
    C_2^\perp, & \type(\sigma)=\{1,3\},\\
    C_3^\perp, &\type(\sigma)=\{1,2\}.
    \end{cases}
\end{equation}
$\mathcal{H}$ is the sheaf generated by $\{C''_\sigma\}_{\sigma\in X(2)}$, where
\begin{equation}
    C''_\sigma = 
    \begin{cases}
    C_1^\perp, & \type(\sigma)=\{2,3\},\\
    C_{\text{rep}}, & \type(\sigma)=\{1,3\},\\
    C_3, &\type(\sigma)=\{1,2\}.
    \end{cases}
\end{equation}
Furthermore, we require that the all-ones vector is contained in $C_1^\perp,C_2,C_3^\perp$.
\end{definition}

\begin{theorem}[$\CCZ$ on 3-dimensional cubical complex codes]\label{thm:ccz1}
    Consider the quantum code constructions $\mathcal Q_1, \mathcal Q_2, \mathcal Q_3$ above. We can find choices of $C_1,C_2,C_3$ such that the following holds on the resulting quantum codes. There exists a cycle $x\in C_{3}(X,\mathcal{F}\otimes\mathcal{G}\otimes\mathcal{H})$, such that trilinear form 
    \begin{equation}
    f_x:C^i(X,\mathcal{F})\times C^j(X,\mathcal{G})\times C^k(X,\mathcal{H})\longrightarrow\mathbb{F},
    \end{equation}
    defined by
    \begin{equation}f_x(\alpha,\beta,\gamma)=\langle\alpha\smile\beta\smile\gamma,x\rangle
    \end{equation}
    has a nontrivial cohomology subrank.
    In other words, there exists a choice of (co)cycles $\alpha\in C^i(X,\mathcal{F})$, $\beta\in C^j(X,\mathcal{H})$, $\gamma\in C^k(X,\mathcal{H})$, $x\in C_{i+j+k}(X,\mathcal{F}\otimes\mathcal{G}\otimes\mathcal{H})$ such that
    \begin{equation}
    f_x(\alpha,\beta,\gamma)\neq 0,
    \end{equation}
    Thus, the quantum codes admit a transversal logical $\CCZ$ gate induced by $f_x$ which is not the logical identity. 
    
    Moreover, the code family has inverse-polylog relative $X$-distance and is local testable with inverse-polylog soundness against $X$ errors. 
\end{theorem}
\begin{proof}

First we show the following claim, which allows us to choose the classical codes to `plant' the all-1's codewords appropriately, while also obtaining the property of product expansion. The planted all-1's vector will be useful in proving the nontriviality of the logical $\CCZ$ circuit. The product expansion will be used to establish the code parameters.

\begin{claim}\label{claim:plant-rep-code}
    We can choose $C_1,C_2,C_3\subseteq\mathbb{F}_q^n$ such that $(C_1,C_2,C_{\text{rep}})$, $(C_{\text{rep}}$, $C_2^\perp,C_3^\perp)$ and $(C_1^\perp$, $C_{\text{rep}}$, $C_3)$ are (one-way) $\rho$-product-expanding at the same time when all-ones vector contained in $\1\in C_1^\perp,C_2,C_3^\perp$.
\end{claim}
\begin{proof}[Proof of Claim~\ref{claim:plant-rep-code}]
    We choose $C_1,C_2,C_3$ in uniformly random distribution on the Grassmannians, with the constraints that the all-ones vector is contained in $C_1^\perp,C_2,C_3^\perp$. The codes $C_1$ and $C_2$ are independent, so are $C_1^\perp$ and $C_3$, $C_2^\perp$ and $C_3^\perp$. By the proof of Theorem \ref{all-ones vector planted product expanding codes}, we can choose $(C_1,C_2)$, $(C_1^\perp, C_3)$, $(C_2^\perp,C_3^\perp)$ to be two-way $\rho$-product-expanding at the same time, for some universal constant $\rho$. Next we use the following lemma.
    
\begin{lemma}[\protect{\cite[Lemma 7.2]{tan2025singleshotuniversalityquantumldpc}}]
    If $(C_1, C_2)$ is $\rho$-product-expanding, then the $(C_1, C_2, C_{\mathrm{rep}})$ is $\rho/3$-product-expanding.
\end{lemma}
This lemma proves that $(C_1,C_2,C_{\text{rep}})$, $(C_{\text{rep}},C_2^\perp,C_3^\perp)$ and $(C_1^\perp,C_{\text{rep}},C_3)$ are one-way $\rho/3$-product-expanding. 
\end{proof}

\paragraph{Logical $\CCZ$ gate:}
We show the existence of a nontrivial logical $\CCZ$ gate. Similarly to the proof of Theorem~\ref{thm:cz}, we observe that $\1_{\{1\}}\in C^1(X,\mathcal{F})$, $\1_{\{2\}}\in C^1(X,\mathcal{G})$ and $\1_{\{3\}}\in C^1(X,\mathcal{H})$ are cocycles, and
\begin{equation}
\1_{\{1\}}\smile\1_{\{2\}}\smile\1_{\{3\}}=\1_{\{1,2,3\}}\in C^3(X,\mathcal{F}\otimes\mathcal{G}\otimes\mathcal{H}).
\end{equation}
We claim that when $|X(3)|$ is an odd number, then $\1_{\{1,2,3\}}$ is not a coboundary. This is because, for $y\in C^2(X,\mathcal{F}\otimes\mathcal{G}\otimes\mathcal{H})$, without loss of generality, we assume $y=y(\sigma)\cdot\sigma$, where $\sigma\in X(2)$ and $y(\sigma)=u\otimes v\otimes w$, $u\in C_\sigma$, $v\in C'_\sigma$, $w\in C''_\sigma$. Then note that 
\begin{equation}
\langle\delta y,[X]\rangle=\sum_{\tau\in U_\sigma(3)}\mathcal{F}_{\sigma,\tau}u\otimes\mathcal{G}_{\sigma,\tau}v\otimes\mathcal{H}_{\sigma,\tau}w=\sum_{\tau\in U_\sigma}u_\tau\cdot v_\tau \cdot w_\tau=0
\end{equation}
This equals zero because by our construction, $u,v,w$ are always chosen from a local code, the dual of the local code and a repetition code. However,
\begin{equation}
\langle \1_{\{1,2,3\}},[X]\rangle=|X(3)|\neq0.
\end{equation}
Therefore, $\1_{\{1,2,3\}}$ is always not a coboundary. By Proposition~\ref{dual pairing between homology groups}, there always exists a cycle $x\in C_3(X,\mathcal{F}\otimes\mathcal{G}\otimes\mathcal{H})$ such that 
\begin{equation}
\langle\1_{\{1\}}\smile\1_{\{2\}}\smile\1_{\{3\}},x\rangle\neq0,
\end{equation}
which means that the code admits transversal $\CCZ$ inducing $k_{\CCZ} \geq 1$.

\paragraph{Code parameters:} Now we evaluate the code parameters. 

First, we deduce that the cocycles $\1_{\{1\}}\in C^1(X,\mathcal{F})$, $\1_{\{2\}}\in C^1(X,\mathcal{G})$ and $\1_{\{3\}}\in C^1(X,\mathcal{H})$ are not co-boundaries. This is because $\1_{\{1,2,3\}}$ would be a coboundary otherwise, a contradiction. So each quantum code encodes at least $1$ logical qubits. In the coboundary direction, the X distance and local testability soundness have the same scaling as~\cite{Dinur2024sheaf} because the the product-expansion in Claim~\ref{claim:plant-rep-code} suffices for the analysis in their Section 7. For the boundary direction, we believe that the $Z$ distance has a polynomial scaling, as in~\cite{zhu2025transversalnoncliffordgatesqldpc}, and leaves the proof for a future version of this work.
\end{proof}

Finally, recall that in Conjecture~\ref{main conjecture} Let $X$ be the $t$-dimensional cubical complex and $\F$ be the sheaf satisfying the requirement in~\cite{Dinur2024sheaf}. We conjecture that for $2\leq i,j,k\leq t$, $i+j+k\leq t$, there exist (co)homology classes $\alpha\in H^i(X,\F)$, $\beta\in H^j(X,\F)$, $\gamma\in H^k(X,\F)$ such that $\alpha\smile\beta\smile\gamma\neq 0\in H^{i+j+k}(X,\F^{\otimes 3})$. This will lead to the existence of transversal $\CCZ$ because, for example, if $\alpha\smile\beta\smile\gamma\neq 0$, then by dual pairing (Proposition~\ref{dual pairing between homology groups}), we may find a cocycle $\xi\in C_{t-i-j-k}(X,\F)$ such that $\langle [\xi],(\alpha\smile\beta)\frown\theta\rangle=1$. Therefore the following trilinear cohomological invariant
\[
f_{\xi}(-,-,-):C^i(X,\F)\times C^j(X,\F)\times C^k(X,\F)\longrightarrow\mathbb{F},
\]
defined by first taking the cup product of the three inputs, then pairing with $\xi$, gives the desired logical gate with $k_{\CCZ}\geq 1$. 

Although more advanced techniques are needed to theoretically determine the exact logical action, which we leave as a direction for future work, our framework provides the possibility of  settling this with numerical computation. 
Further numerical studies are also expected to help more precisely understand the logical actions, which is  valuable for practical use.


\printbibliography[heading=bibintoc,title=References]

@misc{lin2024transversalnoncliffordgatesquantum,
      title={Transversal non-Clifford gates for quantum LDPC codes on sheaves}, 
      author={Ting-Chun Lin},
      year={2024},
      eprint={2410.14631},
      archivePrefix={arXiv},
      primaryClass={quant-ph},
      url={https://arxiv.org/abs/2410.14631}, 
}

@misc{FreedmanHastings2021,
	title={Building manifolds from quantum codes}, 
	author={Michael Freedman and Matthew B. Hastings},
	year={2021},
	volume={34},
	ISSN={1420-8970},
	url={https://doi.org/10.1007/s00039-021-00567-3},
	DOI={10.1007/s00039-021-00567-3},
	journal={Geometric and Functional Analysis},
	publisher={Springer Nature},
}

@misc{Portnoy2023,
	title={Local Quantum Codes from Subdivided Manifolds}, 
	author={Elia Portnoy},
	year={2023},
	eprint={2303.06755},
	archivePrefix={arXiv},
	primaryClass={quant-ph},
	url={https://arxiv.org/abs/2303.06755}, 
}

@INPROCEEDINGS{Dinur2024sheaf,
	author={Dinur, Irit and Lin, Ting-Chun and Vidick, Thomas},
	booktitle={2024 IEEE 65th Annual Symposium on Foundations of Computer Science (FOCS)}, 
	title={Expansion of High-Dimensional Cubical Complexes: with Application to Quantum Locally Testable Codes}, 
	year={2024},
	volume={},
	number={},
	pages={379-385},
	doi={10.1109/FOCS61266.2024.00031}
}

@inproceedings{DHLV2022,
	author = {Dinur, Irit and Hsieh, Min-Hsiu and Lin, Ting-Chun and Vidick, Thomas},
	title = {Good Quantum LDPC Codes with Linear Time Decoders},
	year = {2023},
	isbn = {9781450399135},
	publisher = {Association for Computing Machinery},
	address = {New York, NY, USA},
	url = {https://doi.org/10.1145/3564246.3585101},
	doi = {10.1145/3564246.3585101},
	booktitle = {Proceedings of the 55th Annual ACM Symposium on Theory of Computing},
	pages = {905–918},
	numpages = {14},
	location = {Orlando, FL, USA},
	series = {STOC 2023}
}

@article{PK2021,
	title={Quantum LDPC Codes With Almost Linear Minimum Distance},
	volume={68},
	ISSN={1557-9654},
	url={http://dx.doi.org/10.1109/TIT.2021.3119384},
	DOI={10.1109/tit.2021.3119384},
	number={1},
	journal={IEEE Transactions on Information Theory},
	publisher={Institute of Electrical and Electronics Engineers (IEEE)},
	author={Panteleev, Pavel and Kalachev, Gleb},
	year={2022},
	month=jan, pages={213–229} 
}

@inproceedings{PK2022Good,
	author = {Panteleev, Pavel and Kalachev, Gleb},
	title = {Asymptotically good Quantum and locally testable classical LDPC codes},
	year = {2022},
	isbn = {9781450392648},
	publisher = {Association for Computing Machinery},
	address = {New York, NY, USA},
	url = {https://doi.org/10.1145/3519935.3520017},
	doi = {10.1145/3519935.3520017},
	booktitle = {Proceedings of the 54th Annual ACM SIGACT Symposium on Theory of Computing},
	pages = {375–388},
	numpages = {14},
	location = {Rome, Italy},
	series = {STOC 2022}
}

@misc{First2022,
	title={On Good $2$-Query Locally Testable Codes from Sheaves on High Dimensional Expanders}, 
	author={Uriya A. First and Tali Kaufman},
	year={2024},
	eprint={2208.01778},
	archivePrefix={arXiv},
	primaryClass={math.CO},
	url={https://arxiv.org/abs/2208.01778}, 
}

@inproceedings{First2024,
	author = {First, Uriya A. and Kaufman, Tali},
	title = {Cosystolic Expansion of Sheaves on Posets with Applications to Good 2-Query Locally Testable Codes and Lifted Codes},
	year = {2024},
	isbn = {9798400703836},
	publisher = {Association for Computing Machinery},
	address = {New York, NY, USA},
	url = {https://doi.org/10.1145/3618260.3649625},
	doi = {10.1145/3618260.3649625},
	pages = {1446–1457},
	numpages = {12},
	location = {Vancouver, BC, Canada},
	series = {STOC 2024}
}

@Article{10.21468/SciPostPhys.14.4.065,
	title={{Codimension-2 defects and higher symmetries in (3+1)D topological phases}},
	author={Maissam Barkeshli and Yu-An Chen and Sheng-Jie Huang and Ryohei Kobayashi and Nathanan Tantivasadakarn and Guanyu Zhu},
	journal={SciPost Phys.},
	volume={14},
	pages={065},
	year={2023},
	publisher={SciPost},
	doi={10.21468/SciPostPhys.14.4.065},
	url={https://scipost.org/10.21468/SciPostPhys.14.4.065},
}

@inproceedings{Golowich_Lin2024,
	author = {Golowich, Louis and Lin, Ting-Chun},
	title = {Quantum LDPC Codes with Transversal Non-Clifford Gates via Products of Algebraic Codes},
	year = {2025},
	isbn = {9798400715105},
	publisher = {Association for Computing Machinery},
	address = {New York, NY, USA},
	url = {https://doi.org/10.1145/3717823.3718139},
	doi = {10.1145/3717823.3718139},
	pages = {689–696},
	numpages = {8},
	location = {Prague, Czechia},
	series = {STOC '25}
}

@article{Chen_2023,
	title={Higher cup products on hypercubic lattices: Application to lattice models of topological phases},
	volume={64},
	ISSN={1089-7658},
	url={http://dx.doi.org/10.1063/5.0095189},
	DOI={10.1063/5.0095189},
	number={9},
	journal={Journal of Mathematical Physics},
	publisher={AIP Publishing},
	author={Chen, Yu-An and Tata, Sri},
	year={2023},
	month=sep
}

@article{Wang_2024,
	title={Efficient fault-tolerant implementations of non-Clifford gates with reconfigurable atom arrays},
	volume={10},
	ISSN={2056-6387},
	url={http://dx.doi.org/10.1038/s41534-024-00945-3},
	DOI={10.1038/s41534-024-00945-3},
	number={1},
	journal={npj Quantum Information},
	publisher={Springer Science and Business Media LLC},
	author={Wang, Yifei and Wang, Yixu and Chen, Yu-An and Zhang, Wenjun and Zhang, Tao and Hu, Jiazhong and Chen, Wenlan and Gu, Yingfei and Liu, Zi-Wen},
	year={2024},
	month=dec 
}

@book{Bott:1982xhp,
    author = "Bott, Raoul and Tu, Loring W.",
    title = "{Differential Forms in Algebraic Topology}",
    doi = "10.1007/978-1-4757-3951-0",
    isbn = "978-0-387-90613-3, 978-1-4419-2815-3, 978-1-4757-3951-0",
    publisher = "Springer",
    year = "1982"
}

@misc{Panteleev2024,
      title={Maximally Extendable Sheaf Codes}, 
      author={Pavel Panteleev and Gleb Kalachev},
      year={2024},
      eprint={2403.03651},
      archivePrefix={arXiv},
      primaryClass={cs.IT},
      url={https://arxiv.org/abs/2403.03651}, 
}

@misc{kalachev2025maximallyextendableproductcodes,
      title={Maximally Extendable Product Codes are Good Coboundary Expanders}, 
      author={Gleb Kalachev and Pavel Panteleev},
      year={2025},
      eprint={2501.01411},
      archivePrefix={arXiv},
      primaryClass={cs.IT},
      url={https://arxiv.org/abs/2501.01411}, 
}

@misc{tan2025singleshotuniversalityquantumldpc,
      title={Single-Shot Universality in Quantum LDPC Codes via Code-Switching}, 
      author={Shi Jie Samuel Tan and Yifan Hong and Ting-Chun Lin and Michael J. Gullans and Min-Hsiu Hsieh},
      year={2025},
      eprint={2510.08552},
      archivePrefix={arXiv},
      primaryClass={quant-ph},
      url={https://arxiv.org/abs/2510.08552}, 
}

@book{Hatcheralgtop,
  address = {Cambridge},
  author = {Hatcher, Allen},
  isbn = {0-521-79160-X; 0-521-79540-0},
  mrclass = {55-01 (55-00)},
  mrnumber = {1867354 (2002k:55001)},
  mrreviewer = {Donald W. Kahn},
  pages = {xii+544},
  publisher = {Cambridge University Press},
  title = {Algebraic topology},
  username = {mwpb479},
  year = 2002
}

@article{breuckmann2022fold,
  title={Fold-Transversal Clifford Gates for Quantum Codes},
  author={Breuckmann, Nikolas P and Burton, Simon},
  journal={arXiv preprint arXiv:2202.06647},
  year={2022},
  url={https://arxiv.org/abs/2202.06647}
}

@inproceedings{leverrier2022quantum,
  title={Quantum tanner codes},
  author={Leverrier, Anthony and Z{\'e}mor, Gilles},
  booktitle={2022 IEEE 63rd Annual Symposium on Foundations of Computer Science (FOCS)},
  pages={872--883},
  year={2022},
  organization={IEEE},
  url={https://arxiv.org/abs/2202.13641}
}

@inproceedings{eldar2017local,
  title={Local Hamiltonians whose ground states are hard to approximate},
  author={Eldar, Lior and Harrow, Aram W},
  booktitle={2017 IEEE 58th annual symposium on foundations of computer science (FOCS)},
  pages={427--438},
  year={2017},
  organization={IEEE}
}

@article{gottesman2013fault,
  title={Fault-tolerant quantum computation with constant overhead},
  author={Gottesman, Daniel},
  journal={arXiv preprint arXiv:1310.2984},
  year={2013}
}

@article{yamasaki2024time,
  title={Time-efficient constant-space-overhead fault-tolerant quantum computation},
  author={Yamasaki, Hayata and Koashi, Masato},
  journal={Nature Physics},
  pages={1--7},
  year={2024},
  publisher={Nature Publishing Group UK London}
}

@article{bravyi2012magic,
  title={Magic-state distillation with low overhead},
  author={Bravyi, Sergey and Haah, Jeongwan},
  journal={Physical Review A—Atomic, Molecular, and Optical Physics},
  volume={86},
  number={5},
  pages={052329},
  year={2012},
  publisher={APS}
}

@article{Paetznick_2013,
	title = {Universal Fault-Tolerant Quantum Computation with Only Transversal Gates and Error Correction},
	author = {Paetznick, Adam and Reichardt, Ben W.},
	journal = {Phys. Rev. Lett.},
	volume = {111},
	issue = {9},
	pages = {090505},
	numpages = {5},
	year = {2013},
	month = {Aug},
	publisher = {American Physical Society},
	doi = {10.1103/PhysRevLett.111.090505},
	url = {https://link.aps.org/doi/10.1103/PhysRevLett.111.090505}
}

@article{zhu2023non,
  title={Non-Clifford and parallelizable fault-tolerant logical gates on constant and almost-constant rate homological quantum LDPC codes via higher symmetries},
  author={Zhu, Guanyu and Sikander, Shehryar and Portnoy, Elia and Cross, Andrew W and Brown, Benjamin J},
  journal={arXiv preprint arXiv:2310.16982},
  year={2023}
}

@inproceedings{nguyen2025quantum,
  title={Quantum fault tolerance with constant-space and logarithmic-time overheads},
  author={Nguyen, Quynh T and Pattison, Christopher A},
  booktitle={Proceedings of the 57th Annual ACM Symposium on Theory of Computing},
  pages={730--737},
  year={2025}
}

@inproceedings{nguyen2025good,
  title={Good binary quantum codes with transversal CCZ gate},
  author={Nguyen, Quynh T},
  booktitle={Proceedings of the 57th Annual ACM Symposium on Theory of Computing},
  pages={697--706},
  year={2025}
}

@inproceedings{anshu2024circuit,
  title={Circuit-to-Hamiltonian from tensor networks and fault tolerance},
  author={Anshu, Anurag and Breuckmann, Nikolas P and Nguyen, Quynh T},
  booktitle={Proceedings of the 56th Annual ACM Symposium on Theory of Computing},
  pages={585--595},
  year={2024}
}

@misc{curry2014sheavescosheavesapplications,
      title={Sheaves, Cosheaves and Applications}, 
      author={Justin Curry},
      year={2014},
      eprint={1303.3255},
      archivePrefix={arXiv},
      primaryClass={math.AT},
      url={https://arxiv.org/abs/1303.3255}, 
}

@article{Breuckmann:2024cupandgate,
    author = "Breuckmann, Nikolas P. and Davydova, Margarita and Eberhardt, Jens N. and Tantivasadakarn, Nathanan",
    title = "{Cups and Gates I: Cohomology invariants and logical quantum operations}",
    eprint = "2410.16250",
    archivePrefix = "arXiv",
    primaryClass = "quant-ph",
    month = "10",
    year = "2024"
}

@misc{zhu2025transversalnoncliffordgatesqldpc,
      title={Transversal non-Clifford gates on qLDPC codes breaking the $\sqrt{N}$ distance barrier and quantum-inspired geometry with $\mathbb{Z}_2$ systolic freedom}, 
      author={Guanyu Zhu},
      year={2025},
      eprint={2507.15056},
      archivePrefix={arXiv},
      primaryClass={quant-ph},
      url={https://arxiv.org/abs/2507.15056}, 
}

@misc{stacks-project,
  author       = {The {Stacks project authors}},
  title        = {The Stacks project},
  howpublished = {\url{https://stacks.math.columbia.edu}},
  year         = {2025},
}

@article{KrishnaPoulin2021_FTGatesHGP,
  author  = {Krishna, Anirudh and Poulin, David},
  title   = {Fault-tolerant gates on hypergraph product codes},
  journal = {Physical Review X},
  year    = {2021},
  volume  = {11},
  number  = {1},
  pages   = {011023},
  doi     = {10.1103/PhysRevX.11.011023},
  url     = {https://doi.org/10.1103/PhysRevX.11.011023}
}

@article{ScrubyPesahWebster2024_QuantumRainbowCodes,
  author  = {Scruby, Thomas R. and Pesah, Arthur and Webster, Mark},
  title   = {Quantum Rainbow Codes},
  journal = {arXiv preprint arXiv:2408.13130},
  year    = {2024},
  url     = {https://arxiv.org/abs/2408.13130}
}

@article{BurtonBrowne2022_HGPTransversalLimits,
  author  = {Burton, Simon and Browne, Dan},
  title   = {Limitations on Transversal Gates for Hypergraph Product Codes},
  journal = {IEEE Transactions on Information Theory},
  year    = {2022},
  volume  = {68},
  number  = {3},
  pages   = {1772--1781},
  doi     = {10.1109/TIT.2021.3131043},
  url     = {https://doi.org/10.1109/TIT.2021.3131043}
}

@article{FuZhengLiLiu2025_ProductCodesNoGo,
  author  = {Fu, Esther Xiaozhen and Zheng, Han and Li, Zimu and Liu, Zi-Wen},
  title   = {No-go Theorems for Logical Gates on Product Quantum Codes},
  journal = {arXiv preprint arXiv:2507.16797},
  year    = {2025},
  url     = {https://arxiv.org/abs/2507.16797}
}

@article{QuintavalleWebsterVasmer2023_PartitionHGPClifford,
  author  = {Quintavalle, Alessandro Ottaviani and Webster, Peter and Vasmer, Michael},
  title   = {Partitioning qubits in hypergraph product codes to implement logical gates},
  journal = {Quantum},
  year    = {2023},
  volume  = {7},
  pages   = {1153},
  doi     = {10.22331/q-2023-10-24-1153},
  url     = {https://doi.org/10.22331/q-2023-10-24-1153}
}

@article{PatraBarg2025_TargetedCliffordHGP,
  author  = {Patra, Adway and Barg, Alexander},
  title   = {Targeted Clifford logical gates for hypergraph product codes},
  journaltitle = {Quantum},
  year    = {2025},
  volume  = {9},
  pages   = {1842},
  doi     = {10.22331/q-2025-08-29-1842},
  eprint  = {2411.17050},
  archivePrefix = {arXiv},
  primaryClass  = {quant-ph},
  url     = {https://doi.org/10.22331/q-2025-08-29-1842}
}

@article{AharonovEldar2015_QLTC,
  author  = {Aharonov, Dorit and Eldar, Lior},
  title   = {Quantum Locally Testable Codes},
  journal = {SIAM Journal on Computing},
  year    = {2015},
  volume  = {44},
  number  = {5},
  pages   = {1230--1262},
  doi     = {10.1137/140975498},
  url     = {https://doi.org/10.1137/140975498}
}

@article{AharonovAradVidick2013_qPCPSurvey,
  author  = {Aharonov, Dorit and Arad, Itai and Vidick, Thomas},
  title   = {The Quantum {PCP} Conjecture},
  journal = {SIGACT News},
  year    = {2013},
  volume  = {44},
  number  = {2},
  pages   = {47--79},
  doi     = {10.1145/2491533.2491549},
  url     = {https://doi.org/10.1145/2491533.2491549}
}

@inproceedings{anshu2023nlts,
  title={NLTS Hamiltonians from good quantum codes},
  author={Anshu, Anurag and Breuckmann, Nikolas P and Nirkhe, Chinmay},
  booktitle={Proceedings of the 55th Annual ACM Symposium on Theory of Computing},
  pages={1090--1096},
  year={2023}
}

@inproceedings{Golowich2024_NLTSPlantedqLDPC,
  author    = {Golowich, Louis and Kaufman, Tali},
  title     = {{NLTS} Hamiltonians and Strongly-Explicit {SoS} Lower Bounds from Planted Quantum {LDPC} Codes},
  booktitle = {15th Innovations in Theoretical Computer Science Conference (ITCS 2024)},
  year      = {2024},
  series    = {LIPIcs},
  volume    = {287},
  pages     = {54:1--54:34},
  doi       = {10.4230/LIPIcs.ITCS.2024.54},
  url       = {https://doi.org/10.4230/LIPIcs.ITCS.2024.54}
}

@inproceedings{Vidick2023_QCodesLocalTestabilityIP,
  author    = {Vidick, Thomas},
  title     = {Quantum Codes, Local Testability and Interactive Proofs},
  booktitle = {50th International Colloquium on Automata, Languages, and Programming (ICALP 2023)},
  year      = {2023},
  series    = {LIPIcs},
  volume    = {261},
  pages     = {4:1--4:20},
  doi       = {10.4230/LIPIcs.ICALP.2023.4},
  url       = {https://doi.org/10.4230/LIPIcs.ICALP.2023.4}
}

@article{CrossHeNatarajanSzegedyZhu2024_QLTCConstantSoundness,
  author       = {Cross, Andrew and He, Zhiyang and Natarajan, Anand and Szegedy, Mario and Zhu, Guanyu},
  title        = {Quantum Locally Testable Code with Constant Soundness},
  journaltitle = {Quantum},
  year         = {2024},
  volume       = {8},
  pages        = {1501},
  doi          = {10.22331/q-2024-10-18-1501},
  url          = {https://doi.org/10.22331/q-2024-10-18-1501},
  eprint       = {2209.11405},
  archivePrefix = {arXiv},
  primaryClass  = {cs.IT}
}

@inproceedings{FawziGrospellierLeverrier2018_ConstOverheadFT,
  author    = {Fawzi, Omar and Grospellier, Antoine and Leverrier, Anthony},
  title     = {Constant Overhead Quantum Fault-Tolerance with Quantum Expander Codes},
  booktitle = {2018 IEEE 59th Annual Symposium on Foundations of Computer Science (FOCS)},
  year      = {2018},
  pages     = {743--754},
  doi       = {10.1109/FOCS.2018.00075},
  url       = {https://doi.org/10.1109/FOCS.2018.00075},
  eprint    = {1808.03821},
  archivePrefix = {arXiv},
  primaryClass  = {quant-ph}
}

@article{CohenEtAl2022_LowOverheadFTqLDPC,
  author  = {Cohen, Lawrence Z. and Kim, Isaac H. and Bartlett, Stephen D. and Brown, Benjamin J.},
  title   = {Low-Overhead Fault-Tolerant Quantum Computing Using Long-Range Connectivity},
  journal = {Science Advances},
  year    = {2022},
  volume  = {8},
  number  = {20},
  pages   = {eabn1717},
  doi     = {10.1126/sciadv.abn1717},
  url     = {https://doi.org/10.1126/sciadv.abn1717},
  eprint  = {2110.10794},
  archivePrefix = {arXiv},
  primaryClass  = {quant-ph}
}

@article{GuEtAl2024_SingleShotGoodqLDPC,
  author  = {Gu, Shouzhen and Tang, Eugene and Caha, Libor and Choe, Shin Ho and He, Zhiyang and Kubica, Aleksander},
  title   = {Single-Shot Decoding of Good Quantum {LDPC} Codes},
  journal = {Communications in Mathematical Physics},
  year    = {2024},
  volume  = {405},
  number  = {2},
  pages   = {57},
  doi     = {10.1007/s00220-024-04951-6},
  url     = {https://doi.org/10.1007/s00220-024-04951-6},
  eprint  = {2306.12470},
  archivePrefix = {arXiv},
  primaryClass  = {quant-ph}
}

@article{TamiyaKoashiYamasaki2024_PolylogConstSpaceqLDPC,
  author        = {Tamiya, Shiro and Koashi, Masato and Yamasaki, Hayata},
  title         = {Polylog-time- and constant-space-overhead fault-tolerant quantum computation with quantum low-density parity-check codes},
  journal       = {arXiv preprint arXiv:2411.03683},
  year          = {2024},
  eprint        = {2411.03683},
  archivePrefix = {arXiv},
  primaryClass  = {quant-ph},
  doi           = {10.48550/arXiv.2411.03683},
  url           = {https://arxiv.org/abs/2411.03683}
}

@article{LeverrierZemor2023_DecodingQTanner,
  author  = {Leverrier, Anthony and Z{\'e}mor, Gilles},
  title   = {Decoding Quantum Tanner Codes},
  journal = {IEEE Transactions on Information Theory},
  year    = {2023},
  volume  = {69},
  number  = {8},
  pages   = {5041--5051},
  doi     = {10.1109/TIT.2023.3267945},
  url     = {https://doi.org/10.1109/TIT.2023.3267945}
}

@misc{LLL2026nontrivial,
      title={Transversal non-Clifford gates on almost-good quantum {LDPC} and quantum locally testable codes}, 
      author={Yiming Li and Zimu Li and Zi-Wen Liu},
      year={2026},
      eprint={2604.01874},
      archivePrefix={arXiv},
      primaryClass={quant-ph},
      url={https://arxiv.org/abs/2604.01874}, 
}

\appendix

\section{Two-way product expansion of planted random codes}
In this appendix, we verify the existence of a set of product-expanding codes containing the all-ones vector. \label{appendix plant all-ones vector}

\begin{theorem}\label{all-ones vector planted product expanding codes}
    For each collection of intervals $I_1,\dots,I_t\subseteq(0,1)$, there exists some $\rho>0$ such that for $n\in \mathbb{N}$ there exist codes $\mathcal{C}_1,\dots, \mathcal{C}_t\subseteq\mathbb{F}^n_{q}$ such that each $\mathcal{C}_i^\perp$ contains the all-ones vector, $\frac{1}{n}\dim \mathcal{C}_i\in I_i$, and
    \begin{equation}
    \rho(\mathcal{C}_1,\dots,\mathcal{C}_t)>\rho,\ \ \ \ \  \rho(\mathcal{C}_1^\perp,\dots,\mathcal{C}_t^\perp)>\rho,
    \end{equation}
    where $q$ is a power of 2 sufficiently large such that
    \begin{equation}
    \left(1-n\frac{q^{n-\min_{i\in[t]}\dim\mathcal{C}_i}-1}{q^n-1}\right)^t-n^t2^{n^t+2}q^{-1}>0
    \end{equation}
\end{theorem}
We first introduce two concepts from \cite{kalachev2025maximallyextendableproductcodes}.

\begin{definition}[\cite{kalachev2025maximallyextendableproductcodes}]
    For the product code $\bigotimes_{i\in [t]}\mathcal{C}_i^\perp$, we say that a set $M\subseteq [n]^t$ is extendable in the product code if for every local codeword $c_M\in \mathbb{F}_q^n$ satisfying all local checks $z\in \mathcal{C}_\ell$, $\ell\subseteq M$ can be extended to a global codeword $c\in \bigotimes_{i\in [t]}\mathcal{C}_i^\perp$, $c|_M=c_M$. Here for each line $\ell\in \mathcal{L}(n,t)$ we define
    \begin{equation}
    \mathcal{C}_\ell\coloneqq \{c\in \mathbb{F}^{[n]^t}_q:\supp c \subseteq \ell, c|_\ell\in \mathcal{C}_i\}
    \end{equation}
\end{definition}
\begin{definition}[\cite{kalachev2025maximallyextendableproductcodes}]
    We say that a product code 
$\mathcal{C} = \bigotimes_{i\in[D]} \mathcal{C}_i \subseteq \mathbb{F}_{q}^{[n]^t}$
is maximally extendable if for every other product code $\mathcal{C}' = \bigotimes_{i\in[D]} \mathcal{C}'_i \subseteq \mathbb{F}_{q}^{[n]^t}$
with $\dim \mathcal{C}_i = \dim \mathcal{C}'_i,\, i\in[t]$,
when $M$ is extendable in $\mathcal{C}'$ it is also extendable in $\mathcal{C}$.
\end{definition}

Then we present the following two lemmas.

\begin{lemma}\label{finding all-nonzero vector}
    For a code $\mathcal{C}\subseteq\mathbb{F}_q^n$ picked uniformly from $\mathrm{Gr}_q(n, k)$, the possibility of containing all-component-nonzero vector is at least $1-n(q^{n-k}-1)/(q^n-1)$.
\end{lemma}
\begin{proof}
    For each $j\in[n]$, let $P_j\coloneqq \{x\in\mathcal{C}:x_j=0\}$ be the linear subspace with $j$-th coordinate being zero. Then 
    \begin{equation}
    |P_j|\leq |\mathrm{Gr}_q(n-1,k)|=\binom{n-1}{k}_q.
    \end{equation}
    By union bound, we have
    \begin{equation}
    \mathbb{P}[\mathcal{C}\in\bigcup_{j\in[t]}P_j]\leq n\cdot\frac{\binom{n-1}{k}_q}{\binom{n}{k}_q}=n\cdot\frac{q^{n-k}-1}{q^n-1},
    \end{equation}
    which proves the lemma.
\end{proof}
\begin{lemma}\label{diagonal matrices preserve maximal extendability}
    Suppose $D_i$ is a invertible diagonal matrix, $i\in[t]$. If $\bigotimes_{i\in [t]}\mathcal{C}_i^\perp$ is maximally extendable, so do $\bigotimes_{i\in [t]}D_i\mathcal{C}_i^\perp$.
\end{lemma}
\begin{proof}
    Let $D_i=\text{diag}(\lambda_{i,1},\cdots,\lambda_{i,t})$. We define a linear map $T:\mathbb{F}_q^{[n]^t}\longrightarrow \mathbb{F}_q^{[n]^t}$ by
    \begin{equation}
    (Tx)_{i_1,\cdots,i_t}\coloneqq \left(\prod_{j=1}^t\lambda_{j,i_{j}}\right) \cdot x_{i_1,\cdots,i_t}
    \end{equation}
    It is easy to see that $T$ gives an isomorphism between $\bigotimes_{i\in [t]}\mathcal{C}_i^\perp$ and $\bigotimes_{i\in [t]}D_i\mathcal{C}_i^\perp$. Suppose $\bigotimes_{i\in [t]}\mathcal{C}_i^\perp$ is maximally extendable, then for every other product code $\mathcal{C}' = \bigotimes_{i\in[D]} \mathcal{C}'_i \subseteq \mathbb{F}_{q}^{[n]^t}$
with $\dim D_i\mathcal{C}_i^\perp = \dim \mathcal{C}'_i,\, i\in[t]$, and $M\subseteq[n]^t$ is extendable in $\mathcal{C}' = \bigotimes_{i\in[D]} \mathcal{C}'_i \subseteq \mathbb{F}_{q}^{[n]^t}$, the $M$ must also be extendable in $\bigotimes_{i\in [t]}\mathcal{C}_i^\perp$ since $\dim \mathcal{C}_i^\perp=\dim D_i\mathcal{C}_i^\perp$. Note that for each local codeword $c_M\in\mathbb{F}_q^{M}$ satisfying all local check $z\in (D\mathcal{C}^\perp)^\perp_{\ell}$, $\ell\in M$, i.e. $\langle z, c_m\rangle=0$. Then we notice that $(D_i\mathcal{C}_i^\perp)^\perp=D_i^{-1}\mathcal{C}_i$, and all local check $\tilde{z}\in (C^\perp)^\perp_\ell$ are in bijection with $z$ via $\tilde{z}=Tz$, hence $\langle \tilde{z},T^{-1}c_M\rangle=\langle Tz, T^{-1}c_M\rangle=\langle z,c_M\rangle=0$, i.e. $Tc_M$ satisfies all local checks of $\bigotimes_{i\in [t]}\mathcal{C}_i^\perp$, hence can be extended to a global codeword $\tilde{c}$ such that $\tilde{c}|_M=T^{-1}c_M$. Note that $T\tilde{c}$ is a global codeword of $\bigotimes_{i\in [t]}D_i\mathcal{C}_i^\perp$ and $(T\tilde{c})|_M=c_M$, hence $M$ is extendable in $\bigotimes_{i\in [t]}D_i\mathcal{C}_i^\perp$, which proves that $\bigotimes_{i\in [t]}D_i\mathcal{C}_i^\perp$ is maximally extendable.
\end{proof} 

Finally, we borrow two more important results from \cite{kalachev2025maximallyextendableproductcodes}.

\begin{lemma}[\cite{kalachev2025maximallyextendableproductcodes}]\label{maximal extendability implies product expansion}
For all $t\in\mathbb{N}$ there is a function $\mu_t:(0,1)^t\rightarrow(0,1)$ such that for every tuple of rates $(r_1,\dots,r_t)\in(0,1)^t$ and maximally extendable code $\mathcal{C}_1^\perp\otimes\cdots\otimes \mathcal{C}_t^\perp$ such that $\mathcal{C}_i\in \mathrm{Gr}_{2^m}(n,k_i)$ and $k_i\leq r_in$, we have $\rho(\mathcal{C}_1,\cdots,\mathcal{C}_t)\geq\mu_t(r_1,\cdots,r_t)$.
\end{lemma}

\begin{theorem}[\cite{kalachev2025maximallyextendableproductcodes}]\label{random two-way maximally extendable codes}
For a collection $(\mathcal{C}_1,\ldots,\mathcal{C}_t)$ picked uniformly at random from 
$\mathrm{Gr}_{2^m}(n,k_1)\times\cdots\times\mathrm{Gr}_{2^m}(n,k_t)$,
the code $\mathcal{C}_1\otimes\cdots\otimes\mathcal{C}_t$ and $\mathcal{C}_1^\perp\otimes\cdots\otimes\mathcal{C}_t^\perp$ are maximally extendable at the same time with probability at least $1 - n^t 2^{n^t - m + 2}$.
\end{theorem}

\begin{proof}[Proof of Theorem \ref{all-ones vector planted product expanding codes}]
By Theorem \ref{random two-way maximally extendable codes}, we may find $\mathrm{Gr}_{2^m}(n,k_1)\times\cdots\times\mathrm{Gr}_{2^m}(n,k_t)$,
the code $\mathcal{C}_1\otimes\cdots\otimes\mathcal{C}_t$ and $\mathcal{C}_1^\perp\otimes\cdots\otimes\mathcal{C}_t^\perp$ satisfying desired rates and being maximally extendable at the same time. By Lemma \ref{finding all-nonzero vector}, we may further choose, for each $\mathcal{C}_i^\perp$, there is an all-component nonzero vector $(c_{i,1},\cdots,c_{i,t})^T\in \mathcal{C}_i$. Then we define $D_i\coloneqq \text{diag}(c^{-1}_{i,1},\cdots,c^{-1}_{i,t})$, this is invertible. Note that each $D_i\mathcal{C}_i^\perp$ contains all-ones vector. By Lemma \ref{diagonal matrices preserve maximal extendability}, $\bigotimes_{i\in[t]}D_i\mathcal{C}_i^\perp$ is still maximally extendable. Also, since $(D_i\mathcal{C}_i^\perp)^\perp=D_i^{-1}\mathcal{C}_i$, $\bigotimes_{i\in [t]}(D_i\mathcal{C}_i^\perp)^\perp$ is still maximally extendable. Therefore, by Lemma \ref{maximal extendability implies product expansion}, $\bigotimes_{i\in [t]}(D_i\mathcal{C}_i^\perp)$ and $\bigotimes_{i\in [t]}(D_i\mathcal{C}_i^\perp)^\perp$ are the desired all-ones vector planted two-way product expanding codes.
\end{proof}

\end{document}